\documentclass[regno]{amsart}
\usepackage[british]{babel}
\usepackage[OT2,T2A,T1]{fontenc}
\usepackage[utf8]{inputenc}		
\usepackage{eulervm}	
\usepackage[ bb=boondox]{mathalfa}

\usepackage{titlesec}
\titleformat{\section}
	{\centering\Large\scshape\bfseries}{\thesection.}{1em}{} 
\titleformat{\subsection}
 	{\normalfont\bfseries\scshape}{\thesubsection.}{.5em}{}
\titleformat*{\paragraph}{\bfseries}

\usepackage{titletoc}

    \titlecontents{section}
    [0em] %
    {\medskip}
    {\thecontentslabel\quad}
    {}
    {\hfill\contentspage}

    \newenvironment{acknowledgements}{%
  
  \begin{abstract}
}{%
  \end{abstract}
}

    \newenvironment{datastatement}{%
  
  \begin{abstract}
}{%
  \end{abstract}
}

    \newenvironment{conflict}{%
  
  \begin{abstract}
}{%
  \end{abstract}
}

\usepackage{amsmath,amsthm,amssymb,dsfont,extarrows,amscd,amsfonts,mathtools,bm,physics}
\usepackage{thmtools}
\usepackage[v1]{subfiles}
\usepackage{xr-hyper} 
\usepackage{hyperref}
\usepackage{mathrsfs}
\usepackage[mathscr]{euscript}
\usepackage{tensor}
\usepackage[most]{tcolorbox}
\usepackage{stmaryrd}
\usepackage[capitalize]{cleveref}
\usepackage{MnSymbol}
\usepackage{xcolor}

\usepackage{tikz-cd}
\usetikzlibrary{graphs,decorations.pathmorphing,decorations.markings}
\tikzset{
    labl/.style={anchor=south, rotate=90, inner sep=.5mm}
}

\usepackage{todonotes}
\usepackage{geometry}
\usepackage[alphabetic, initials]{amsrefs}

\DeclareFontFamily{U}{MnSymbolC}{}
\DeclareSymbolFont{MnSyC}{U}{MnSymbolC}{m}{n}
\DeclareFontShape{U}{MnSymbolC}{m}{n}{
	<-6>  MnSymbolC5
	<6-7>  MnSymbolC6
	<7-8>  MnSymbolC7
	<8-9>  MnSymbolC8
	<9-10> MnSymbolC9
	<10-12> MnSymbolC10
	<12->   MnSymbolC12}{}
\DeclareMathSymbol{\intprod}{\mathbin}{MnSyC}{'270}

\theoremstyle{plain} 
\newtheorem{thm}{Theorem}[section]
\newtheorem{cor}{Corollary}[section]
\newtheorem{lem}{Lemma}[section]
\newtheorem{prop}{Proposition}[section]

\theoremstyle{definition}
\newtheorem{defn}{Definition}[section]
\newtheorem{ex}{Example}[section]
\newtheorem{conj}{Conjecture}[section]

\theoremstyle{remark} 
\newtheorem{oss}{Remark} [section]

\numberwithin{equation}{section}

\crefalias{thm}{theorem}
\crefalias{cor}{corollary}
\crefalias{lem}{lemma}
\crefalias{prop}{proposition}
\crefalias{ex}{example}
\crefalias{defn}{definition}
\crefalias{conj}{conjecture}
\crefalias{oss}{remark}

\newcommand{\CC}{\mathbb{C}}
\newcommand{\ZZ}{\mathbb{Z}}

\DeclareMathOperator{\X}{\mathfrak{X}}
\DeclareMathOperator{\Lie}{\mathscr{L}}
\DeclareMathOperator{\id}{id}
\DeclareMathOperator{\pr}{pr}

\newcommand{\Hol}{\mathscr{O}}
\newcommand{\RS}{\mathbb{P}^{1}}
\newcommand{\Hw}{\mathscr{H}}
\newcommand{\UC}{\mathscr{C}}
\newcommand{\Fr}{\mathscr{F}}

\newcommand{\Frst}{\overset{\star}{\Fr}}
\newcommand{\Li}{\mathrm{Li}}
\newcommand{\ELi}{\Lambda \iota}
\newcommand{\EE}{\mathscr{E}}
\newcommand{\End}{\mathrm{End}}
\newcommand{\ee}{\mathbb{e}}
\newcommand{\EEE}{\mathbb{E}}

\newcommand{\inp}[2]{\langle\, \bm{#1}\,,\, \bm{#2} \,\rangle}

\newcommand{\extprod}[2]{#1\,\,\widetilde{\bullet}\,\,#2}
\newcommand{\extdprod}[2]{#1\,\,\widetilde{\ast}\,\,#2}

\newcommand{\ideal}[1]{\mathopen{<}\,#1\,\mathclose{>}}

\title{Dubrovin duality for open Hurwitz flat F-manifolds}
\author[A.~Proserpio]{Alessandro Proserpio}
\address[A.~Proserpio]{
School of Mathematics and Statistics,
University of Glasgow, Glasgow, G12 8QQ, United Kingdom.}
\email{a.proserpio.1 [at] research.gla.ac.uk}

\author[I. A. B.~Strachan]{Ian A. B. Strachan}
\address[I. A. B.~Strachan]{
School of Mathematics and Statistics,
University of Glasgow, Glasgow, G12 8QQ, United Kingdom.}
\email{ian.strachan [at] glasgow.ac.uk}

\begin{document}
\renewcommand{\hbar}{\hslash}
\begin{abstract}
 We prove that the Dubrovin dual of a Hurwitz Frobenius manifold extends naturally to an F-manifold with compatible flat connection on the universal curve, in the sense of the open WDVV equations. A similar result had been proven for the Frobenius manifold itself by G. Feitosa de Almeida. This equips the universal curve with two F-manifolds with compatible flat structure, and we study their duality. We show that they combine into a bi-flat F-manifold. Conditions on open WDVV solutions imposed in previous work are retrieved in this setting, thus providing them with a geometrical meaning. Finally, explicit examples are computed. For Saito Frobenius manifolds of types $A$ and $D\,$, the extended prepotentials coincide with open WDVV solutions computed independently, whereas even the existence of the solution in type $E$ had not been previously discussed. On the other hand, new non-homogeneous solutions are constructed by duality. 
\end{abstract}

\maketitle
\tableofcontents

\newpage
\thispagestyle{empty}

\section{Introduction}
(Dubrovin-)Frobenius manifolds were introduced by B. Dubrovin in the nineties as a geometrical way to encode solutions $\Fr$ to a system of third-order PDEs arising from topological field theories, known as the (closed) \emph{Witten-Dijkgraaf-Verlinde-Verlinde (WDVV) equations} \cite{Dubrovin1996}:
\begin{align}\label{eq:WDVV}
  \sum_{\mu,\nu=1}^n(\eta^\sharp)_{\mu\nu}\,\Fr_{\alpha \beta\mu}\,\Fr_{\nu \gamma \delta}= \sum_{\mu,\nu=1}^n(\eta^\sharp)_{\mu\nu}\,\Fr_{\alpha\gamma \mu }\,\Fr_{\nu \beta \delta}\,.
\end{align}
Here, $\Fr_\alpha\equiv\partial_\alpha\Fr$ and $\eta^\sharp$ is the inverse of the matrix $\eta$ -- assumed to be constant and non-singular -- whose entries, in the case of a Frobenius manifold, are defined by $\eta_{\alpha\beta}:=\Fr_{1\alpha\beta}\,$. The function $\Fr$ is called \emph{prepotential} (or free energy), and it is used to define a multiplication on the tangent spaces of a complex manifold. The WDVV equations coincide with the requirement that this operation be associative. In the following decades, the WDVV equations -- and hence Frobenius manifolds -- have appeared in various areas of Mathematics, from enumerative geometry and Gromow-Witten invariants to singularity theory and integrable systems.

Various geometric structures generalising Frobenius manifolds were introduced over the years, originally in the works of Y. Manin and C. Hertling \cites{weakfrobmflds,maninflatfmanifolds,hertling_2002}. This has led to considering a different notion of potentiality, shifting from a single function to a vector-valued potential. Again, this is used to define a multiplication on the tangent spaces of a complex manifold. The requirement that this product be associative then forces its components, in a distinguished coordinate system, to solve a system of second-order PDEs, known as \emph{oriented WDVV equations}.

More recently, another set of associativity equations was introduced in the context of open Gromow-Witten theory \cite{horev2012opengromovwittenwelschingertheoryblowups}. For this reason, they are called \emph{open WDVV equations}. Various solutions to these equations have been constructed over the years, e.g. in \cites{alcoladophd,buryakopen,IntSystemsAD,BuryakOpenAD,BCT1,BuryakArsingularity,PSTInttheory,dealmeida2025openhurwitzflatf}. The key remark, made explicit, for example, in \cite{alcoladophd}, is that open WDVV equations can be seen as oriented WDVV equations with a ‘source’. More precisely, open WDVV equations coincide with the associativity equations for a vector potential defined on a rank-one fibre bundle over a Frobenius manifold in such a way that the projection map induces homomorphisms between the tangent algebras. In \cite{alcoladophd}, a characterisation of solutions to the open WDVV equations for any given source Frobenius manifold is given in terms of a function $h$ satisfying a simpler system of PDEs. In \cite{dealmeida2025openhurwitzflatf}, it was shown that the Landau-Ginzburg superpotential $\lambda$ of any Hurwitz Frobenius manifold satisfies such conditions, thus giving rise to an open WDVV solution naturally associated to any such Frobenius manifold, and this solution is very easily constructed from the superpotential (up to a constant of integration).

\subsection{Summary of results}
In this work, we complete the picture we have discussed so far, by filling up the following diagram with the missing arrows, denoted in red:
\begin{equation}
       \begin{tikzcd}[every cell/.append style={align=center}]
   \begin{tabular}{c}
         $\Fr$ \\
          \footnotesize Weighted-homogeneous \\\footnotesize WDVV solution.
      \end{tabular} &  \begin{tabular}{c}
         $\Omega$ \\
          \footnotesize Weighted-homogeneous \\\footnotesize open WDVV solution.
      \end{tabular}\arrow[dd, leftrightarrow,red, "\text{Generalised duality?}"]\\
       \begin{tabular}{c}
          $(\lambda,\omega)$\\
         \footnotesize LG model 
        \end{tabular} \arrow[ur,rightsquigarrow, "\text{\cite{dealmeida2025openhurwitzflatf}}"]\arrow[u, rightsquigarrow]\arrow[d, rightsquigarrow]&\\
  \begin{tabular}{c}
         $\Frst$ \\
          \footnotesize Dual WDVV solution.
      \end{tabular} \arrow[uu, bend left=60, leftrightarrow, "\text{\cite{Dub04}}"]  &\begin{tabular}{c}
         $\overset{\star}{\Omega}$ \\
          \footnotesize Dual open WDVV solution.
      \end{tabular} \arrow[lu, leftsquigarrow, red, "?"] 
    \end{tikzcd}
\end{equation}

More precisely, we consider the Dubrovin dual to a Hurwitz Frobenius manifold \cite{Dub04}. We show that this also extends to an F-manifold with compatible flat connection on the universal curve over the Hurwitz space, as defined in \cites{maninflatfmanifolds,RomanoDoubleHurwitz} respectively. In particular, this corresponds to the standard replacement of $\lambda$ by $\log\lambda\,$. 

Furthermore, we show that the existence of the extension is very closely related to the twisted period representation of the flat coordinates for the deformed connection associated to any F-manifold with compatible connection. As a consequence, our result holds true whenever one has a Landau-Ginzburg model for an F-manifold with flat connection, generalising the definition of \cite{Dubrovin1999} in the Frobenius manifold setting. In particular, the construction of a flat extended Dubrovin connection is not strictly necessary for the existence of the extension. Landau-Ginzburg models for F-manifolds with compatible connections are, however, notoriously difficult to construct for structures other than the ones coming from a Frobenius manifold via Dubrovin duality.\\

Secondly, we obtain two F-manifolds with flat connection on the universal curve over a Hurwitz Frobenius manifold -- i.e. the extension of the Frobenius manifold structure and of its Dubrovin dual. In \cref{sec:duality}, we study their duality in the sense of \cite{maninflatfmanifolds}. By doing so, we are able to retrieve the conditions imposed on the solutions to the open WDVV equations, e.g. in the works \cites{horev2012opengromovwittenwelschingertheoryblowups,buryakopen,BuryakOpenAD,BuryakArsingularity,PSTInttheory} and more, dealing with the problem starting from the A-model side -- more precisely from open Gromow-Witten theory. 

In particular, we show that the identities of the two extended multiplication coincide with the natural extensions of the identity and Euler vector field on the Hurwitz space to the universal curve, as considered in \cites{Dubrovin1996,RomanoDoubleHurwitz} and therein denoted by $D_e$ and $D_E$ respectively. We further argue that the conditions on the primary form, on the other hand, make sure that such vector fields on the universal curve exhibit the same properties as their projections on the moduli space, under the natural extension of the Frobenius manifold structure from \cite{dealmeida2025openhurwitzflatf}. In particular, they are respectively flat and affine with respect to the compatible connection. This gives rise to a bi-flat F-manifold on the universal curve, thus fitting such a more general geometric structure in the B-model side. On the WDVV-solution side, this produces weighted-homogeneous solutions to the open WDVV equations, with the same homogeneity as the one appearing in open Gromow-Witten theory. In particular, the degree of the additional coordinate along the fibre is determined by the conformal dimension (or charge) of the underlying Frobenius manifold.

Next, we study how to define the extended Dubrovin connection for the bi-flat F-manifold structure on the universal curve. We show that, upon requiring compatibility with the extended Dubrovin connection on the Hurwitz space, the obvious generalisation of the definition in the Frobenius manifold setting is the unique flat connection in this case. More generally, we notice that, for rank-one extensions in which the eventual identity should fail to be affine, some further conditions need to be met in order for the extended Dubrovin connection on the extension to be flat.

Finally, we determine the relation between the extensions of the Levi-Civita connections of the Frobenius metric $\eta$ and the intersection form $g\,$. In particular, we show that this is naturally given by the same equation that takes place on the Hurwitz Frobenius manifold, with the obvious replacements of the corresponding objects with the ones on the universal curve. \\

Thirdly, in \cref{ss:typeDE}, the case of Saito Frobenius manifolds on the parameter spaces of miniversal deformations of surface singularities of type-$ADE$ is discussed in detail. In types $A$ and $D\,$, the open WDVV solutions that gives the rank-one extension of the Hurwitz Frobenius manifold can be computed explicitly. They are, respectively, polynomial and Laurent-polynomial solutions, and coincide with the ones from \cites{BuryakArsingularity,BuryakOpenAD}. Notice that these can also be found as reductions of the open WDVV solutions naturally associated to suitable infinite-dimensional Frobenius manifolds \cite{Ma25}. In type-$E$, even the existence of the solution was not known. However, this construction will work even in the exceptional case, and we can therefore write an implicit expression for the solution in terms of the miniversal deformation. In particular, this will not be a (Laurent) polynomial, since the curve has positive genus. Nevertheless, the seemingly-accidental property, noticed in \cite{BuryakOpenAD}, that the dependence of the open WDVV solution on the flat coordinates of the underlying Saito Frobenius manifold is almost entirely determined by the transition functions between the parameters of the miniversal unfolding and the Saito flat coordinates is now a mere consequence of the fact that the open WDVV solution is constructed directly from the Landau-Ginzburg superpotential. Therefore, the same property will hold true in type-$E$.

Furthermore, as described in \cite{topstringsd<1}, the Landau-Ginzburg superpotentials for these Frobenius manifolds are given by restrictions of the miniversal deformations to a curve in $\CC^2$. The Landau-Ginzburg superpotentials are, therefore, non-polynomial except for type-$A$. As a consequence, the open WDVV solutions are not polynomial (save for type-$A$), as already noted in \cite{BuryakOpenAD}. We conjecture that the operation of restricting to a curve in $\CC^2$ commutes with the operation of extending. More precisely, we conjecture that there exist a rank-two weighted-homogeneous \emph{polynomial} extension of each Saito Frobenius manifold of type-$ADE$ which restrict to the known (possibly non-polynomial) rank-one solution. We provide a rank-two extension satisfying these properties in type-$A\,$. 

Finally, in \cref{sec:dualsolutions}, we provide explicit classes of dual solutions to the open WDVV equations. These are completely new to our knowledge. In particular, they are not weighted-homogeneous. Some of these are generalised to a root-theoretic construction of open WDVV prepotentials in the direction of \cite{veselovcech} in the subsequent work \cite{opencech}.

\begin{acknowledgements}
  A. P. would like to thank Andrea Brini, Jingxiang Ma, and Karoline van Gemst for helpful discussions and suggestions. A. P. also wishes to thank Youjin Zhang, Di Yang, and Dafeng Zuo, as well as Tsinghua University and the University of Science and Technology of China, for their hospitality during the final stages of this work.
   A. P. was supported by a Ph.D. studentship of the EPSRC Doctoral Training Partnership (EP/W524359/1). 
\end{acknowledgements}



\section{Preliminaries}\label{sec:preliminaries}

\subsection{Almost-flat F-manifolds}

\begin{defn}[\cites{Arsie_2022}] An \emph{almost-flat F-manifold} is a complex manifold $M$ equipped with:
\begin{itemize}
\item An $\Hol_M$-bilinear multiplication $\bullet:\X_M\otimes_{\Hol_M}\X_M\to\X_M$ on the tangent sheaf,
\item An affine connection $\nabla:\X_M\to\Omega_M^1\otimes_{\Hol_M}\X_M\,$,
        \item A distinguished holomorphic vector field $e\in\X_M(M)\,$,
        \end{itemize}
    satisfying:
    \begin{enumerate}
           \item [(AFM1)] For any $z\in\CC\,$, the \emph{deformed connection} ${}^z\nabla:=\nabla+z\,\bullet$ is flat and torsion-free.
        \item [(AFM2)] At any $p\in M\,$, $e_p$ is the identity of $\bullet_p\,$.
    \end{enumerate}
    If, further, $e$ is a flat section of $\nabla\,$, $M$ is called \emph{flat F-manifold}.
\end{defn}
\begin{oss}
The condition (AFM1) is a shorthand for a series of requirements on $\bullet$ and $\nabla\,$. In particular, a simple computation shows that:
\begin{align}\label{eq:torsiondeformed}
      {}^z T(X,Y)& =  {}^\nabla T(X,Y)+z (X\bullet Y-Y\bullet X)\,.
\end{align}
Hence, ${}^z\nabla$ is torsion-free identically in $z$ if and only if $\nabla$ is torsion-free and $\bullet$ is commutative.

On the other hand, one similarly obtains the following expression for the Riemann tensor of ${}^z\nabla\,$:
\begin{equation}
    \begin{split}
        {}^z R(X,Y)Z&={}^\nabla R(X,Y)Z+z\bigl[(\nabla_X\bullet)(Y,Z)-(\nabla_Y\bullet)(X,Z)\bigr]+\\
        &\quad+z^2\bigl[X\bullet(Y\bullet Z)-Y\bullet(X\bullet Z)\bigr]\,.
    \end{split}
\end{equation}
Therefore, (AFM1) is equivalent to:
\begin{enumerate}
    \item [(AFM1a)] $\bullet$ is associative\footnote{Strictly speaking, the conditions are commutativity and right-associativity. However, it can be easily shown that they are equivalent to commutativity and fully-fledged associativity.} and commutative.
    \item [(AFM1b)] $\nabla$ is a flat, torsion-free connection.
    \item [(AFM1c)] For any one-form $\alpha\,$, the rank-two tensor field $(X,Y)\mapsto (X\bullet Y)\intprod \alpha$ is a Codazzi tensor \cite{codazzi}, i.e. a totally symmetric tensor field whose covariant derivative is also totally symmetric.
\end{enumerate}
\end{oss}
\begin{oss}
    Condition (AFM1) could be loosened by requiring that, while still being torsion-free, the deformed connection ${}^z\nabla$ have a curvature that is independent of $z\,$. In such a way, one allows for non-flat connections, while still retaining compatibility with the product. These structures are usually called \emph{F-manifolds with (compatible) connection} \cite{maninflatfmanifolds}. Clearly, in our definition, almost-flat F-manifolds are F-manifolds with compatible flat connection.
\end{oss}
\begin{oss}
    Perhaps counter-intuitively, the adjective ‘flat' here refers to flatness of the identity rather than of the connection, as the latter will, in this work, always be assumed to be flat. In other words, we use the definition in \cite{Arsie_2022}, rather than e.g. the one in \cite{ArsieBuryakLorRossi}, where non-flat connections are also considered. In the language of the latter paper, what we call (almost)flat F-manifold would be an F-manifold with compatible flat connection and (non)flat unity. While it is true that, from the F-manifold point of view, flatness of the unity is an accessory condition that simply picks out one of the flat connections compatible with a given multiplication -- and therefore it can always be achieved by a different choice of $\nabla\,$, see e.g. \cite{LP11} -- from the point of view of this paper it is more natural to consider a different (flat) connection in the family of compatible ones, for which the unity might not happen to be flat, hence the specification. This will hopefully not cause any confusion to the reader.
\end{oss}

Returning to the properties of almost-flat F-manifolds in our case, we have said that the multiplication is a Codazzi tensor. From \cite{codazzi}, since $^\nabla R=0$ (in particular, its sectional curvature is constant), it follows that, for any $p\in M\,$, there exist a neighbourhood $U$ and a local holomorphic vector field $\Fr\in\X_M(U)$ such that $\bullet\lvert_U=\nabla^2\Fr\,$, i.e. $X\bullet Y=\nabla_X\nabla_Y\Fr-\nabla_{\nabla_XY}\Fr$ for any two vector fields $X$ and $Y$. The vector field $\Fr$ is called \emph{(local) vector (pre)potential} of the almost-flat F-manifold. Associativity of $\bullet\,$, then, translates into $\Fr$ satisfying the following identity identically in $X,Y,Z\in\X_M(U)\,$:
\begin{equation}\label{eq:orientedWDVVcoordinatefree}    \nabla^2_{X,\nabla^2_{Y,Z}\Fr}\,\Fr=\nabla^2_{Z,\nabla^2_{X,Y}\Fr}\,\Fr\,.
\end{equation}
In a local coordinate system, this is a second-order system of PDEs for the components of $\Fr\,$. It is a well-known fact that finding solutions to this system is highly non-trivial. Notice that $\Fr$ is only defined up to an affine vector field with respect to $\nabla\,$, i.e. an element of $\ker\nabla^2\,$.

Now, since $\nabla$ is flat, we have distinguished coordinate systems in a neighbourhood of each point. Namely, the tangent sheaf is generated, as a sheaf of $\Hol_M$-modules, by the flat sections of $\nabla\,$. Since flat sections are a subsheaf of commutative Lie subalgebras of $\X_M$, it follows that there exists a set of coordinates $v_1,\dots, v_n$ in a neighbourhood of each point of $M$ such that $\bigl\{\partial_{v_1},\dots,\partial_{v_n}\bigr\}$ is a local frame of flat sections of $\nabla\,$. Coordinate systems with this property are usually referred to as being themselves \emph{flat}. 

In a system of flat coordinates, \cref{eq:orientedWDVVcoordinatefree} reads:
\begin{equation}
    \sum_{\mu=1}^n\pdv[2]{\Fr_\alpha}{v_\beta}{v_\mu}\pdv[2]{\Fr_\mu}{v_\gamma}{v_\delta}=\sum_{\mu=1}^n\pdv[2]{\Fr_\alpha}{v_\gamma}{v_\mu}\pdv[2]{\Fr_s}{v_\beta}{v_\delta}\,.
\end{equation}
These is the aforementioned system of \emph{oriented Witten-Dijkgraaf-Verlinde-Verlinde equations} \cite{maninflatfmanifolds}.\\

In the case of a flat F-manifold, the identity is a flat section of $\nabla\,$, and its restriction to a suitable open subset can therefore be chosen to be a vector field in any local flat frame. Unless otherwise stated, in such a case we are always going to make the choice that $e=\partial_{v_1}\,$.


\subsection{Frobenius manifolds and WDVV equations}
\begin{defn}
    A (Dubrovin-)Frobenius manifold with charge (or conformal dimension) $d\in\CC$ is a flat F-manifold equipped with a non-degenerate, bilinear form $\eta:\X_M\otimes_{\Hol_M}\X_M\to\Hol_M$ -- called \emph{Frobenius metric} -- and a distinguished global holomorphic vector field $E$ -- called \emph{Euler vector field} -- such that:
    \begin{enumerate}
        \item [(FM1)] The flat connection $\nabla$ is the Levi-Civita connection of $\eta\,$, i.e. $\nabla\eta=0\,$.
        \item [(FM2)] $\eta$ is compatible with $\bullet$ in the sense that, for any holomorphic vector field $X$, the endomorphism on $TM$ given by multiplication by $X$ is self-adjoint with respect to $\eta\,$.
        \item [(FM3)] The Euler vector field satisfies:
        \[
        \begin{aligned}
            \Lie_Ee&=-e\,,&&& \Lie_E\bullet&=\bullet\,, &&& \Lie_E\eta&=(2-d)\,\eta\,.
        \end{aligned}
        \]
        Here, $\Lie$ denotes the Lie derivative.
    \end{enumerate}
\end{defn}
\begin{oss}
    If the multiplication $\bullet$ is generically semi-simple -- i.e. it is a semi-simple multiplication on the tangent space of a generic point -- we are going to say that the Frobenius manifold is semi-simple.
\end{oss}
\begin{oss}
    To make it clear that $\nabla$ is the Levi-Civita connection of $\eta\,$, we shall, from now on, use the notation ${}^\eta\nabla\,$.
\end{oss}
Since Frobenius manifolds are flat F-manifolds, then locally their product admits a vector potential. However, the existence of a compatible non-degenerate, symmetric bilinear form gives a canonical identification of the tangent and cotangent bundles. As a consequence, one can consider the tensor field $c\,$, $\eta$-dual to $\bullet\,$, defined as follows on any triple of holomorphic vector fields $X\,,Y$ and $Z$:
\begin{equation}
    c(X,Y,Z):=\eta(X\bullet Y\,,\,Z)\,.
\end{equation}
Owing to (FM2) and the commutativity of $\bullet\,$, this is a section of the third symmetric power of $T^*M\,$. (AFM1c), then, implies that ${}^\eta\nabla  c$ is also totally symmetric. Hence, $c$ is still a Codazzi tensor, now of rank three, on a space of constant curvature. Therefore, there exists, in a neighbourhood $U$ of each point, a local \emph{function} $\Fr\in\Hol_M(U)$ such that:
\begin{equation}
    c\,\lvert_U={}^\eta\nabla ^3\Fr\,.
    \end{equation}
$\Fr$ is the \emph{prepotential} (or free energy) of the Frobenius manifold \cite{Dubrovin1996}, and it determined up to a local holomorphic function in the kernel of ${}^\eta\nabla ^3\,$. A vector potential can be constructed from the free energy by taking the $\eta$-dual of $\dd\Fr\,$ -- i.e. the gradient of $\Fr\,$.

As for the metric, since it clearly coincides with $c(e,-,-)$ as a section of $T^*M^{\otimes 2}\,$, it can also be locally expressed in terms of the prepotential via:
\begin{equation}
    \eta\,\lvert_U={}^\eta\nabla ^2(\Lie_e\Fr)\,,
\end{equation}
since $e$ is flat.

Finally, for the Euler vector field, we firstly notice that:
\begin{lem}\label{lem:EVFaffine}
    Let $M$ be a Frobenius manifold. The Euler vector field $E$ is affine with respect to ${}^\eta\nabla $, i.e. ${}^\eta\nabla ^2 E=0$.
\end{lem}
\begin{proof}
If $\xi $ is a conformal Killing vector field for the metric $\eta$, i.e. if $\Lie_\xi\eta=f\,\eta$ for some $f\in\Hol_M$, then a simple generalisation of the second Killing identity shows that, for any two $X,Y\in\X_M$:
\[
 {}^\eta\nabla ^2_{X,Y}\xi+ {}^\eta R(X,\xi)Y=\tfrac12\bigl(Y(f)\,X+X(f)\,Y-\eta(X,Y)\, {}^\eta\nabla  f\bigr)\,,
\]
 A quick proof of this relies on the observation that the left-hand side coincides with $(\Lie_\xi {}^\eta\nabla )_XY\,$, hence the right-hand follows from how the Levi-Civita connections changes under conformal transformations. This is a standard result, see e.g. \cite{leeRiemannian}.

Now, $E$ is a conformal Killing vector field for the Frobenius metric -- whose Riemann tensor vanishes -- with constant conformal factor $f=2-d$. This proves the statement.
\end{proof}
Notice, then, that the properties $\Lie_E\bullet=\bullet$ and $\Lie_E\eta=(2-d)\,\eta$ imply that $\Lie_Ec=(3-d)\,c\,$. As a consequence, the prepotential satisfies:
\begin{equation}\label{eq:quasihomogeneityF}
    \Lie_E\Fr=(3-d)\,\Fr\,.
\end{equation}

In a system of flat coordinates for the connection ${}^\eta\nabla \,$, now, the metric $\eta$ will be represented by a constant, non-degenerate matrix according to (FM1). In a local frame with $e=\partial_{v_1}\,$, then this is the requirement that the matrix $(\Fr_{1\alpha\beta})_{\alpha,\beta=1,\dots, n}$ be constant and invertible. The associativity of the product is, as already mentioned, equivalent to the prepotential being a solution to the WDVV equations \ref{eq:WDVV}. Finally, the components of the Euler vector field are at most linear functions of the flat coordinates, due to \cref{lem:EVFaffine}. In particular, we can always choose the flat coordinates $\bm{v}$ so that the coordinate vector fields are eigenvectors of $X\mapsto{}^\eta\nabla _XE$ \cite{Dubrovin1999}. In other words, if we denote $k:=\rank{}^\eta\nabla  E\,$, $E$ is decomposed in the associated local frame as follows:
\begin{equation}\label{eq:EulerVFflatcoords}
\begin{aligned}
        E&=\sum_{\mu=1}^kd_\mu \,v_\mu\partial_{v_\mu}+\sum_{\mu=k+1}^nr_{\mu-k} \,\partial_{v_\mu}\\
        &\equiv \sum_{\mu=1}^k(1-q_\mu)\,v_\mu\,\partial_{v_\mu}+\sum_{\mu=k+1}^nr_{\mu-k} \,\partial_{v_\mu}\,,
\end{aligned}
\end{equation}
for some complex numbers $d_1,\dots, d_k\in\CC^\ast\,$, $r_1,\dots, r_{n-k}\in\CC\,$. $\Lie_Ee=-e$ finally imposes $d_1=1\,$.

As a consequence, \cref{eq:quasihomogeneityF} says that $\Fr$ is a weighted-homogeneous function with weights $(d_1,\dots, d_k\,;\,r_1,\dots,r_{n-k}\,;\,3-d)$ of the flat coordinates $\bm{v}\,$, by which we mean:
\begin{equation}
    \Fr(c^{d_1}\,v_1,\dots, c^{d_k}\,v_k,v_{k+1}+c\,r_1,\dots,v_n+c\,r_{n-k})=c^{3-d}\,\Fr(\bm{v})\,,\qquad \forall c\in\CC^\ast\,.
\end{equation}
$3-d$ will sometimes be called the \emph{overall weight} (or degree). Terminology is borrowed from \cite{ET11}.

As anticipated, this gives a one-to-one correspondence between Frobenius manifold structures and weighted-homogeneous solutions to the WDVV equations. 

\subsection{Extended Dubrovin connection}
Using all the elements in the structure of a Frobenius manifold, one can define an extension of the Levi-Civita connection ${}^\eta\nabla$ to the trivial bundle $\RS\times M\,$. We remark that a key point is that, since the bundle is trivial, there is a canonical way to identify the pull-back of the tangent bundle of $M$ with a sub-bundle to the tangent bundle to $\RS\times M\,$. It is well-known that this is not necessarily true for a non-trivial fibre bundle without additional structure.
    \begin{defn}
Let $M$ be a Frobenius manifold, $\pr:\RS\times M\twoheadrightarrow M$ be the trivial bundle. The \emph{extended Dubrovin connection} of $M$ is the connection $\widehat{\nabla}$ on $\pr^*TM$ defined by, for any $Y\in\pr^*\X_M$:
\[
\begin{aligned}
    \forall X&\in \pr^*\X_M: & && \widehat{\nabla}_XY&:={}^\eta\nabla_XY+ z \, X\bullet Y\,,\\
    &&&&\widehat{\nabla}_{\partial_ z }Y&:={}^\eta\nabla_{\partial_ z }Y+E\bullet Y-\tfrac{1}{ z }\mu(Y)\,,
\end{aligned}
\]
where $\mu$ is the following section of $\End(TM)\,$, called \emph{grading operator}:
\[
\begin{aligned}
  \mu&:=\bigl(1-\tfrac{d}{2} \bigr)\id-{}^\eta\nabla E\,.
\end{aligned}
\]
\end{defn}
This is important in the relation between Frobenius manifolds and the theory of isomonodromic deformations of meromorphic connections on the Riemann sphere. We shall not delve into this topic here. For our purpose, we are interested in the following:
\begin{thm}[\cite{Dubrovin1999}]\label{thm:flatnessExtDubConnection}
    Flatness of the extended Dubrovin connection is equivalent to flatness of ${}^\eta\nabla$, compatibility of ${}^\eta\nabla$ with the product $\bullet$, associativity and commutativity of $\bullet$ and ${}^\eta\nabla^2 E=0\,$.
\end{thm}

We shall denote flat coordinates of the extended Dubrovin connection by $\widetilde{\bm{v}}\,$. If $\bm{v}$ is a system of flat coordinates for ${}^\eta\nabla$ defined on the open set $U\,$, then $\widetilde{\bm{v}}(\bm{v}, z )$, defined on $\RS\times U\,$, can be thought of as a deformation of $\bm{v}\,$, $ z $ being the deformation parameter, i.e. $\widetilde{v}_\alpha\in v_\alpha\oplus z\bigl(\CC\llbracket z \rrbracket\otimes \Hol_M(U)\bigr)\,$. The coefficients of the positive powers of $ z $ are computed recursively \cites{Dubrovin1999}, and are the Hamiltonian densities for the flows of the principal hierarchy \cite{DZ01}.

\subsection{Dubrovin (almost) duality}
We introduce the following
\begin{defn}
    An \emph{almost-flat Riemannian F-manifold} is an almost-flat F-manifold $M$ equipped with a non-degenerate, bilinear form $\eta$ satisfying (FM1) and (FM2).
   
       If, further, the identity $e$ is an affine (resp. flat) vector field, we are going to say that $M$ is an \emph{affine (resp. flat) Riemannian F-manifold}.
\end{defn}
\begin{oss}
    Flatness of the deformed connection again ensures that the almost-flat Riemannian F-manifold structure is specified by a prepotential $\Fr$ solving the WDVV equations, as a function of the flat coordinates of $\eta\,$. Such a solution needs not be weighted-homogeneous, however.
\end{oss}
\begin{oss}
   Affine Riemannian F-manifolds are also flat F-manifolds, as noted in \cites{maninflatfmanifolds,Arsie_2022} -- even though the opposite is not necessarily true. However, we want here to emphasise the fact that the connection comes from a metric -- whereas the connection that makes an affine Riemannian F-manifold into a flat F-manifold is not necessarily a metric connection.
\end{oss}
\begin{oss}
    Almost-flat Riemannian F-manifolds are, in other words, almost-flat F-manifolds whose connection is the Levi-Civita connection of an invariant metric. In particular, they are clearly Riemannian F-manifolds in the sense of \cite{ArsieBuryakLorRossi}, with the additional requirement that the connection be flat (whereas the unity might not).
\end{oss}

This subsection contains a brief description of the following result:
\begin{thm}[\cite{Dub04}]
    Let $M$ be a Frobenius manifold. There are a closed subset $\Delta_M$ of $M$ -- called \emph{discriminant} -- a symmetric, non-degenerate bilinear form $g$ on the complement of $\Delta_M$ -- called \emph{intersection form} -- and multiplication of tangent vectors $\ast$ on $M\smallsetminus\Delta_M$ such that $M\smallsetminus\Delta_M\,$, equipped with $\ast$ and $g\,$, is an affine Riemannian F-manifold with identity $E\,$. Such a structure is said to be Dubrovin dual (or almost-dual) to $M\,$.
\end{thm}
We start by noticing that, using the identification between $TM$ and $T^*M$ induced by the metric $\eta\,$, one can define a multiplication on $\Omega_M\,$, which we are going to denote by the same symbol $\bullet\,$.
The intersection form is, then, defined first on the cotangent bundle.
\begin{defn}
    Let $M$ be a Frobenius manifold. The \emph{intersection form} is the symmetric bilinear form $g^\sharp:\Omega_M\otimes_{\Hol_M}\Omega_M\to\Hol_M$ defined by:
    \[
    \begin{aligned}
        \forall \alpha,\,\beta&\in \Omega_M:&&& g^\sharp(\alpha,\beta)&:=E\intprod(\alpha\bullet \beta)\,.
    \end{aligned}
    \]
\end{defn}
This also defines a metric on $TM$\,, whenever it is non-degenerate. It is not difficult to see that this happens on the complement of:
\begin{defn}
    Let $M$ be a Frobenius manifold, the \emph{discriminant} of $M$ is the closed subset:
    \[
    \Delta_M:=\bigl\{\,p\in M: \quad E_p\text{ does not admit an inverse with respect to }\bullet\,\bigr\}\,.
    \]
\end{defn}
On the complement of $\Delta_M\,$, $g^\sharp$ defines a symmetric, non-degenerate metric $g$, which we are still going to call intersection form. The product that gives flat deformations of the Levi-Civita connection of $g$ is the following:
\begin{defn}
    Let $M$ be a Frobenius manifold. The \emph{dual product} is the multiplication $\ast$ defined on the complement of $\Delta_M$ by:
    \[
    X\ast Y:=E^{-1}\bullet X\bullet Y\,,
    \]
    where $E^{-1}$ is the holomorphic vector field on $M\smallsetminus\Delta_M$ such that $E\bullet E^{-1}=e\,$.
    \end{defn}
    \begin{oss}
        Clearly, $\ast$ is unital with identity $E\,$.
    \end{oss}
\begin{oss}
    The prepotential $\Frst$ of the Dubrovin dual structure is said to be (almost) dual to the one of the Frobenius manifold it comes from.
\end{oss}

Now, in the original work \cite{Dub04}, Dubrovin duality relies on the full structure of a Frobenius manifold. The question of introducing a suitable generalisation to (almost-)flat F-manifolds was considered in the work \cite{AL13}.
\begin{defn}\label{defn:biflatFmfld}
    A \textit{bi-flat F-manifold} is a complex manifold equipped with two flat F-manifold structures $(\bullet, \nabla, e)$ and $(\ast,{}^\ast\nabla,E)$ related as follows:
    \begin{enumerate}
        \item $\Lie_E\bullet=\bullet\,$.
        \item $X\ast Y=E^{-1}\bullet X\bullet Y\,$, where $E^{-1}$ is the inverse of $E$ with respect to $\bullet\,$.
        \item For any three holomorphic vector fields, it is:
        \[
        \nabla_X(Y\bullet Z)-\nabla_Y(X\bullet Z)={}^\ast\nabla_X(Y\bullet Z)-{}^\ast\nabla_Y(X\bullet Z)\,.
        \]
    \end{enumerate}
\end{defn}
\begin{oss}\label{rem:dualconnectionbiflatFmfld}
    It is easy to see that the dual connection $\nabla^\ast$ is uniquely determined by the axioms of a bi-flat F-manifold, once the rest of the structure has been fixed:
    \[
    \nabla^\ast_XY=\nabla_XY-\nabla_{X\ast Y}E\,.
    \]See, for instance, \cite[Eq. (2.13)]{arsieGMbiflat} for the explicit expression.
\end{oss}

Building on the previous remark, in \cite{ArsieBuryakLorRossi}, bi-flat F-manifolds are shown to be equivalent to the following geometric structure, called \textit{homogeneous flat F-manifold} in \cite{ABLR23} or \textit{manifold with a Saito structure} in \cite{KMS18}. We shall here use the former:
\begin{defn}\label{defn:homogflatFmfld}
    A \textit{homogeneous flat F-manifold} is a flat F-manifold equipped with an invertible, affine vector field $E$ -- called \textit{Euler vector field} -- such that $\Lie_E\bullet=\bullet\,$.
\end{defn}


\subsection{Extensions of almost-flat F-manifolds and open WDVV equations}

\begin{defn}
    Let $M$ be an almost-flat F-manifold. An \emph{extension of $M$} is an almost-flat F-manifold $\widetilde{M}$ with a fibration $\pi:\widetilde{M}\to M$ and an Ehresmann connection $\vartheta:\pi^*(TM)\to T\widetilde{M}$ such that:
    \begin{enumerate}
        \item [(EAF1)] for any $p\in \widetilde{M}$, $(\pi_*)_p:T_p\widetilde{M}\to T_{\pi(p)}M$ is a homomorphism of complex algebras;
        \item [(EAF2)]The Ehresmann connection is a morphism of vector bundles with connection between $(\pi^\ast TM,\pi^\ast\nabla)$ and $(T\widetilde{M},\widetilde{\nabla})\,$, i.e. $\widetilde{\nabla}\circ\vartheta=(\id\otimes\vartheta)\circ\pi^*\nabla\,$.
    \end{enumerate}

    The codimension of $M$ in $\widetilde{M}$ is called \emph{rank} of the extension, and we shall denote it by $r\in\ZZ_{>0}\,$. 
\end{defn}
\begin{oss}
    Notice that (EAF1) implies that $\ker(\pi_*)_p$ is an ideal in $T_p\widetilde{M}$ for any $p\in \widetilde{M}\,$.
\end{oss}
\begin{oss}
    For (EAF2), we recall that, to any fibre bundle $\pi:\widetilde{M}\to M$, we associate the following short exact sequence of vector bundles on the total space $\widetilde{M}$:
    \begin{equation}\label{diag:sesfibrebundles}
        \begin{tikzcd}
                        0\arrow[r ]&\ker \pi_*\arrow[r] & T\widetilde{M}\arrow[r] & \pi^* TM \arrow[r]& 0\,.
        \end{tikzcd}
    \end{equation}
$\ker\pi_*$ is called \emph{vertical bundle}, and its sections are \emph{vertical} vector fields on $\widetilde{M}\,$. An Ehresmann connection $\vartheta$ is a splitting of the short exact sequence, which determines an identification of $\pi^* TM$ with a sub-bundle of $T\widetilde{M}$ in such a way that $T\widetilde{M}\cong\vartheta(\pi^*TM)\oplus\ker\pi_\ast\cong TM\oplus\ker\pi_\ast\,$. In other words, it is a way to lift vector fields on the base to vector fields on the total space. $\vartheta(\pi^*TM)$ is called the \emph{horizontal bundle}. Condition (EAF2) is the same as requiring that flat sections of $\nabla$ lift to flat sections of $\widetilde{\nabla}\,$, which is a natural requirement in the given setting. From now on, we shall identify $\pi^*TM$ with its image through the Ehresmann connection, and shall call its sections \emph{horizontal} vector fields. Finally we shall say that a vector field $X$ on $\widetilde{M}$ is \emph{projectable} if it is a section of $\pi^{-1}TM$ -- i.e. if it is horizontal and constant along the fibres. It is a well-known fact that projectable vector fields generate the space of horizontal vector fields as a $\Hol_{\widetilde{M}}$-module.
\end{oss}

The following result is, then, an obvious consequence of (EAF2):
\begin{prop}
    Let $\pi:\widetilde{M}\to M$ be an extension of the almost-flat F-manifold $M$. If $v_1,\dots, v_n$ is a system of flat coordinates for $M$ on the open set $U$, then there exists an \emph{adapted system of flat coordinates} $v_1,\dots, v_n,x_1,\dots, x_r$ for $\widetilde{M}$ on (an open subset of) $\pi^{-1}(U)$.
\end{prop}
\begin{oss}
    In the definition of \cite{alcoladophd}, the Ehresmann connection is replaced by the requirement that $\pi$ be a so-called map of affine manifolds, i.e. it is an affine transformation when represented in flat coordinate systems on $\widetilde{M}$ and $M\,$. This gives a way to lift local flat sections on $M$ to local flat sections on $\widetilde{M}$, which glue together to a holomorphic map between the two tangent bundles, i.e. an Ehresmann connection. Conversely, given an Ehresmann connection, one can construct adapted flat coordinate systems in $\widetilde{M}$ as in the previous Proposition. In any such coordinate system, $\pi$ is, then, just the projection on the first $n$ factors. Therefore it is, in particular, an affine map.
\end{oss}

Next, we want to relate a vector potential $\Fr$ of $M$ on some open subset $U$ and the vector potential $\widetilde{\Fr}$ of $\widetilde{M}$ on (possibly) an open subset of $\pi^{-1}(U)$. 
\begin{prop}[\cite{alcoladophd}]
Let $\pi:\widetilde{M}\to M$ be an extension of the almost-flat F-manifold $M$, let $\Fr$, $\widetilde{\Fr}$ denote the vector prepotentials of $M$ and $\widetilde{M}$ respectively on the open subsets $U$ and $\pi^{-1}(U)$. There exists a vector field $\Omega$ on (a subset of) $\pi^{-1}(U)$ -- called \emph{extended (pre)potential} -- such that:
     \begin{itemize}
         \item $\Omega\in\ker\pi_*$.
         \item $\widetilde{\Fr}=\Fr\circ\pi+\Omega$, up to affine vector fields.
     \end{itemize}
\end{prop}
\begin{proof}
    According to the splitting in \cref{diag:sesfibrebundles}, we can uniquely decompose $\widetilde{\Fr}$ as a sum of a section of $\pi^*TM$ and of $\ker\pi_*\,$. The fact that the former is the pull-back of $\Fr$ is a straightforward consequence of (EAF1).
\end{proof}

Let us, now, fix a system of adapted flat coordinates $v_1,\dots, v_n,x_1,\dots, x_r$ on $\widetilde{M}\,$. For the sake of clarity, we are going to employ Greek indices ranging from $1$ to $n\,$, and Latin indices ranging from $1$ to $r\,$. The multiplication table for the coordinate vector fields is then given by:
\begin{equation}\label{eq:multtableextension}
    \begin{aligned}
    \pdv{v_\mu}\,\widetilde{\bullet}\,\pdv{v_\nu} &=\sum_{\alpha=1}^n\pdv[2]{\Fr_\alpha}{v_\mu}{v_\nu}\,\pdv{v_\alpha}+\sum_{a=1}^r\pdv[2]{\Omega_a}{v_\mu}{v_\nu}\,\pdv{x_a}\,,\\
    \pdv{v_\mu}\,\widetilde{\bullet}\,\pdv{x_a}&=\sum_{b=1}^r\pdv[2]{\Omega_b}{v_\mu}{x_a}\,\pdv{x_b}\,,\\
    \pdv{x_a }\,\widetilde{\bullet}\,\pdv{x_b}&=\sum_{c=1}^r\pdv[2]{\Omega_c}{x_a}{x_b}\,\pdv{x_c}\,.
\end{aligned}
\end{equation}
The requirement of associativity will then constrain $\Omega$ to be a solution to a system of second-order PDEs having the second derivative of the components of the (known) vector potential $\Fr$ of the underlying almost-flat F-manifold as coefficients. For the sake of simplicity, we denote $c\indices{^\mu_{\alpha\beta}}:=\partial_\alpha\partial_\beta \Fr_\mu$ the structure constants on the base space. The associativity equations are, then \cite[equations 3.6-9]{alcoladophd}:
\begin{equation}\label{eq:openWDVVrankr}
    \begin{split}
\sum_{\mu=1}^nc\indices{^\mu_{\alpha\beta}}\pdv[2]{\Omega_a}{v_\mu}{ v_\gamma}+\sum_{b=1}^r\pdv[2]{\Omega_b}{v_\alpha }{v_\beta}\pdv[2]{\Omega_a}{x_b}{v_\gamma}&=\sum_{\mu=1}^nc\indices{^\mu_{\beta\gamma}}\pdv[2]{\Omega_a}{v_\alpha}{v_\mu}+\sum_{b=1}^r\pdv[2]{\Omega_b}{v_\beta}{v_\gamma}\pdv[2]{\Omega_a}{x_b}{v_\alpha}\,,\\
\sum_{\mu=1}^nc\indices{^\mu_{\alpha\beta}}\pdv[2]{\Omega_a}{v_\mu}{x_b}+\sum_{c=1}^r\pdv[2]{\Omega_c}{v_\alpha}{v_\beta}\pdv[2]{\Omega_a}{x_c}{x_b}&=\sum_{c=1}^r\pdv[2]{\Omega_c}{v_\beta}{x_b}\pdv[2]{\Omega_a}{x_c}{v_\alpha}\,,\\
\sum_{d=1}^r\pdv[2]{\Omega_d}{x_a}{x_b}\pdv[2]{\Omega_c}{v_\alpha}{x_d}&=\sum_{d=1}^r\pdv[2]{\Omega_d}{v_\alpha}{x_a}\pdv[2]{\Omega_c}{x_d}{x_b}\,,\\
\sum_{k=1}^r\pdv[2]{\Omega_k}{x_b}{x_c}\pdv{\Omega_d}{x_a}{x_k}&=\sum_{k=1}^r\pdv[2]{\Omega_k}{x_a}{x_b}\pdv[2]{\Omega_d}{x_k}{x_c}\,.
\end{split}
\end{equation}
These are the \emph{rank-$r$ open WDVV equations}.

In the case of a rank-one extension, the last two sets of equations are vacuous. Therefore, the rank-one open WDVV equations reduce to:
\begin{equation}\label{eq:openWDVV}
    \begin{aligned}
        \sum_{\alpha=1}^nc\indices{^\alpha_{\mu\nu}}\pdv[2]{\Omega}{v_\alpha}{v_\rho}+\pdv{\Omega'}{v_\rho}\pdv[2]{\Omega}{v_\mu}{v_\nu}&=\sum_{\alpha=1}^nc\indices{^\alpha_{\nu\rho}}\pdv[2]{\Omega}{v_\alpha}{v_\mu}+\pdv{\Omega'}{v_\mu}\pdv[2]{\Omega}{v_\nu}{v_\rho}\,,\\
        \sum_{\alpha=1}^nc\indices{^\alpha_{\mu\nu}}\pdv{\Omega'}{v_\alpha}+\Omega''\pdv[2]{\Omega}{v_\mu}{v_\nu}&=\pdv{\Omega'}{v_\mu}\pdv{\Omega'}{v_\nu}\,,
    \end{aligned}
\end{equation}
where $'\equiv\partial_x$ denotes the derivative in the fibre direction. We shall simply refer to the system \cref{eq:openWDVV} as \emph{open WDVV equations} \cite{alcoladophd}.
\begin{oss}
    In the rank-one case, fixed an adapted flat coordinate system, the extended prepotential is simply one local function on $\widetilde{M}\,$.
\end{oss}
\begin{oss}
    Some authors \cites{  horev2012opengromovwittenwelschingertheoryblowups,BCT1,BBvirasoro,BuryakOpenAD,PSTInttheory,dealmeida2025openhurwitzflatf} call open WDVV equations \cref{eq:openWDVV} only when $\Fr$ solves the ‘ordinary' -- or ‘closed' -- WDVV equations \cref{eq:WDVV}. This is also the case we will only be considering in this work. However, we wanted to highlight that, when discussing the geometrical setting where these kind of associativity equations arise, there is no need for such restriction.
    \end{oss}
\begin{lem}[\cite{alcoladophd}]\label{lem:auxiliaryextension}
    If $\Omega''\neq 0$ and $\Omega$ is a solution to the second set of equations in \cref{eq:openWDVV}, then it also solves the former. In such a case, the extension is called \emph{auxiliary}.
\end{lem}
\begin{proof}
    Since $\Omega''\neq 0$, we can multiply the first equation in \cref{eq:openWDVV} by $\Omega''$ and use the second one. The terms that do not cancel out give the associativity equations for $\Fr$.
\end{proof}
Auxiliary rank-one extensions of almost-flat F-manifolds were fully classified in \cite{alcoladophd}.
\begin{thm}[Theorem 3 in \cite{alcoladophd}]\label{thm:alcolado}
    Let $M$ be an almost-flat F-manifold, $v_1,\dots, v_n$ be a system of flat coordinates defined on an open set $U$, and $\Fr$ be a vector potential on $U$. If $h:U\times \CC\to \CC$ is a function satisfying $h'\neq 0$ and:
    \begin{equation}
        h_{\alpha\beta}=\pdv{x}\bigl[\tfrac{1}{h'}\bigl(h_\alpha h_\beta-\sum_{\mu=1}^n\Fr_{\mu,\alpha\beta}h_\mu\bigr)\bigr]\,,\qquad \forall \alpha,\beta=1,\dots, n\,,
    \end{equation}
then there exists a unique rank-one extension of $M$ with flat coordinates $v_1,\dots, v_n,x$ and extended prepotential $\Omega$ satisfying $\Omega'=h\,$.
\end{thm}

\subsection{Hurwitz Frobenius manifolds}\label{ssec:Hurwitz}
By Hurwitz Frobenius manifolds, we mean Frobenius manifold structures defined on Hurwitz spaces. These are moduli spaces of meromorphic functions on a smooth algebraic curve of fixed genus with given ramification over $\infty\,$. More precisely:
\begin{defn}[\cite{Dubrovin1996}]
    Let $g\in\ZZ_{\geq 0}$ and $\bm{n}\in\ZZ_{\geq 0}^k$. The \emph{Hurwitz space} $\Hw_{g\,;\,\bm{n}}$ is the space of equivalence classes of pairs $(C_g,\lambda)\,$, where $C_g$ is a genus-$g$ closed Riemann surface and $\lambda:C_g\to\RS$ is a meromorphic function of $C_g$ with $k$ poles $\infty_1,\ldots, \infty_k$ of order $n_1+1,\dots, n_k+1$ respectively. The equivalence relation is defined by the action of $\mathrm{Aut}(C_g)$ on the source curve. In other words, two pairs $(C,\lambda)\,$, $(D,\mu)$ are identified if $C\cong D$ and there exists an automorphism $f$ of $C$ such that $\mu=\lambda\circ f\,$.
\end{defn}
\begin{ex}
    Let $g=0\,$, $k=1$ and $\bm{n}=\ell\,$. $\Hw_{0\,;\,\ell}$ is the space of equivalence classes of meromorphic functions on $\RS$ with one pole of order $\ell+1\,$. A point in such a space can be uniquely written as follows:
    \[
    \lambda(x)=x^{\ell+1}+a_1x^{\ell-1}+\dots+a_\ell\,.
    \]
    This defines a system of global coordinates on it.
\end{ex}
\begin{defn}[\cite{RomanoDoubleHurwitz}]
    The \emph{universal curve} $\UC_{g\,;\,\bm{n}}$ over the Hurwitz space $\Hw_{g\,;\,\bm{n}}$ is the fibre bundle over the latter whose fibre over a point $[(C,\lambda)]$ is the curve $C$ with marked points given by the poles of $\lambda\,$.
\end{defn}
For any given point in the Hurwitz space, we have a function $\lambda:C_g\to \RS$ defined on the fibre of the universal curve over that point. These maps glue to a meromorphic map on the universal curve $\lambda:\UC_{g\,;\,\bm{n}}\to\RS\,$, called \emph{Landau-Ginzburg superpotential} (or simply superpotential for short) \cite{RomanoDoubleHurwitz}. This is a key ingredient in the construction of a Frobenius manifold structure over (a covering of) a Hurwitz space.

As already noted in \cite{Dubrovin1996}, and made further explicit in \cites{Shr05,Shr05a,RomanoDoubleHurwitz,DNOPS18}, $\lambda$ is a local biholomorphism away from its critical points and poles, and can therefore be used as a local coordinate in the fibre direction. This fails at the critical points $P_1,\ldots, P_N$, where the natural parameter is given by $\zeta_\alpha^2:=\lambda-\lambda(P_\alpha)$ for $\alpha=1,\ldots, N\,$. Notice that this is also true at the poles of $\lambda\,$. However, crucially, the natural local parameter $z_\alpha^{-(n_\alpha+1)}=\lambda $ around the pole $\infty_\alpha$ does not explicitly depend on the moduli point but only on the fibre coordinate $\lambda\,$. The choice of $\lambda$ as a coordinate on the fibres, whenever well-defined, induces an Ehresmann connection on the universal curve, by identifying nearby fibres along the level sets of $\lambda\,$. In \cite{RomanoDoubleHurwitz}, it is argued that this extends to a meromorphic connection on the universal curve, with singularities at the critical points. We shall denote the horizontal lift of a vector field $X$ on the moduli space to the universal curve under the Ehresmann connection defined by $\lambda$ by $X\lvert_\lambda\,$. \\

Another crucial tool in the definition of a Hurwitz Frobenius manifold is a \emph{primary differential} (or Saito form \cites{Saito81,Sai83}) $\omega\,$. These are \emph{some} holomorphic differentials on the fibre bundle $\widetilde{\UC}_{g\,;\,\bm{n}}\,$ over $\Hw_{g\,;\,\bm{n}}$ whose fibre over a point is the universal cover of the fibre of the universal curve over the same point. In other words, in each fibre of the universal curve we give a (possibly multi-valued) meromorphic differential on the fibre with poles at most at the marked points. $\omega$ needs to satisfy some admissibility conditions with respect to $\lambda\,$. These are described in \cite[Lecture 5]{Dubrovin1996}, where a complete classification of primary differentials into five families is given.

As noted e.g. in \cite{BvG22}, the $\lambda$-admissibility endows the universal curve with a further Ehresmann connection. This is defined by taking the horizontal distribution to the be the one spanned by $\omega$ itself. This gives another way of lifting vector field on the Hurwitz space to the universal curve. Since we mostly work with this connection in the later chapters, we shall more simply denote the lift of a vector field $X$ on the moduli space to the universal curve under the $\omega$-connection by $X$ itself. The relation between the vertical components of the two natural horizontal lifts goes by the name of \textit{thermodynamic identity} \cite{Dubrovin1996}.\\

The Landau-Ginzburg superpotential and the primary differential define the Frobenius manifold structure as follows. The Frobenius metric is given by the following residue formula:
\begin{equation}\label{eq:etaLGmodel}
    \eta(X,Y):=\sum_{x\in\mathrm{Cr}(\lambda)}\Res_x\bigl\{X(\lambda)\,Y(\lambda)\,\tfrac{\omega^2}{\dd\lambda}\bigr\}\,,
\end{equation}
where $\mathrm{Cr}(\lambda)$ is the set of critical point of $\lambda\,$, and it is understood that, when evaluating the expression at a point on the Hurwitz space, $\lambda$ and $\omega$ ought to be restricted to the corresponding fibres. Accordingly, $\dd\lambda$ is the differential on the fibre of $\UC_{g\,;\,\bm{n}}\,$. The fact that such a metric is non-degenerate and flat again follows from $\lambda$-admissibility, \cite{Dubrovin1996}.


The multiplication is, then, defined similarly e.g. by giving the tensor field $c$ instead, by an analogous residue formula:
\begin{equation}
     c(X,Y,Z):=\sum_{x\in\mathrm{Cr}(\lambda)}\Res_x\bigl\{X(\lambda)\,Y(\lambda)\,Z(\lambda)\,\tfrac{\omega^2}{\dd\lambda}\bigr\}\,.
     \end{equation}
     
As for the identity and Euler vector field, they are respectively defined to be the generators of the translations and rescalings in the stabiliser of $\infty$ for the action of the group of M\"{o}bius transformations on the target $\RS\,$. In other words, they are the generators of the actions on the superpotential $\lambda\mapsto\lambda+a$ and $\lambda\mapsto b\,\lambda$ for $a\in\CC$ and $b\in\CC^\ast\,$. As noted in \cites{Dubrovin1996,RomanoDoubleHurwitz}, if we use $\lambda$ as a coordinate on the fibres of the universal curve, $e$ and $E$ naturally extend to the generators of the action on $\UC_{g\,;\,\bm n}\,$:
\begin{equation}\label{eq:extensionseE}
    \begin{aligned}
        \ee&:=\partial_\lambda+e\lvert_\lambda\,,&&&\EEE&:=\lambda\partial_\lambda+E\lvert_\lambda\,.
    \end{aligned}
\end{equation}

It can be shown that this gives rise to a semi-simple Frobenius manifold structure on the Hurwitz space.
By Riemann-Hurwitz formula, its dimension is:
\begin{equation}
    d_{g,\bm{n}}:=\dim\Hw_{g\,;\,\bm{n}}=2(g-1)+2k+n_1+\dots+n_k\,.
\end{equation}
The idempotents are the coordinate vector fields in the system of coordinates given by the critical values of the superpotential \cites{Dubrovin1996,Dubrovin1999}. This is a complete frame of coordinate vector fields on the tangent bundle, whose coordinates are called \emph{canonical}, and are usually denoted by the letter $\bm{u}\,$.

The coordinate vector fields $\partial_{u_1},\dots, \partial_{u_{d_{g,\bm{n}}}}$ are eigenvectors of the endomorphism of $TM$ defined by $X\mapsto X\bullet E\,$, with eigenvalues $u_1,\dots, u_{d_{g,\bm{n}}}$ respectively. Therefore, the discriminant coincides with the subset where at least one of the critical values of $\lambda$ vanish.

Since we mostly work with the $\omega$-connection, we rather express $\ee$ and $\EEE$ in terms of the horizontal lifts $e$ and $E$ of the identity and Euler vector field respectively, rather than $e\lvert_\lambda$ and $E\lvert_\lambda\,$.
\begin{prop}\label{Hw:prop:liftsgeneratorsomegaconnection}
    Let $\ee\,$, $\EEE$ be the vector fields on the universal curve defined by \cref{eq:extensionseE}. If $e$ and $ E$ denote the horizontal lifts to the universal curve of the identity and Euler vector field on the Hurwitz space with respect to the connection induced by $\omega\,$, then:
    \begin{align*}
        \ee&=e +\tfrac{1}{\lambda_x}\bigl(1-\Lie_{e }\lambda\bigr)\,\partial_x\,,&
        \EEE&=E +\tfrac{1}{\lambda_x}\bigl(\lambda-\Lie_{E }\lambda
        \bigr)\,\partial_x\,,
    \end{align*}
    where $x$ is a local coordinate on the fibres of $\UC_{g\,;\,\bm n}$ such that $\omega=\dd{x}\,$, and $\lambda_x:=\partial_x\lambda$. In particular, $\Lie_\ee\lambda=1\,$, $\Lie_\EEE\lambda=\lambda\,$.
\end{prop}
\begin{proof}
We compute the expression for $\ee\,$, the other can be worked out similarly. Clearly, $e=\sum_\alpha\partial_{u_\alpha}$ in canonical coordinates \cite{Dubrovin1996}. Using the thermodynamic identities, the horizontal lifts of $\partial_{u_\alpha}$ with respect to the connection induced by $\omega$ and $\lambda$ respectively are related as follows:
    \[
    \partial_{u_\alpha}\lvert_\lambda=\partial_{u_\alpha}-\tfrac{1}{\lambda_x}\lambda_{u_\alpha}\,\partial_x\,,
    \]
    where $\lambda_{u_\alpha}=\partial_{u_\alpha}\lambda\,$. This gives:
    \[
    \begin{aligned}
        \ee&=\partial_\lambda+\sum_{\alpha=1}^{d_{g;\bm n}}\partial_{u_\alpha}\lvert_\lambda=\tfrac{1}{\lambda_x}\,\partial_x+\sum_{\alpha=1}^{d_{g;\bm n}}\bigl[\partial_{u_\alpha}-\tfrac{1}{\lambda_x}\lambda_{u_\alpha}\,\partial_x\bigr]\\
        &=\tfrac{1}{\lambda_x}\biggl(1-\sum_{\alpha=1}^{d_{g;\bm n}}\lambda_{u_\alpha}\biggr)\,\partial_x+\sum_{\alpha=1}^{d_{g;\bm n}}\partial_{u_\alpha}\,.
    \end{aligned}
    \]
\end{proof}

The intersection form is shown to be given by a similar residue expression:
\begin{equation}
    g(X,Y)=\sum_{x\in\mathrm{Cr}(\lambda)}\Res_x\bigl\{X(\log\lambda)\,Y(\log\lambda)\,\tfrac{\omega^2}{\dd\log\lambda}\bigr\}\equiv\sum_{x\in\mathrm{Cr}(\lambda)}\Res_x\bigl\{\tfrac{1}{\lambda}X(\lambda)Y(\lambda)\tfrac{\omega^2}{\dd\lambda}\bigr\}\,.
\end{equation}

Notice, finally, that giving a Landau-Ginzburg superpotential and a primary differential uniquely defines a Hurwitz Frobenius manifold, possibly on a subset of a Hurwitz space. For this reason, a pair $(\lambda,\omega)$ satisfying a suitable set of axioms is said to be an Landau-Ginzburg  model (or B-model, in the mirror symmetry language) for a Frobenius manifold. More precisely:
\begin{defn}[\cite{Dubrovin1999}]
    A \emph{genus-$g$ Landau-Ginzburg model} (or B-model) for a semi-simple Frobenius manifold $M$ is a pair $(\lambda,\omega)$ such that:
    \begin{itemize}
        \item $\lambda:M_g\to \RS$ is a meromorphic function on a holomorphic fibre bundle over $M$ whose fibres are isomorphic to an open domain $D_g$ of a closed Riemann surface of genus $g\,$.
        \item $\omega$ is a (possibly multi-valued) meromorphic differential on $D_g\,$.
    \end{itemize}
The pair needs to satisfy:
\begin{enumerate}
        \item [(LG1)] The critical values of $\lambda$ (restricted to each fibre) in $D_g$ are canonical coordinates for $M\,$.
        \item [(LG2)] The following expressions for the Frobenius metric $\eta$ on $M$ and the tensor field $c\,$, which is $\eta$-dual to the multiplication, hold:
        \[
        \begin{aligned}
                \eta(X,Y)&=\sum_{x\in\mathrm{Cr}(\lambda)}\Res_x\bigl\{X(\lambda)\,Y(\lambda)\,\tfrac{\omega^2}{\dd\lambda}\bigr\}\,,\\
     c(X,Y,Z)&=\sum_{x\in\mathrm{Cr}(\lambda)}\Res_x\bigl\{X(\lambda)\,Y(\lambda)\,Z(\lambda)\,\tfrac{\omega^2}{\dd\lambda}\bigr\}\,.
        \end{aligned}
        \]
        Here, $\mathrm{Cr}(\lambda)$ is the set of critical point of $\lambda\,$.
        \item [(LG3)] If $n=\dim M\,$, there exist cycles $\gamma_1,\dots, \gamma_n$ in $D_g$ such that a system of flat coordinates for the extended Dubrovin connection is given by the twisted periods:
            \[
\widetilde{v}_\mu:=\int_{\gamma_\mu}e^{z\,\lambda}\,\omega\,,\qquad \mu=1,\dots, n\,.
        \]
    \end{enumerate}
\end{defn}

\begin{thm}[\cite{dealmeida2025openhurwitzflatf}]
    Let $\lambda:\mathscr{D}\subseteq\UC_{g\,;\,\bm{n}}\to\RS$ be a Landau-Ginzburg superpotential (possibly defined on a subset of the universal curve) and $\omega$ be a compatible primary differential. There exists a unique rank-one extension of the Hurwitz Frobenius manifold structure on (a subset of) the Hurwitz space $\Hw_{g\,;\,\bm{n}}$ to a flat F-manifold on (a subset of) the universal curve such that:
    \begin{enumerate}
        \item A system of flat coordinates $\bm{v}$ on the subset $U$ in the Hurwitz space can be completed to an adapted system of flat coordinates $(\bm{v},x)$ on an open subset of $\pi^{-1}(U)$ on the universal curve, where $x$ is a coordinate on each fibre such that $\omega=\dd{x}\,$.
        \item The extended prepotential $\Omega$ in the system of adapted flat coordinates satisfies $\Omega_x=\lambda\,$.
    \end{enumerate}
\end{thm}
\begin{oss}
    The properties required for a Landau-Ginzburg model play an important role in the proof of this result. In some sense, they are equivalent to the existence of the extension. We shall comment on this more explicitly later on.
\end{oss}
\begin{oss}
    We stress that the Landau-Ginzburg superpotential might be defined on a subset of the universal curve because there are known examples in mirror symmetry of primary differentials that give a degenerate metric on the whole Hurwitz space, but produce a legitimate, meaningful Hurwitz Frobenius manifold when restricted to a suitable subset. For instance, this is what happens for the Landau-Ginzburg models for Saito and Dubrovin-Zhang Frobenius manifolds of types $BCD\,$, see \cites{topstringsd<1,DSZZ} respectively.
\end{oss}

\section{Dubrovin duality and rank-one extensions}\label{sec:dualityextension}

Throughout this section, we fix a genus-$g$ Hurwitz Frobenius manifold $\Hw_{g;\bm{n}}$ of dimension $n$ with primary differential $\omega$ and Landau-Ginzburg superpotential $\lambda:\UC_{g;\bm{n}}\to\RS$. We denote by $\bm{w}:=(w_1,\dots, w_n)$ a system of flat coordinates for its intersection form, so that $\Frst(\bm{w})$ is the corresponding solution to the WDVV equations. Furthermore, $\Hw_{g,\bm{n}}^\star$ is going to be its Dubrovin dual.

Unless otherwise stated, Latin indices will be used for flat coordinates $\bm{w}$ of the intersection form, and we are going to employ Einstein's summation convention over them in the proofs to make it more readable. On the other hand, we are going to use Greek indices for canonical coordinates, and no implicit summation over them will be assumed.
\begin{lem}\label{lem:lemma1}
   For $a\leq b\in\{1,\dots, n\}$, the functions:
    \begin{equation}\label{eq:psiab}
        \begin{aligned}
              K_{ab}&:=(\log\lambda)_{ab}-\dv{x}\biggl[\frac{1}{(\log\lambda)'}\biggl((\log\lambda)_a(\log\lambda)_b-\sum_{c,d=1}^n(g^\sharp)_{cd}\,\Frst_{abc}(\log\lambda)_d\biggr)\biggr]\\
    &\equiv\tfrac{1}{\lambda^2} \bigl(\lambda\lambda_{ab}-\lambda_a\lambda_b\bigr)-\dv{x}\biggl[\tfrac{1}{\lambda'}\biggl(\tfrac{1}{\lambda}\lambda_a\lambda_b-\sum_{c,d=1}^n(g^\sharp)_{cd}\,\Frst_{abc}\lambda_d\biggr)\biggr]\,,
        \end{aligned}
    \end{equation}
   where $'\equiv\dv{x}$ and $f_a\equiv\pdv{f}{w_a}$, are holomorphic at the critical points of $\lambda$.
\end{lem}
\begin{proof}
Let $q_1,\dots, q_n$ denote the critical points of $\lambda$, and let $u_1,\dots, u_n$ be the corresponding critical values.
    Firstly, we note that they are also the critical points of $\log\lambda$. On the other hand, if we denote by $y_\mu:=(\log\lambda)(q_\mu)$ the corresponding critical values, the following obvious relation holds for $\mu=1,\dots, n$: $u_\mu=e^{y_\mu}$. As noted e.g. in \cite[Lemma 1.55]{RileyPhD}, the critical values $y_1,\dots, y_n$ of $\log\lambda$ are canonical coordinates for the dual product. For convenience, we introduce the coordinates $\xi_1,\dots, \xi_n$ on $\Sigma_g$ centered at the critical point $q_1,\dots, q_n$, respectively. As stated in \cite{Dubrovin1996}, the superpotential has the following asymptotic expansion around its critical points:
    \begin{equation}\label{eq:lambdacritical}
            \lambda(\xi_\mu)=u_\mu-\frac{1}{2\eta_\mu}\,\xi_\mu^2+\order{\xi_\mu^3}\,,\qquad \mu=1,\dots, n\,,
    \end{equation}
    where $\eta_\mu$ is the $\mu^{\mathrm{th}}$ diagonal entry of the (diagonal) matrix representing $\eta$ in canonical coordinates.
    It then easily follows that a completely analogous expansion applies to its logarithm:
    \[
   \Lambda(\xi_\mu):= (\log\lambda)(\xi_\mu)=y_\mu-\frac{e^{-y_\mu}}{2\eta_\mu}\,\xi_\mu^2++\order{\xi_\mu^3}\,.
    \]
In particular, $(\log\lambda)_{ab}$ is holomorphic at $0$.
    
    On the other hand:
 \[
 \begin{aligned}
     \frac{\Lambda_a\Lambda_b-g^{cd}\,\Frst_{abc}\Lambda_d}{\Lambda'}&=\frac{\eta_\mu e^{y_\mu}}{\xi_\mu}\biggl[\pdv{y_\mu}{w^a}\pdv{y_{\mu}}{w^b}-(c^\star)\indices{^d_{ab}}\pdv{y_\mu}{w^d}\biggr]+\order{1}\\
     &=\frac{\eta_\mu e^{y_\mu}}{\xi_\mu}\biggl[\pdv{y_\mu}{w^a}\pdv{y_\mu}{w^b}-\sum (c^\star)\indices{^q_{\nu\rho}}\pdv{w^d}{y_\alpha}\pdv{y_\nu}{w^a}\pdv{y_\rho}{w^b}\pdv{y_\mu}{w^d} \biggr]+\order{1}\\
     &=\frac{\eta_\mu e^{y_\mu}}{\xi_\mu}\biggl[\pdv{y_\mu}{w^a}\pdv{y_\mu}{w^b}-\sum\delta\indices{^\mu_\nu}\delta\indices{^\mu_\rho}\pdv{y_\nu}{w^a}\pdv{y_\rho}{w^b}\biggr]+\order{1}\\
     &=\frac{\eta_\mu e^{y_\mu}}{\xi_\mu}\biggl[\pdv{y_\mu}{w^a}\pdv{y_\mu}{w^b}-\pdv{y_\mu}{w^a}\pdv{y_\mu}{w^b}\biggr]+\order{1}=\order{1}\,.
 \end{aligned}
 \]
 This proves the statement.
\end{proof}
\begin{lem}\label{lem:holozero}
    The functions $\{K_{ab}\}_{1\leq a\leq b\leq n}$ in \cref{eq:psiab} are holomorphic at the zeros of $\lambda$.
\end{lem}
\begin{proof}
    Let \(x_0=x_0(\bm z)\) be a locally defined zero section of \(\lambda\), of multiplicity \(\ell\). A zero of \(\lambda\) cannot be a pole of $\omega\,$, therefore \(x\) is a regular local coordinate in a neighbourhood of \(x_0\,\).
Let, then, $q(\bm z):=x(x_0(\bm z))$ be the location of the zeros in such coordinate, and introduce $\xi:=x-q(\bm z)$ as a coordinate centered at $x_0\,$, for computational convenience. 
Now, according to e.g. \cite{Dubrovin1996}*{Remark 5.3}, $q$ is an affine function of the flat coordinates $\bm z$ of the intersection form. Since differentiation with respect to the base coordinates is performed with \(x\) held constant, \( \xi_a=-q_a\,. \)

Since \(\lambda\) has a zero of order \(\ell\) at \(\xi=0\,\), we may write, with the same notation as in the previous Lemma:
\[ 
\Lambda = \ell\log\xi + A(\bm z) + B(\bm z)\xi + \order{\xi^2}\,.
\]
Differentiating gives 
\[ 
\begin{aligned} 
    \Lambda' &= \ell\tfrac{1}{\xi} + B + \order{\xi}\,,\\ 
    \Lambda_a &= -\tfrac{1}{\xi}\ell q_a + A_a-Bq_a + \order{\xi}\,,\\ 
    \Lambda_{ab} &= -\tfrac{1}{\xi^2}\ell q_aq_b - \tfrac{1}{\xi}\ell q_{ab} + \order{1}=-\tfrac{1}{\xi^2} \ell\,q_aq_b+\order{1}\,, 
\end{aligned} \]
where $'\equiv\partial_x=\partial_\xi\,$. We then have
\[ 
    \Lambda_a\Lambda_b - (c^\star)^d{}_{ab}\Lambda_d = \tfrac{1}{\xi^2} \ell^2q_aq_b+ \tfrac{1}{\xi}N_{ab} + \order{1} \,,
\] 
for some function \(N_{ab}\) holomorphic in the base variables. Furthermore, 
\[ 
    \tfrac{1}{\Lambda'} = \tfrac{1}{\ell}\,\xi \left( 1-\tfrac{1}{\ell}B\,\xi+\order{\xi^2} \right)\,. 
\]
It follows that
\[
\partial_x \left[ \tfrac{1}{\Lambda'} \left( \Lambda_a\Lambda_b - (c^\star)^d{}_{ab}\Lambda_d \right) \right] = -\tfrac{1}{\xi^2}\, \ell q_aq_b+ \order{1}.
\]
Along with the expansion for $\Lambda_{ab}\,$, this proves that $K_{ab}=\order{1}$ at $\xi=0\,$.
\end{proof}
\begin{lem}\label{lem:lemma3}
Each of the functions $K_{ab}$ defined in \cref{lem:lemma1} is holomorphic at the poles of $\lambda\,$.
\end{lem}
\begin{proof}
Let \(x_0=x_0(\bm z)\) be a locally defined pole section of \(\lambda\,\), of order \(\ell\). We distinguish according to whether the primary differential is regular or singular at \(x_0\,\). 

The former case is formally analogous to the one discussed in \cref{lem:holozero}. In particular, $x$ is a local coordinate in a neighbourhood of $x_0\,$, and one can therefore introduce a coordinate $\xi:=x-r(\bm z)\equiv x-x(x_0(\bm z))$ centered at $x_0\,$. The fact that $x_0$ is now a pole of $\lambda\,$, rather than a zero, only changes the sign of the first term in:
\[
\Lambda=-\ell\,\log\xi+A(\bm z)+B(\bm z)\,\xi+\order{\xi^2}\,.
\]
This changes some of the signs in the later derivatives, but the same cancellations occur.

The case in which $x_0$ is also a pole of $\omega$ is more interesting. Choose a natural local parameter \(\zeta\) centred at \(x_0\, \), in which 
\[ 
    \lambda = c(\bm z)\,\zeta^{-\ell} \bigl(1+\order{\zeta}\bigr), \qquad c(\bm z)\neq0\,. 
\]
Suppose, first, that $\omega$ has a simple pole at $x_0\,$, with residue $r\neq 0$ in the local parameter $\zeta\,$. As already noted in \cite{Dubrovin1996}, and e.g. made explicit in a more general formalism in \cite{Rej23}, the principal part of the primary differential at the poles of $\lambda$ is independent of the moduli point. This gives:
\[
\zeta=e^{x\,/ r}\bigl(a(\bm z)+\order{e^{x/r}}\bigr)\,.
\]
As a consequence:
\[ \Lambda = -\tfrac{\ell}{r}\,x + A(\bm z) + \order{e^{x/r}}\,. \]
It follows that 
\[ 
\begin{aligned} 
    \Lambda' &= -\tfrac{\ell}{r} + \order{e^{x/r}}\,,&&& 
    \Lambda_a &= A_a+\order{e^{x/r}}\,,&&&
    \Lambda_{ab} &= A_{ab}+\order{e^{x/r}}\,. 
\end{aligned} 
\] 
Thus \(\Lambda_{ab}\) is holomorphic in \(\zeta\,\), and so is 
\( \partial_x \left[ \frac{1}{\Lambda'} \left( \Lambda_a\Lambda_b - (c^\star)^d{}_{ab}\Lambda_d \right) \right]\,. \) Therefore \(K_{ab}\) is holomorphic at \(x_0\,\).

If $\omega$ has a pole of order $m+1$ at $x_0\,$, let $r\neq 0$ still denote its leading principal Laurent coefficient. Integrating gives now the following relation:
\[ \zeta = r^{1/m}\,x^{-1/m} \bigl(1+\order{x^{-1/m}}\bigr). \]
Consequently, 
\[ 
\Lambda = \tfrac{\ell}{m}\log x + A(\bm z) + \order{x^{-1/m}}\,.
\] 
Differentiating, we obtain 
\[ 
\begin{aligned} 
    \Lambda' &= \tfrac{\ell}{m}\tfrac1p + \order{x^{-1-1/m}}\,,&&& \Lambda_a &= A_a+\order{x^{-1/m}}\,,&&&\Lambda_{ab} &= A_{ab}+\order{x^{-1/m}}\,. 
\end{aligned} \] 
In particular, \( \Lambda_a\Lambda_b - (c^\star)^d{}_{ab}\Lambda_d = \order{1}\,, \) and hence
\[ 
    \partial_x\bigl[\tfrac{1}{\Lambda'} \left( \Lambda_a\Lambda_b - (c^\star)^d{}_{ab}\Lambda_d \right)\bigr] = \order{1}\,. 
\]
Re-expressed in the local parameter \(\zeta\,\), neither this term nor \(\Lambda_{ab}\) has negative powers of \(\zeta\,\). It follows that \(K_{ab}\) extends holomorphically across \(x_0\,\). This concludes the proof.
\end{proof}
\begin{lem}\label{lem:lemma2}
Let $\{\gamma_1,\dots,\gamma_n\}$ be a basis of $\Lambda^\ast(z)$, as defined in \cite[Section 4.1 and Proposition 3]{MR2429320}, and assume that the twisted periods:
\begin{equation}\label{eq:twistedperiods}
    \widetilde{w}_k:=\int_{\gamma_k}\lambda^z\omega\,,\qquad k=1,\dots, n
\end{equation}
give a system of flat coordinates for the deformed dual connection ${}^z\nabla:={}^g\nabla+z\,\star$. Then, for any $\gamma\in \Lambda^\ast(z)$ and for any $a\leq b\in\{1,\dots,n\}$:
      \[
      \int_{\gamma}K_{ab}\,\omega=0\,.
      \]
\end{lem}
\begin{proof}
By assumption, the twisted periods satisfy:
    \[
    \begin{aligned}
        0&=\widetilde{w}_{k,ab}-z\,g^{cd}\,\Frst_{abc}\,\widetilde{w}_{k,d}\\
&=\int_{\gamma_k}\bigl[z(\log\lambda)_{ab}+z^2(\log\lambda)_a(\log\lambda)_b-z^2g^{cd}\Frst_{abc}(\log\lambda)_d\bigr]\,\lambda^z\omega\\
        &=z\int_{\gamma_k}K_{ab}\,\lambda^z\omega\,.
    \end{aligned}
    \]
Since the equality is true for any $z$, it follows that the integral must vanish. Setting, then, $z=0$ proves the statement.
\end{proof}
\begin{oss}
   In particular, we recall that $\Lambda^\ast(z):=H_1(C_g^z\,,\,\partial C_g^z)\,$, where $C_g^z$ is the manifold with boundary obtained as follows. Firstly, we remove from $C_g$ the interior of $k+1$ disks $D_0,\dots, D_k$ respectively centered at the pole $\infty_0,\dots, \infty_k$ of the superpotential. Then, one replaces each boundary component of such manifold -- i.e. each circle $\partial D_0,\dots, \partial D_k$ -- with the $n_k+1$ arcs of circle determined by the condition $\arg\{z\lambda\}\in \bigl(\tfrac{\pi}{2}\,,\,\tfrac{3}{2}\pi\bigl)\,$.
\end{oss}
\begin{oss}
    If, rather, one has the following twisted periods as flat coordinates of the deformed dual connection ${}^z\nabla$: $\widetilde{w}_k=\int_{\gamma_k}\lambda^{z+\zeta}\omega$ for some $\zeta\in\CC\,$, the result will still hold. As a matter of fact, a similar calculation gives:
    \[
    \begin{aligned}
        0&=\widetilde{w}_{k,ab}-z\,g^{cd}\,\Frst_{abc}\,\widetilde{w}_{k,d}\\
        &=(z+\zeta)\int_{\gamma_k}K_{ab}\,\lambda^{z+\zeta}\omega-\zeta\, g^{cd}\Frst_{abc}\int_{\gamma_k}\tfrac{\dd}{\dd{x}}\bigl(\tfrac{\lambda_d}{\lambda'}\bigr)\,\,\lambda^{z+\zeta}\omega\,.
    \end{aligned}
    \]
    Since this must hold true identically in $z$, it follows that the two integrals vanish independently. This is what happens in the examples in \cite[Section 5.1]{Dub04} when $N\neq 1$, i.e. for the spaces of miniversal unfoldings of surface singularities of type $D$ and $E$. We shall discuss these two cases in more details in \cref{ss:typeDE}.
\end{oss}

\begin{thm}\label{thm::main}
Let $\Hw_{g;\bm{n}}^\omega$ be a Hurwitz Frobenius manifold of dimension $n$ with Landau-Ginzburg superpotential $\lambda$ and primary form $\omega$, $\Hw_{g;\bm{n}}^\star$ denote its Dubrovin dual.

If the twisted periods \cref{eq:twistedperiods} give a set of flat coordinates for the deformed dual connection ${}^z\nabla:={}^g\nabla+z\,\ast\,$, then there exists a unique almost-flat F-manifold structure $\UC_{g;\bm{n}}^\star$, on an open subset of the universal curve $\UC_{g;\bm{n}}$, that extends $\Hw_{g;\bm{n}}^\star$ in such a way that:
\begin{itemize}
    \item if $w_1,\dots, w_n$ is a system of flat coordinates for $\Hw_{g,\bm{n}}^\star$ on $U$, then let $x$ be a coordinate on the fibres over the points of $U$ such that $\omega=\dd{x}$. Then $w_1,\dots, w_n,x$ is an adapted system of flat coordinates on an open subset of $\pi^{-1}(U)$ in $\UC^\star_{g;\bm{n}}$;
    \item in the given adapted system, the extended prepotential $\Omega$ satisfies $\Omega'=\log\lambda$, where $'\equiv\partial_x$.
\end{itemize}
\end{thm}
\begin{proof}
From \cref{lem:lemma1} and \cref{lem:lemma3}, we find that each of the functions $K_{ab}$ are holomorphic and bounded on a closed Riemann surface, and therefore they are constant. In particular, since $\omega\neq 0$, there will be some $\gamma\in\Lambda^\ast(z)$ such that $\int_{\gamma}\omega\neq 0$. As a consequence, such a constant needs to be zero due to \cref{lem:lemma2}. In other words, we have that:
\begin{equation}\label{eq:conditionAlcextension}
    \pdv[2]{\log\lambda}{w^a}{w^b}=\dv{x}\biggl[\frac{1}{(\log\lambda)'}\biggl(\pdv{\log\lambda}{w^a}\pdv{\log\lambda}{w^b}-g^{cd}\,\Frst_{abc}\pdv{\log\lambda}{w^d}\biggr)\biggr]\,.
\end{equation}
The statement follows from \cite[Theorem 3]{alcoladophd}.
\end{proof}
\begin{oss}
    It may seem striking at first that the flat coordinate along the fibres is the same for the extension of the Frobenius manifold structure and for its Dubrovin's dual. However, it is quite natural if one recalls that the intersection form $g$ is given, in terms of the superpotential, by the same expression that gives $\eta$, after the replacement of $\lambda$ with $\log\lambda$. In particular, the primary form remains unchanged.
\end{oss}
\begin{oss}
The condition that $\omega=\dd{x}$ may also seem too restrictive, as one is used to flat coordinates being defined up to affine transformations. Indeed, one can slightly loosen it up to make up for this fact. Here, we should be looking at affine transformations that preserve the splitting into coordinates along the base space and the fibre, hence we want to allow affine transformations of $x$ of the form $x\mapsto a\,x+b$. Clearly, the differential is blind to shifts, i.e. $\dd{x}=\dd{(x+b)}$, hence we only need to discuss rescaling. Let, then, $y=ax$ for some $a\in \CC$. As a function of $x$, we know that $\Omega$ satisfies the open WDVV equations. Using a straightforward change of variables, one has:
\[
\begin{aligned}
    \Omega_{xx}\Omega_{\alpha\beta}&=\Omega_{\alpha x}\Omega_{\beta x}-c\indices{^\mu_{\alpha\beta}}\Omega_{\alpha \mu}\\
    (a\Omega)_{yy}(a\Omega)_{\alpha\beta}&=(a\Omega)_{\alpha y}(a\Omega)_{\beta y}-c\indices{^\mu_{\alpha\beta}}(a\Omega)_{\mu y}\,.
\end{aligned}
\]
As a consequence, $a\Omega$ satisfies the open WDVV equations in the new set of variables. Equivalently, in terms of the Landau-Ginzburg model, if $\omega=a\dd{x}$, the result of \cref{thm::main} still holds true -- so, in particular, $x$ is still a flat coordinate -- but the extended prepotential will satisfy $\Omega_x=a\,\log\lambda$ instead. This will be useful in some of the examples below.
\end{oss}
\subsection{Twisted periods and homogeneity of the superpotential}
In \cref{thm::main}, it is assumed that the twisted periods \cref{eq:twistedperiods} are flat coordinate for the deformed dual connection, in order to have the extension of the dual structure to the universal curve. We shall, now, discuss a sufficient condition for that to hold.

As a starting point, in \cite[Proposition 5.1]{Dub04}, the case of the Saito Frobenius manifolds on the space of parameters of miniversal deformations of isolated surface singularities of type $ADE$ is discussed. In particular, the twisted periods are proven to be flat coordinates for the deformed dual connection as a consequence of the fact that the miniversal deformation $f:\CC^2\times \CC^\ell\to \CC$ itself satisfies two identities. The first one is the following: 
    \begin{equation}\label{eq:openwdvvf}
    f_\alpha f_\beta=\eta^{\mu\nu}\Fr_{\mu\alpha\beta}f_\nu+f_z K\indices{^z_{\alpha\beta}}+f_w K\indices{^w_{\alpha\beta}}\,,
    \end{equation}
    where $\partial_zK\indices{^z_{\alpha\beta}}+\partial_w K\indices{^w_{\alpha\beta}}=f_{\alpha\beta}$, in flat coordinates $\bm{v}\in\CC^\ell$ of $\eta$ \cite{tftsperiodintegrals}. This identity is a consequence of the fact that the algebra structure on the tangent spaces is constructed from deformations of the Jacobian ring of the corresponding surface singularity, see also \cite{Mil17}. Eq. \ref{eq:openwdvvf} will, therefore, hold whenever $f$ is an admissible polynomial in the sense of Fan-Jarvis-Ruan-Witten theory.
    
    As described in \cites{topstringsd<1,LW96,DubLiuZhangDEspotentials}, Saito Frobenius manifold of type $ADE$ all admit a Landau-Ginzburg description with superpotential $\lambda$ given by restricting the miniversal deformation to the curve $f_w=0\,$. See \cref{ss:typeDE} for further details. Therefore, in all these cases, the superpotential satisfies:
    \begin{equation}\label{eq:openwdvvlambda}
        \lambda_\alpha\lambda_\beta=\eta^{\mu\nu}\Fr_{\mu\alpha\beta}\lambda_\nu+K_{\alpha\beta}\,\lambda'\,,\qquad K_{\alpha\beta}'=\lambda_{\alpha\beta}\,.
    \end{equation}
    
    Now, if we define $\Omega$ by $\Omega'=\lambda$, then it is clear that the previous equation is equivalent to $\Omega$ being a solution to the second set of open WDVV equations in \cref{eq:openWDVV}. Equivalently, solving for $K_{\alpha\beta}$ and then taking the $x$ derivative of both sides gives that $\lambda$ satisfies the condition in \cref{thm:alcolado}. As a consequence, \cref{eq:openwdvvlambda} is equivalent to the existence of a flat F-manifold on the universal curve $\UC_{g;\bm{n}}$ extending the given Hurwitz Frobenius manifold.
    
    On the other hand, in \cite{dealmeida2025openhurwitzflatf}, the existence of the rank-one extension is proven starting from the twisted period representation for the flat coordinates of the deformed connection ${}^\eta\nabla+z\,\bullet\,$. In particular, it is shown that this implies \cref{eq:openwdvvlambda}. A key point in the proof is the identity:
 \begin{equation}\label{eq:twistedperiodsareflat}
         \pdv[2]{}{v_\alpha}{v_\beta}\int_\gamma e^{z\lambda}\omega-\eta^{\mu\nu}c_{\mu\alpha\beta}\pdv{v_\nu}\int_\gamma e^{z\lambda}\omega=z\int_\gamma\biggl[\lambda_{\alpha\beta}-\dv{x}\biggl(\frac{\lambda_\alpha\lambda_\beta-\eta^{\mu\nu}c_{\mu\alpha\beta}\lambda_\nu}{\lambda'}\biggr)\biggr]\,e^{z\lambda}\omega\,.
 \end{equation}
The point is, then, proving that the integrand is constant by showing that it is holomorphic. 

Vice versa, if \cref{eq:openwdvvlambda}, then clearly the right-hand side of \cref{eq:twistedperiodsareflat} vanishes, and, therefore, the twisted periods are a solution to the Gauss-Manin system for the deformed connection.

 The two statements are, therefore, equivalent. The superpotential $\lambda$ of a Landau-Ginzburg model $(\lambda,\omega)$ satisfies \cref{eq:openwdvvlambda} -- in a system of flat coordinates $\bm{v}$ for the flat metric it defines by Grothendieck residues, \cref{eq:etaLGmodel} -- if and only if the corresponding twisted periods are a system of flat coordinates for the deformed connection spanned by the very same metric.

 As for the dual structure, on the other hand, the same result will hold: the twisted periods \cref{eq:twistedperiods} are flat coordinates for the deformed dual connection ${}^g\nabla+z\,\ast$ if and only if $\log\lambda$ satisfies \cref{eq:openwdvvlambda} in a system of flat coordinates of $g$ -- where, crucially, the structure constants are the ones of the dual multiplication. It is natural to, then, ask what is the relation between the twisted period representations of the flat coordinates of the two pencils, which is equivalent to determining when both structures have an extension of rank one to the universal curve. Already in the case of $ADE$ singularities, however, as discussed in \cite[Section 5.1]{Dub04}, \cref{eq:openwdvvlambda} must be supplemented by a second identity -- which is the explicit expression of $\Lie_E\lambda\,$, \cite[Eq. 5.8]{Dub04} -- to prove the corresponding twisted periods give flat coordinates for the deformed dual connection. 
 
We start our discussion by noticing that the computations in \cite[Proposition 5.1]{Dub04} do not depend on the explicit example considered:
 \begin{prop}\label{prop:twistedperiodssufficient}
     Let $\bm{v}$ be a system of flat coordinates on an open subset $U$ the Hurwitz Frobenius manifold $\Hw_{g\,;\,\bm{n}}$ and $x$ a coordinate along the fibres over the points in $U$ of the universal curve such that $\omega=\dd{x}\,$. If the Landau-Ginzburg superpotential $\lambda$ satisfies:
     \begin{align}
     \lambda_\alpha\,\lambda_\beta&=\sum_{\mu=1}^n c\indices{^\mu_{\alpha\beta}}\lambda_\mu+K_{\alpha\beta}\,\lambda_x\,,\\
      \label{eq:LieElambda}   \Lie_E\lambda&=\lambda-\bigl[\tfrac12(1-d)x+d_0\bigr]\,\lambda_x\,,
     \end{align}
     for some $K_{\alpha\beta}$ such that $\partial_xK_{\alpha\beta}=\lambda_{\alpha\beta}\,$ and $d_0\in\CC$, then the twisted periods \cref{eq:twistedperiods} give a system of flat coordinates for the deformed dual connection ${}^g\nabla+z\,\ast\,$.
 \end{prop}
 \begin{proof}
     The proof is formally the same as \cite[Proposition 5.1]{Dub04}. In particular, we want to show that it satisfies \cite[eq. 3.10]{Dub04}:
   \[
    \begin{aligned}
       \mathcal{U}\indices{^\nu_\beta}\,\widetilde{w}_{\alpha\nu}=\bigl(z+\tfrac{1-d}{2}\bigr)\,c\indices{^\nu_{\alpha\beta}}\,\widetilde{w}_\nu-\partial_\nu E^\mu \,\widetilde{w}_\mu\,c\indices{^\nu_{\alpha\beta}}\,,
    \end{aligned}
    \]
   where $\mathcal{U}$ is the matrix representing the endomorphism of the tangent bundle $X\mapsto E\bullet X$ in the system of flat coordinates $\bm{v}\,$.
   
Using \cref{eq:openwdvvlambda}, we obtain:
\[
\widetilde{w}_{\alpha\beta}=z(z-1)\,c\indices{^\mu_{\alpha\beta}}\,\int_\gamma\lambda_\mu\,\lambda^{z-2}\dd{x}\,.
\]
As a consequence, using associativity:
\[
\begin{aligned}
   \mathcal{U}\indices{^\nu_\beta}\,\widetilde{w}_{\alpha\nu}&=z(z-1) c\indices{^\nu_{\alpha\beta}}\int_\gamma E^\mu\,c\indices{^\rho_{\mu\nu}} \lambda_\rho \,\lambda^{z-2}\,\dd{x}\,.
\end{aligned}
\]
Now, multiplying \cref{eq:LieElambda} by $\lambda_\nu$ and using \cref{eq:openwdvvlambda}, we get the identity:
\[
\lambda\,\lambda_\nu=E^\mu c\indices{^\rho_{\mu\nu}}\lambda_\rho+\bigl[E^\mu K_{\mu\nu}+\bigl(\tfrac{1-d}{2}x+d_0\bigr)\lambda_\nu\bigr]\lambda_x\,.
\] Plugging this into the previous equation gives:
\[
\begin{aligned}
     \mathcal{U}\indices{^\nu_\beta}\,\widetilde{w}_{\alpha\nu}&=(z-1)\,c\indices{^\nu_{\alpha\beta}}\,\widetilde{w}_{\nu}-z\,c\indices{^\nu_{\alpha\beta}}\int_\gamma\bigl[E^\mu K_{\mu\nu}+\bigl(\tfrac{1-d}{2}x+d_0\bigr)\lambda_\nu\bigr]\,(\lambda^{z-1})_x\,\dd{x}\\
     &=(z-1)\,c\indices{^\nu_{\alpha\beta}}\,\widetilde{w}_{\nu}+z\,c\indices{^\nu_{\alpha\beta}}\int_\gamma\bigl[E^\mu \lambda_{\mu\nu}+\tfrac{1-d}{2}\lambda_\nu+\bigl(\tfrac{1-d}{2}x+d_0\bigr)\lambda_{\nu x}\bigr]\,\lambda^{z-1}\,\dd{x}\,,
\end{aligned}
\]
where we used $\partial_xK_{\alpha\beta}=\lambda_{\alpha\beta}\,$.
On the other hand, differentiating \cref{eq:LieElambda} along $v^\nu$ gives:
\[
\lambda_\nu-\partial_\nu E^\mu\,\lambda_\mu=E^\mu \lambda_{\mu\nu}+\bigl(\tfrac{1-d}{2}x+d_0\bigr)\lambda_{\nu x}\,.
\]
As a consequence, we finally get:
\[
\begin{aligned}
    \mathcal{U}\indices{^\nu_\beta}\,\widetilde{w}_{\alpha\nu}&=(z-1)\,c\indices{^\nu_{\alpha\beta}}\,\widetilde{w}_{\nu}+\bigl(1+\tfrac{1-d}{2}\bigr)\,c\indices{^\nu_{\alpha\beta}}\widetilde{w}_\nu-\partial_\nu E^\mu\,c\indices{^\nu_{\alpha\beta}}\widetilde{w}_\mu\\
   &=\bigl(z+\tfrac{1-d}{2}\bigr)\,c\indices{^\nu_{\alpha\beta}}\,\widetilde{w}_\nu-\partial_\nu E^\mu \,\widetilde{w}_\mu\,c\indices{^\nu_{\alpha\beta}}\,.
\end{aligned}
\]
 \end{proof}
 \begin{oss}
     $d_0$ can be set to zero by a change of coordinate on the curve whenever $d\neq 1\,$.
 \end{oss}
In particular, the first identity in the previous Proposition is equivalent to the analogous twisted period representation for the flat coordinates of ${}^\eta\nabla+z\,\bullet\,$, as we have just discussed. We, therefore, focus on the second one, starting from the following:
\begin{lem}
    Let $\bm{v}$ be flat coordinates for the Frobenius metric and $\widetilde{v}(\bm{v},z)$ be a solution to the Gauss-Manin system for the deformed connection ${}^\eta\nabla+z\,\bullet$ normalised so that:
\begin{equation}\label{eq:normalisationtildet}
    z\partial_z\widetilde{v}-\Lie_E\widetilde{v}=\bigl(\tfrac{d}{2}-1\bigr)\,\widetilde{v}\,.
\end{equation} 
    Then, the Laplace transform:
    \begin{equation}\label{eq:xtildeLaplace}
        \widetilde{x}(\bm{v},\lambda):=\int_\mathfrak{C}\tfrac{1}{\sqrt{z}}\widetilde{v}(\bm{v},z)\,e^{-z\lambda}\dd{z}\,,
    \end{equation}
    $\mathfrak{C}$ being a Hankel contour in $\RS\,$, is a solution to the Gauss-Manin system for the pencil of cometrics $g^\sharp-\lambda\,\eta^\sharp\,$.
\end{lem}
\begin{proof}
    This is a rephrasing of \cite[Proposition H.1]{Dubrovin1996} with a different normalisation for $\widetilde{v}$. In particular, the main point is that a similar identity as the one in \cite[Lemma H.3]{Dubrovin1996} still holds true, but with a different constant. In particular, it is:
    \[
  z\,  g^\sharp\bigl(\dd{v}_\alpha\,,\,\sum_\beta \dd{v}_\beta \,{}^g\nabla_\beta \dd\widetilde{v}\bigr)=\dd{\bigl(\eta^\sharp\bigl(\dd{v}_\alpha\,,\,z\partial_z\dd{\widetilde{v}}-\tfrac12\dd{\widetilde{v}}\bigr)\bigr)}\,,
    \]
    This is the reason for the different power of $z$ in the Laplace transform compared to \cite[Proposition H.1]{Dubrovin1996}.
\end{proof}
The reason for the different normalisation is the following:
\begin{lem}[\cite{Dubrovin1999}]
     Let $\bm{v}$ be flat coordinates for the Frobenius metric and $\widetilde{v}(\bm{v},z)$ be a solution to the Gauss-Manin system for the deformed connection ${}^\eta\nabla+z\,\bullet\,$. If $\widetilde{v}$ satisfies \cref{eq:normalisationtildet}, then it is a solution to the Gauss-Manin system for the extended Dubrovin connection.
\end{lem}
\begin{proof}
   This is a remark in the proof of \cite[Theorem 2.3]{Dubrovin1999}. In particular, differentiating \cref{eq:normalisationtildet} with respect to $t_\alpha$ and using the Gauss-Manin equation for the deformed connection, we have:
    \[
    \begin{aligned}
        z\,\partial_z\partial_\alpha\widetilde{v}&=\partial_\alpha E^\mu\,\partial_\mu\widetilde{v}+  E^\mu \partial_\mu\partial_\alpha \widetilde{v}+\bigl(\tfrac{d}{2}-1\bigr)\partial_\alpha\widetilde{v}\\
        &=zE^\mu c\indices{^\nu_{\mu\alpha}}\partial_\nu\widetilde{v}+\partial_\alpha E^\mu\,\partial_\mu\widetilde{v}-\bigl(1-\tfrac{d}{2}\bigr)\,\partial_\alpha \widetilde{v}\,.
        \end{aligned}
    \]
    The latter corresponds to $\widehat{\nabla}_{\partial_z}\dd\widetilde{v}=0\,$.

\end{proof}
\begin{cor}\label{cor:signmistake}
    The solution $\widetilde{x}(\bm{v},\lambda)$ to the Gauss-Manin system for the flat pencil of cometrics $g^\sharp-\lambda\,\eta^\sharp$ in \cref{eq:xtildeLaplace} satisfies:
    \begin{equation}\label{eq:homogeneityxtilde}
        \lambda\partial_\lambda\widetilde{x}+\Lie_E\widetilde{x}=\tfrac12(1-d)\,\widetilde{x}\,.
    \end{equation}
\end{cor}
\begin{proof}
    This is an easy exercise with Laplace transforms:
    \[
    \begin{aligned}
        \lambda\partial_\lambda\widetilde{x}+\Lie_E\widetilde{x}&=\int_\mathfrak{C}\sqrt{z}\,\widetilde{v}\,(-\lambda)e^{-z\lambda}\dd{z}+\int_\mathfrak{C}\tfrac{1}{\sqrt{z}}(\Lie_E\widetilde{v})\,e^{-z\lambda}\dd{z}\\
        &=-\tfrac12\int_\mathfrak{C}\tfrac{1}{\sqrt{z}}\,\widetilde{v}\,e^{-z\lambda}\dd{z}-\int_\mathfrak{C}\tfrac{1}{\sqrt{z}}\bigl(z\partial_z\widetilde{v}-\Lie_E\widetilde{v}\bigr)\,e^{-z\lambda}\dd{z}\\
        &=\tfrac12(1-d)\,\widetilde{x}\,.
    \end{aligned}
    \]
\end{proof}
The relation between these results and the superpotential comes from Dubrovin's construction of a superpotential for any semi-simple Frobenius manifold in \cite[Appendix I]{Dubrovin1996}. In particular, the previous three results all lead to the following:
\begin{prop}\label{prop:LieElambda}
    If there is a complete set of solutions $\widetilde{v}_1,\dots, \widetilde{v}_n$ to the Gauss-Manin system for the deformed connection satisfying the weighted-homogeneity condition in \cref{eq:normalisationtildet}, then Dubrovin's Landau-Ginzburg superpotential $\lambda$ satisfies \cref{eq:LieElambda}.
\end{prop}
\begin{proof}
The construction in \cite[Appendix I]{Dubrovin1996} this relies on a system of flat coordinates $\widetilde{x}_1(\bm{u},\lambda)\,$,$\dots\,$, $\widetilde{x}_n(\bm{u},\lambda)$ for the pencil of flat metrics $g^\sharp-\lambda \eta^\sharp\,$, previously computed in \cite[Appendix H]{Dubrovin1996} -- $\bm{u}$ being canonical coordinates -- satisfying the same weighted-homogeneity condition as in \cref{eq:homogeneityxtilde}; see \cite[Eq. H.19]{Dubrovin1996}.

In particular, $\lambda(x,\bm{u})$ is the inverse function defined by \cite[Eq. I.10]{Dubrovin1996}:
\[
x-\widetilde{x}_1(\bm{u},\lambda)-\dots-\widetilde{x}_n(\bm{u},\lambda)=0\,.
\]
Notice that $x$ is precisely a coordinate on the curve such that $\omega=\dd{x}\,$, \cite[Theorem I.1]{Dubrovin1996}. 

Since the Euler vector field in canonical coordinates is simply $E=\sum_\mu u_\mu\partial_{u_\mu}\,$. It is, then, an easy exercise with the implicit function Theorem to show that:
\[
\tfrac12(1-d)x\,\lambda_x+\sum_\mu u_\mu\,\lambda_{\mu}=\frac{\tfrac12(1-d)x-\sum_\mu u_\mu x_\mu}{\pdv{x}{\lambda}}=\lambda\,.
\]
This proves the statement.
\end{proof}

Therefore, we finally have the following sufficient condition for the existence of both extensions:
\begin{thm}
    If there exist cycles $\{\gamma_1,\dots, \gamma_n\}$ in $\Lambda^\ast(z)$ such that twisted periods:
    \[
    \widetilde{v}_\alpha:=\int_{\gamma_\alpha}e^{z\lambda}\omega\,,\qquad \alpha=1,\dots, n\,,
    \]
    are flat coordinates for the deformed connection ${}^\eta\nabla+z\,\bullet$ satisfying \cref{eq:normalisationtildet}, then they are also flat coordinates for the extended Dubrovin connection and both the Frobenius manifold and the Dubrovin dual structure admit a rank-one extension to the universal curve, respectively given by the one in \cite{dealmeida2025openhurwitzflatf} and in \cref{thm::main}.
\end{thm}

\subsection{Generalisation to almost-flat Riemannian F-manifolds}
The remark regarding the equivalence of the twisted periods being flat coordinates for the pencil of the deformed connection and the existence of the rank-one extension is actually obviously independent of the fact that we are dealing with Hurwitz Frobenius manifolds. All the concept are well-defined in principle for any F-manifold with compatible flat connection. In particular, we need to make the following:
\begin{defn} Let $D_g$ be an open domain of a closed Riemann surface $C_g$ of genus-$g$ and $M$ be a semi-simple almost-flat Riemannian F-manifold. A genus-$g$ \emph{Landau-Ginzburg model} for $M$ consists of a holomorphic function $\Lambda:\widetilde{M}_g\to\RS$ -- where $\widetilde{M}_g$ is a suitable holomorphic fibre bundle over $M$ whose typical fibre is the universal cover $\widetilde{D}_g$ of $D_g$ -- and an Abelian differential $\omega$ on $\widetilde{M}_g$ satisfying:
    \begin{itemize}
        \item [(LG1)] The critical points of the restriction of \(\Lambda\) to each fibre lie over \(D_g\), and their critical values form canonical coordinates on \(M\).
         \item [(LG2)] The following expressions for the flat metric $\eta$ and the $\eta$-dual tensor field $c$ to the multiplication of $M$ hold:
        \[
        \begin{aligned}
         \eta(X,Y)&=\sum_{x\in \mathrm{Cr}(\Lambda)}\Res_x\{X(\Lambda)Y(\Lambda)\tfrac{\omega^2}{\dd\Lambda}\}\,,\\
         c(X,Y,Z)&=\sum_{x\in \mathrm{Cr}(\Lambda)}\Res_x\{X(\Lambda)Y(\Lambda)Z(\Lambda)\tfrac{\omega^2}{\dd\Lambda}\}\,,
        \end{aligned}
        \]
        where $ \mathrm{Cr}(\Lambda)$ is the set of critical points of $\Lambda$.
        \item [(LG3)] There exist cycles $\gamma_1,\dots, \gamma_n$ in $D_g$ such that a system of flat coordinates $\widetilde{v}_1,\dots,\widetilde{v}_n$ for the deformed connection ${}^z\nabla:={}^\eta\nabla+z\,\bullet$ of $M$ are given by the twisted periods:
        \[
\widetilde{v}_\mu=\int_{\gamma_\mu}e^{z\Lambda}\,\omega\,,\qquad \mu=1,\dots, n\,.
        \]
       \end{itemize}
    \end{defn}
\begin{oss}
    A general procedure endowing almost-flat Riemannian F-manifolds -- other than the ones coming from almost-duality with Frobenius manifolds -- with a Landau-Ginzburg model as just defined is, however, still unknown. The recent construction and study of the Gauss-Manin connections for bi-flat F-manifolds in \cite{arsieGMbiflat} suggests that such procedure might exist, since the Gauss-Manin connection is a crucial tool in the Landau-Ginzburg model of Frobenius manifolds, as described in \cite{Dubrovin1996}.
\end{oss}
 As a consequence of the discussion in the previous subsection, therefore, \cref{thm::main} is an instance of the following more general result:
    \begin{thm}\label{thm:extensiongeneral}
        Let $M$ be an almost-flat Riemannian F-manifold with genus-$g$ Landau-Ginzburg model $(\Lambda, \omega)\,$. $M$ admits a unique rank-one extension to an almost-flat F-manifold structure on $\widetilde{M}_g$ such that:
        \begin{itemize}
            \item If $v_1,\dots, v_n$ is a set of flat coordinates in an open set $U\subseteq M$, then let $x$ be a coordinate on the fibres of $\widetilde{M}_g$ over the points of $U$ such that $\omega=\dd{x}$. Then $v_1,\dots, v_n,x$ is an adapted set of flat coordinates on an open subset of $\pi^{-1}(U)\subseteq \widetilde{M}_g$.
            \item In such system of flat coordinates, the extended prepotential $\Omega$ satisfies $\Omega_x=\Lambda$.
        \end{itemize}
    \end{thm}

\section{Dubrovin duality on the universal curve}\label{sec:duality}

\subsection{Eventual identity and bi-flat F-manifold structure}\label{ss:eventualidentity}
Our result about the existence and uniqueness of the rank-one extension for Dubrovin duality of the Hurwitz Frobenius manifold fits, together with the one in \cite{ dealmeida2025openhurwitzflatf}, in the following diagram:
\begin{equation}\label{eq:almostudalityextension}
    \begin{tikzcd}
        \UC_{g;\bm{n}}\arrow[rr, twoheadrightarrow]\arrow[d, dashrightarrow, red, "?"]&& \Hw_{g;\bm{n}}\arrow[d, "(E\bullet)^{-1}"]\\
        \UC_{g;\bm{n}}^\star \arrow[rr, twoheadrightarrow]&&\Hw^\star_{g;\bm{n}}
    \end{tikzcd}\,,
\end{equation}
A natural question is, therefore, whether one can close up such a diagram by providing a map like the one represented by the red arrow.

Now, generalisations of almost-duality to F-manifolds were studied in \cite{maninflatfmanifolds, davidstrDubduality}. In particular, since F-manifolds lack the choice of an Euler vector field, Dubrovin duality in this case will rely on a so-called eventual identity.
\begin{defn}[\cite{maninflatfmanifolds}]
    Let $M$ be an F-manifold. An \emph{eventual identity} is an invertible vector field $\EE$ on $M$ such that the multiplication $\ast$, defined by:
    \begin{equation}\label{eq:eveidentity}
        X\ast Y:=\EE^{-1}\bullet X\bullet Y\,,\qquad\qquad \forall X,Y\in\X_M\,,
    \end{equation}
    makes $M$ into an F-manifold with identity $\EE$.
\end{defn}
\begin{oss}
    Eventual identities are usually built using the constructive equivalent characterisation in \cite{davidstrDubduality}:
        \begin{equation}
                    \Lie_\EE\widetilde{\bullet}=\comm{e}{\EE}\,\widetilde{\bullet}\,.
        \end{equation}
\end{oss}

We start by stating the following expression for the extended multiplication $\widetilde{\bullet}$ on horizontal vector fields, which is a straightforward consequence of the open WDVV equations:
\begin{lem}\label{lem:multiplicationhorizontal}
    Let $\UC_{g\,;\,\bm{n}}$ be the universal curve over the Hurwitz space $\Hw_{g\,;\,\bm{n}}$ equipped with a primary differential $\omega$. Let $\widetilde{\bullet}$ and $\widetilde{\ast}$ denote respectively the multiplications on $T\UC_{g\,;\,\bm{n}}$ coming from extending the Frobenius product $\bullet$ and its Dubrovin dual $\ast$ on $\Hw_{g\,;\,\bm{n}}$. If $X,Y$ are projectable vector fields on $\UC_{g\,;\,\bm{n}}$, then:
    \[
    \begin{aligned}
            \extprod{X}{Y}&=X\bullet Y+\tfrac{1}{\lambda'}\bigl(\Lie_X\lambda\,\Lie_Y\lambda-\Lie_{X\bullet Y}\lambda\bigr)\,\partial_x\,,\\
    \extdprod{X}{Y}&=X\ast Y+\tfrac{1}{\lambda'}\bigl(\tfrac{1}{\lambda}\Lie_X\lambda\,\Lie_Y\lambda-\Lie_{X\ast Y}\lambda\bigr)\,\partial_x\,.
    \end{aligned}
    \]
\end{lem}
\begin{proof}
    According to \cref{eq:multtableextension} and \cref{eq:openWDVV}:
    \[
    \begin{aligned}
            \extprod{X}{Y}&=X\bullet Y+X^\alpha Y^\beta \Omega_{\alpha\beta}\,\partial_x\\
            &=X\bullet Y+\tfrac{1}{\Omega''}X^\alpha Y^\beta\bigl(\Omega'_\alpha\Omega'_\beta-c\indices{^\mu_{\alpha\beta}}\Omega'_\mu\bigr)\,\partial_x\\
            &=X\bullet Y+\tfrac{1}{\lambda'}\bigl(X^\alpha\lambda_\alpha\,Y^\beta\lambda_\beta-(X\bullet Y)^\mu\lambda_\mu\bigr)\,\partial_x\,.
    \end{aligned}
    \]
The relation for the extended dual product follows similarly by replacing $\Omega$ with $\Omega^\star$ while keeping in mind that $\Omega^\star_x=\log\lambda\,$.
\end{proof}

As a first consequence, it is, then, easy to see that:
\begin{lem}
  The extension $\ee$ of the Frobenius-manifold identity $e$ to the universal curve, as defined in \cref{eq:extensionseE}, is the identity for the rank-one extension to $\UC_{g\,;\,\bm{n}}$ introduced in \cite{dealmeida2025openhurwitzflatf}.
  \end{lem}
\begin{proof}
    This follows from $\Lie_\ee\lambda=1$, see \cref{Hw:prop:liftsgeneratorsomegaconnection}. Using the previous Lemma and the multiplication table in \cref{eq:multtableextension}:
    \[
    \begin{aligned}
        \extprod{\ee}{X}&=\ee\,\bullet \, X+\tfrac{1}{\lambda'}\bigl(\Lie_\ee\lambda\,\Lie_X\lambda-\Lie_{X}\lambda\bigr)\,\partial_x=X\,,\\
       \extprod{\ee}{\partial_x}&=(\Lie_\ee\lambda)\,\partial_x=\partial_x\,,
    \end{aligned}
    \]  for any projectable vector field $X\,$.
\end{proof}
Using this fact, we can start looking at the eventual identity $\EE$ on the universal curve using the following obvious result, which is implicitly already encoded in \cite{alcoladophd}.
\begin{lem}\label{lem:projectioneventualidentity}
Let $\EE$ be the identity of $\widetilde{\ast}$ on $\UC_{g\,;\,\bm{n}}^\star$. If $\EE$ is an eventual identity on the universal curve, and $\EE^{-1}$ is its inverse, then:
\[
\EE-E\,,\EE^{-1}-E^{-1}\in\ker\pi_\ast\,,
\]where $E$ and $E^{-1}$ denote the horizontal lifts with respect to the Ehresmann connection induced by $\omega\,$.
\end{lem}
\begin{proof}
    We specify \cref{eq:eveidentity} on $\UC_{g\,;\,\bm{n}}$ to the case of $X$ and $Y$ projectable. Then, we project the relation to the Hurwitz space point-wise:
    \[
    \pi_{*p}(\EE^{-1})_p\bullet X_{\pi(p)} \bullet  Y_{\pi(p)} =X_{\pi(p)} \ast  Y_{\pi(p)} =(E^{-1})_{\pi(p)} \bullet X_{\pi(p)} \bullet Y_{\pi(p)} \,.
    \]
    Taking, then $X=Y=e\circ\pi$, gives $\pi_{*p}(\EE^{-1})_p=E^{-1}_{\pi(p)}$. The other relation follows from projecting $\EE\,\,\widetilde{\bullet }\,\,\EE^{-1}=\ee\,$.
\end{proof}

As for the component along the fibres, we have:
\begin{prop}\label{prop:eventualidentityhere}
        Assume that the vertical dashed red arrow in \cref{eq:almostudalityextension} exists, and it is given by an eventual identity $\EE$ on $\UC_{g;\bm{n}}$ as in \cref{eq:eveidentity}. Then:
        \[
        \EE=E\circ\pi+\tfrac{1}{\lambda'}\bigl(\lambda-\Lie_E\lambda\bigr)\,\partial_x\,.
        \]
\end{prop}
\begin{proof}
According to the previous Lemma and the splitting of $T\UC_{g\,;\,\bm{n}}$ induced by $\pi$, $\EE=E\circ\pi+f\,\partial_x$ for some $f\in\Hol_{\UC_{g;\bm{n}}}$. We shall from now on omit composition by $\pi$ to the right of $E$ for the sake of clarity.
    We fix $f$ so that the multiplication $\widetilde{\ast}$ satisfies the multiplication table given in \cref{eq:multtableextension} with extended prepotential $\Omega^\star$ as in \cref{thm::main}. We start with the following:
    \[
    \begin{aligned}
\partial_x\,\,\widetilde{\bullet}\,\,\partial_x&=\EE\,\,\widetilde{\bullet}\,\,\partial_x\,\,\widetilde{\ast}\,\,\partial_x\\
\Omega_{xx}\,\partial_x&=\Omega^\star_{xx}\,\bigl(E\,\,\widetilde{\bullet}\,\,\partial_x+f\,\partial_x\,\,\widetilde{\bullet}\,\,\partial_x\bigr)\\
&=\Omega^\star_{xx}\bigl[(\Lie_E\Omega_x)+f\,\Omega_{xx}\bigr]\partial_x\,.
    \end{aligned}
    \]
    This gives:
    \[
    \begin{aligned}
        f&=\frac{1}{\Omega^*_{xx}}-\frac{1}{\Omega_{xx}}\,\Lie_E\Omega_x=\tfrac{1}{\lambda'}\bigl[\lambda-\Lie_E\lambda\bigr]\,.
    \end{aligned}
    \]
We, then, need to show that this choice for $f$ is consistent with the remaining relations in the multiplication table in \cref{eq:multtableextension}. Hence, for any projectable vector field $X$, we want to check:
\[
    \begin{aligned}
X\,\,\widetilde{\bullet}\,\,\partial_x&=\EE\,\,\widetilde{\bullet}\,\,X\,\,\widetilde{\ast}\,\,\partial_x\\
(\Lie_X\lambda)\,\partial_x&=(\Lie_X\log\lambda)\,\EE\,\,\widetilde{\bullet}\,\,\partial_x\\
&=\tfrac{1}{\lambda}(\Lie_X\lambda)\bigl[\Lie_E\lambda+\lambda' \,f\bigr]\partial_x=(\Lie_X\lambda)\,\partial_x\,.\\
\extdprod{\EE}{X}&=\extdprod{E}{X}+f\,\extdprod{\partial_x}{X}\\
&=E\,\ast\,X+\tfrac{\lambda}{\lambda'}\bigl(\tfrac{1}{\lambda^2}\,\Lie_E\lambda\,\Lie_X\lambda-\tfrac{1}{\lambda}\,\Lie_{E\ast X}\lambda\bigr)\partial_x+\tfrac{1}{\lambda}f\,(\Lie_X\lambda)\,\partial_x\\
&=X+\tfrac{1}{\lambda\lambda'}\bigl(\Lie_E\lambda-\lambda+\lambda' f\bigr)\,(\Lie_X\lambda)\,\partial_x=X\,,
    \end{aligned}
    \]
where we used the dual version of the formula in \cref{lem:multiplicationhorizontal}.
\end{proof}
\begin{cor}\label{cor:EvIdAffine}
  The generator $\EEE$ of the rescalings on the universal curve, as defined in \cref{eq:extensionseE}, coincides with the eventual identity $\EE$ for the open Hurwitz flat F-manifold: $\EEE=\EE\,$. Moreover, it is an affine vector field with respect to $\widetilde{\nabla}\,$. Equivalently, $\EEE$ preserves $\widetilde{\nabla}\,$, i.e.:
    \[
    \Lie_\EEE\widetilde{\nabla}=0\,.
    \]
    \end{cor}
\begin{proof}
   The first part of the statement follows from the expression for $\EEE$ with respect to the $\omega$-lift of $E\,$, see \cref{Hw:prop:liftsgeneratorsomegaconnection}.

For the affine-ness of $\EEE\,$, we write, in view of \cref{lem:projectioneventualidentity}:
\[
\EEE=E+\EEE(x)\,\partial_x\,.
\]
Now, since the underlying structure on the Hurwitz space is a fully fledged Frobenius manifold, it follows that $\omega=\dd x$ is a homogeneous differential, as defined in \cite{RomanoDoubleHurwitz}. In other words, it is an eigenvector of the operator $\Lie_{\EEE}\,$. Its eigenvalue, which is related to the charge of the Frobenius manifold, is denoted $d_\omega\in\CC\,$. As a consequence:
\[
d_\omega\,\omega=\Lie_\EEE\omega=\Lie_\EEE\dd{x}=\dd(\EEE(x))\,.
\]
This proves that the vertical component of $\EEE$ is linear in $x\,$, which is a flat coordinate for $\widetilde{\nabla}\,$. Hence, $\EEE$ is an affine vector field..
Equivalence between the two formulations comes from the identity:
    \[
   ( \Lie_X\nabla)_YZ={}^\nabla R(X,Y)Z+\nabla^2_{Y,Z}X\,,
    \]
    for any three vector fields $X\,$, $Y$ and $Z$ and for any affine connection $\nabla\,$, as we already noted in the proof of \cref{lem:EVFaffine}.
\end{proof}
\begin{oss}\label{rem:quasihomogOmega}
    In particular, according to \cref{prop:LieElambda}, it is:
    \begin{equation}
        \EE=E+\bigl[\tfrac12(1-d)x+d_0\bigr]\,\partial_x\,.
    \end{equation}
    As a consequence, when $d\neq 1$, the extended prepotential $\Omega$ for the extension of the Frobenius structure is a weighted-homogeneous solution with the same weights as the underlying WDVV solution, while the additional ones are given by:
    \begin{equation}
            \Lie_E\Omega+\tfrac12(1-d)x\,\Omega_x=\tfrac12(3-d)\,\Omega\,,
    \end{equation}
    modulo linear polynomials in the flat coordinates. This coincides with the weighted-homogeneity condition considered in \cites{horev2012opengromovwittenwelschingertheoryblowups,BCT1,BBvirasoro,BuryakOpenAD,PSTInttheory}, thus giving it a geometric meaning. In particular, the weighted-homogeneity of $\Omega$ is a direct consequence of the one of $\lambda\,$.
    \end{oss}

As for the inverse of the eventual identity, it is then an easy computation to show that:
\begin{cor}\label{cor:Einverse}
   The generator $\EEE$ of the rescalings on the universal curve is invertible with respect to the extended product $\widetilde{\bullet}\,$, and it is:
    \[
    \EEE^{-1}=E^{-1}+\tfrac{1}{\lambda'}\bigl(\tfrac{1}{\lambda}-\Lie_{E^{-1}}\lambda\bigr)\,\partial_x\,.
    \]
\end{cor}
\begin{proof}
   According to \cref{lem:projectioneventualidentity}, it is $\EEE^{-1}=E^{-1}+g\,\partial_x$ for some $g\in\Hol_{\UC_{g;\bm{n}}}$. We fix $g$ by requiring that $\EEE^{-1}$ is the inverse of $\EEE$ with respect to $\widetilde{\bullet}$:
   \[
    \begin{aligned}
        \ee&=\EEE\,\,\widetilde{\bullet}\,\,\EEE^{-1}\\
        &=e+\tfrac{1}{\lambda'}\bigl(\Lie_E\lambda\,\Lie_{E^{-1}}\lambda-\Lie_e\lambda\bigr)\,\partial_x+f\,\partial_x\,\,\widetilde{\bullet}\,\,E^{-1}+g\,E\,\,\widetilde{\bullet}\,\,\partial_x+fg\,\partial_x\,\,\widetilde{\bullet}\,\,\partial_x\\
        &=\ee+\tfrac{1}{\lambda'}\bigl[\Lie_{E^{-1}}\lambda\,(\Lie_E\lambda+\lambda'f)-1+\lambda'g(\Lie_E\lambda+\lambda'f)\bigr]\,\partial_x\\
        &=\ee+\tfrac{1}{\lambda'}\bigl[\lambda\,\Lie_{E^{-1}}\lambda-1+\lambda\lambda' g\bigr]\,\partial_x\,.
    \end{aligned}
    \]
Setting the coefficient of $\partial_x$ to zero proves the statement. 
\end{proof}
\begin{oss}
    Again, we notice that, since $E^{-1}$ is the $\bullet$-inverse of $E$, it must be $\Lie_{E^{-1}}\lambda=\lambda^{-1}+<\,\lambda'\,>$. Therefore, the vertical component of $\EEE^{-1}$ measures how far $\Lie_{E^{-1}}\lambda$ is from $\lambda^{-1}\,$. In particular, it, then, is $\Lie_{\EEE^{-1}}\lambda=\lambda^{-1}\,$.
\end{oss}

For consistency, we also check that:
    \begin{prop}\label{prop:homogeneityE}
$\EEE$ is an Euler vector field of weight one with respect to $\widetilde{\bullet}\,$, i.e.: \[\Lie_{\EEE}\widetilde{\bullet}=\widetilde{\bullet}\,.\]
    \end{prop}
    \begin{proof}
    It is sufficient to check the equality $(\Lie_\EEE\widetilde{\bullet})(X,Y)=\extprod{X}{Y}$ when $X$ and $Y$ are both projectable, when either of them is $\partial_x$ and finally when both of them are $\partial_x$. We start from the last case.
    It is easy to see that $\Lie_X\partial_x=0$ whenever $X$ is a projectable vector field. Then, if we write $\EEE=E+f\,\partial_x$:
        \[
        \begin{aligned}
        \bigl(\Lie_{\EEE}\widetilde{\bullet}\bigr)(\partial_x,\partial_x)&=\comm{\EEE}{\extprod{\partial_x}{\partial_x}}-2\,\extprod{\comm{\EEE}{\partial_x}}{\partial_x}\\
        &=\comm{\EEE}{\lambda'\partial_x}+2f'\,\extprod{\partial_x}{\partial_x}\\
        &=(\Lie_\EEE\lambda')\,\partial_x+\lambda'\,\comm{\EEE}{\partial_x}+2f'\lambda'\,\partial_x\\
        &=\lambda'\,\partial_x+(\Lie_{\comm{\EEE}{\partial_x}}\lambda)\,\partial_x+f'\lambda'\,\partial_x\\
        &=\lambda'\,\partial_x=\extprod{\partial_x}{\partial_x}\,,
        \end{aligned}
        \]
  where we used the fact that the map $X\mapsto \Lie_X$ is a Lie algebra homomorphism. As for the second case:
        \[
        \begin{aligned}
             \bigl(\Lie_{\EEE}\widetilde{\bullet}\bigr)(X,\partial_x)&=\Lie_\EEE(\extprod{X}{\partial_x})-\extprod{(\Lie_\EEE X)}{\partial_x}-\extprod{X}{(\Lie_\EEE\partial_x)}\\
             &=\comm{E+f\partial_x}{(\Lie_X\lambda)\partial_x}-\extprod{\bigl(\comm{E}{X}-(\Lie_Xf)\partial_x\bigr)}{\partial_x}+\extprod{X}{f'\partial_x}\\
             &=\bigl[(\Lie_E\Lie_X\lambda)+f(\Lie_X\lambda')-f'(\Lie_X\lambda)-\Lie_{\comm{E}{X}}\lambda+\lambda'(\Lie_Xf)+f'(\Lie_X\lambda)\bigr]\,\partial_x\\
             &=\Lie_X\bigl(\Lie_E\lambda+f\lambda'\bigr)\,\partial_x=(\Lie_X\lambda)\,\partial_x=\extprod{X}{\partial_x}\,,
        \end{aligned}
        \]
       When $X$ and $Y$ are projectable, we denote $\Lambda(X,Y):=\Lie_X\lambda\,\Lie_Y\lambda-\Lie_{X\bullet Y}\lambda\,$, so that one has:
       \[
       \begin{aligned}
           \Lie_\EEE(\extprod{X}{Y})&=\Lie_\EEE( X\bullet Y)+\comm{\EEE}{\tfrac{1}{\lambda'}\Lambda(X,Y)\,\partial_x}\\
           &=\Lie_E( X\bullet Y)-(\Lie_{X\bullet Y}f)\,\partial_x+\Lie_\EEE\bigl(\tfrac{1}{\lambda'}\Lambda(X,Y)\bigr)\,\partial_x-\tfrac{1}{\lambda'}f'\Lambda(X,Y)\,\partial_x\\
           &=\Lie_E( X\bullet Y)-(\Lie_{X\bullet Y}f)\,\partial_x+\bigl(-\tfrac{1}{\lambda'^2}(\Lie_\EEE\lambda')\Lambda(X,Y)+\tfrac{1}{\lambda'}\Lie_\EEE\Lambda(X,Y)\bigr)\,\partial_x-\tfrac{1}{\lambda'}f'\Lambda(X,Y)\,\partial_x\\
             &=\Lie_E( X\bullet Y)-(\Lie_{X\bullet Y}f)\,\partial_x+\tfrac{1}{\lambda'}\bigl(\Lie_\EEE\Lambda(X,Y)-\Lambda(X,Y)\bigr)\,\partial_x\,,\\
             \extprod{(\Lie_\EEE X)}{Y}&=\extprod{(\Lie_EX)}{Y}-(\Lie_X f)(\Lie_Y\lambda)\,\partial_x\\
             &=(\Lie_EX)\bullet Y+\tfrac{1}{\lambda'}\Lambda(\comm{E}{X},Y)\,\partial_x-(\Lie_X f)(\Lie_Y\lambda)\,\partial_x\,.
       \end{aligned}
       \]
       The Lie derivative $\Lie_\EEE\Lambda(X,Y)$ will contain second Lie derivatives of the form $\Lie_\EEE\Lie_Z\lambda$ for some projectable vector field $Z\,$. Using again the fact that the Lie derivative is a Lie algebra homomorphism, these terms can be computed as follows:
       \[
       \Lie_\EEE\Lie_Z\lambda=\Lie_Z\lambda+\Lie_{\comm{\EEE}{Z}}\lambda=\Lie_Z\lambda+\Lie_{\comm{E}{Z}}\lambda-\lambda'\Lie_Zf\,\,.
       \]
       As a consequence:
       \[
       \begin{aligned}
           \Lie_\EEE\Lambda(X,Y)&=(\Lie_\EEE\Lie_X\lambda)\Lie_Y\lambda+\Lie_X\lambda(\Lie_\EEE\Lie_Y\lambda)-\Lie_\EEE\Lie_{X\bullet Y}\lambda\\
           &=\Lambda(X,Y)+(\Lie_X\lambda)(\Lie_Y\lambda)-\Lie_{\comm{E}{X\bullet Y}}\lambda+(\Lie_{\comm{E}{X}}\lambda)(\Lie_Y\lambda)+(\Lie_X\lambda)(\Lie_{\comm{E}{Y}}\lambda)+\\
           &\quad-\lambda'\bigl[(\Lie_X f)(\Lie_Y\lambda)+(\Lie_X\lambda)(\Lie_Y f)-\Lie_{X\bullet Y}f\bigr]\\
           &=2\Lambda(X,Y)+\Lambda(\comm{E}{X},Y)+\Lambda(X,\comm{E}{Y})-\lambda'\bigl[(\Lie_X f)(\Lie_Y\lambda)+(\Lie_X\lambda)(\Lie_Y f)-\Lie_{X\bullet Y}f\bigr]\,.\\
       \end{aligned}
       \]
     All the terms cancel out except for the first one, and we are left with:
       \[
       \begin{aligned}
           (\Lie_\EEE\widetilde{\bullet})(X,Y)&=(\Lie_E\bullet)(X,Y)+\tfrac{1}{\lambda'}\Lambda(X,Y)\,\partial_x=\extprod{X}{Y}\,.
       \end{aligned}
       \]
       This concludes the proof.
    \end{proof}
    \begin{oss}
        The result of the previous Proposition was expected in view of the equivalent characterisation of eventual identities given in \cite{davidstrDubduality}. In fact, in this case it is easy to see that $\comm{\ee}{\EEE}=\ee$ \cite{RomanoDoubleHurwitz}.
    \end{oss}
\begin{oss}
    As a consequence, according to \cite[Lemma 4]{davidstrDubduality}, it is also:
    \begin{equation}\label{eq:LieEinvbullet}
        \Lie_{\EEE^{-1}}\widetilde{\bullet}=-\EEE^{-2}\,\widetilde{\bullet}\,.
    \end{equation}
\end{oss}

The results of this subsection can be summed up in the following:
\begin{thm}\label{thm:biflatstructure}
    The flat F-manifold structure on the universal curve $\UC_{g\,;\,\bm{n}}$ given by the extended product $\widetilde{\bullet}$ and identity $\ee$ is homogeneous with respect to $\EEE\,$, in the sense of \cref{defn:homogflatFmfld}. Equivalently, the eventual identity $\EEE$ determines a bi-flat F-manifold structure on the universal curve equipped with $(\widetilde{\bullet},\widetilde{\nabla},\ee)\,$. 
\end{thm}

    \subsection{Twist of the extended Dubrovin connection}
According to \cref{thm:flatnessExtDubConnection}, the whole data encoded into a Frobenius manifold is equivalent to the flatness of the extended Dubrovin connection. Now, in general, for a flat F-manifold, one will not be able to define such a connection, for instance due to the absence of an Euler vector field. On the other hand, we are here dealing with a case that has remarkably more structure, namely a bi-flat F-manifold that admits a foliation whose leaves are Frobenius manifolds. For this reason, one may hope to still be able to define a \emph{twist}\footnote{We use this terminology to avoid extension of the extended Dubrovin connection.} of the extended Dubrovin connection on the underlyng Hurwitz Frobenius manifold whose flatness encodes our structure.

To fix the notation, we let $\pr:\RS\times \Hw_{g\,;\,\bm{n}}\to \Hw_{g\,;\,\bm{n}}$ and $\widetilde{\pr}:\RS\times \UC_{g\,;\,\bm{n}}\to\UC_{g\,;\,\bm{n}}$ be the projections on the second factor, and we set $\widetilde{\pi}:=\id\times\pi:\RS\times\UC_{g\,;\,\bm{n}}\to \RS\times \Hw_{g\,;\,\bm{n}}\,$, where $\pi:\UC_{g\,;\,\bm{n}}\to \Hw_{g\,;\,\bm{n}}$ denotes the canonical projection. Then we have the following commuting diagram:
\begin{equation}\label{cd:extendedDub}
    \begin{tikzcd}
       \RS\times \UC_{g;\bm{n}}\arrow[rr, "\widetilde{\pr}"]\arrow[d, "\widetilde{\pi}"]&& \UC_{g;\bm{n}}\arrow[d, "\pi"]\\
      \RS\times \Hw_{g\,;\,\bm{n}}  \arrow[rr, "\pr"]&&\Hw_{g;\bm{n}}\,.
    \end{tikzcd}
\end{equation}
Now, the pull-back $\widetilde{\pi}^*\widehat{\nabla}$ of the extended Dubrovin connection is a connection on the pull-back bundle $\widetilde{\pi}^*\pr^*T\Hw_{g\,;\,\bm{n}}\,$, which is isomorphic to the sub-bundle $\widetilde{\pr}^*\pi^*T\Hw_{g;\bm{n}}\,$ of $\widetilde{\pr}^*T\UC_{g;\bm{n}}$ thanks to the previous diagram. Again, it is important to remark that we are identifying $\pi^*T\Hw_{g\,;\,\bm{n}}$ with its image through the Ehresmann connection induced by the primary differential. 

If a notion of a twist of the extended Dubrovin connection to the universal curve exists, it is natural to impose that it be in some sense compatible with the underlying extended Dubrovin connection. This leads to the following:
\begin{defn}
    Let $\widehat{\widetilde{\nabla}}$ be a connection on $\widetilde{\pr}^*T\UC_{g;\bm{n}}$. We say that $\widehat{\widetilde{\nabla}}$ is \emph{$\pi$-compatible} if, for any section $Y$ of $(\pi\circ\widetilde{\pr})^*T\Hw_{g;\bm{n}}\,$:
    \[
 \widehat{\widetilde{\nabla}}_XY-   (\widetilde{\pi}^*\widehat{\nabla})_XY\in \ker\pi_*\,,
    \]
    for any tangent vector $X$ to the total space $\RS\times \UC_{g;\bm{n}}\,$.
\end{defn}
\begin{oss}
    More concretely, if $\bm{v}$ is a system of flat coordinates for $\eta$ on $\Hw_{g\,;\,\bm{n}}$ and $x$ is an adapted flat coordinate on the universal curve for the extension, then a section $Y$ of $\widetilde{\pr}^*T\UC_{g;\bm{n}}$ can be uniquely decomposed as follows:
    \[
    Y=\sum_{\alpha}f_\alpha(\bm{v},x,\zeta)\,\partial_{v_\alpha}+h(\bm{v},x,\zeta)\,\partial_x\,,\qquad \zeta\in\RS\,.
    \]
    On the other hand, a section of $(\pi\circ\widetilde{\pr})^*T\Hw_{g;\bm{n}}$ has vanishing component in the $x$ direction. This is clearly the only case where it makes sense to compute covariant derivatives with respect to the extended Dubrovin connection $\widehat{\nabla}$ -- or, rather, to its pull-back. We are therefore requiring that, whenever both are well-defined, the two covariant derivatives be $\pi$-related.
\end{oss}
Since we know that compatibility between $\bullet$ and ${}^\eta\nabla$ on the Hurwitz space lifts to compatibility between $\widetilde{\bullet}$ and $\widetilde{\nabla}$ on the universal curve, we look at connections on $\widetilde{\pr}^*T\UC_{g;\bm{n}}$ of the form:
    \begin{equation}\label{eq:liftedextendedDubconnection}
            \begin{aligned}
        \widehat{\widetilde{\nabla}}_XY&:=\widetilde{\nabla}_XY+\zeta\,\extprod{X}{Y}\,,\\
        \widehat{\widetilde{\nabla}}_{\partial_\zeta}Y&:=\widetilde{\nabla}_{\partial_\zeta }Y+\extprod{\EEE}{Y}-\tfrac{1}{\zeta}\widetilde{\mu}(Y)\,,
    \end{aligned}
    \end{equation}
    for any two sections $X$ and $Y$ of $\widetilde{\pr}^*T\UC_{g;\bm{n}}$ and for some section $\widetilde{\mu}$ of $\End(T\UC_{g\,;\,\bm{n}})$.
\begin{oss}
    Notice that, picking the term $\extprod{\EEE}{Y}$ in the second line of \cref{eq:liftedextendedDubconnection} shifts the $\pi$-compatibility completely into $\widetilde{\mu}\,$.
\end{oss}

\begin{lem}\label{lem:picompatibilitymu}
    A connection $\widehat{\widetilde{\nabla}}$ on $\widetilde{\pr}^*T\UC_{g;\bm{n}}$ defined as in \cref{eq:liftedextendedDubconnection} is $\pi$-compatible if and only if, for any $X$ section of $(\pi\circ\widetilde{\pr})^*T\Hw_{g;\bm{n}}\,$,
    \[
    \widetilde{\mu}(X)-\mu(X)\in\ker\pi_*\,.
    \]
\end{lem}
\begin{proof}
    The fact that $\widehat{\widetilde{\nabla}}_YX-(\widetilde{\pi}^*\widehat{\nabla})_YX$ is vertical with respect to $\pi$ whenever $Y$ lies in the kernel of $\widetilde{\pr}_*$ follows from the definition of extensions of almost-flat F manifolds. One, then, only needs to check that the $\pi$-compatibility holds true when $Y=\partial_\zeta\,$. Thanks to \cref{lem:projectioneventualidentity}, one has:
    \[
    \begin{aligned}
        \widehat{\widetilde{\nabla}}_{\partial_\zeta}X-(\widetilde{\pi}^*\widehat{\nabla})_{\partial_\zeta}X&=\extprod{\EEE}{X}-E\bullet X-\tfrac{1}{\zeta}\bigl(\widetilde{\mu}(X)-\mu(X)\bigr)\\
        &=-\tfrac{1}{\zeta}\bigl(\widetilde{\mu}(X)-\mu(X)\bigr)+\ker\pi_*\,.
    \end{aligned}
    \]
    This concludes the proof.
\end{proof}
Similarly, we calculate the curvature of connections $\widehat{\widetilde{\nabla}}\,$:
\begin{lem}\label{lem:flatnessliftextendedDub}
 A connection $\widehat{\widetilde{\nabla}}$ on $\widetilde{\pr}^*T\UC_{g;\bm{n}}$ defined as in \cref{eq:liftedextendedDubconnection} is flat if and only if, for any $X,Y\in\widetilde{\pr}^*\X_{\UC_{g;,\bm{n}}}$\,:
    \begin{align}
     \label{eq:twistedDUbflat1}     \comm{\widetilde{\mu}}{\widetilde{\nabla}_X}&=0\,,\\
       \label{eq:twistedDUbflat2} \widetilde{\mu}(\extprod{X}{Y})-\extprod{X}{\widetilde{\mu}(Y)}&=\extprod{X}{\widetilde{\nabla}_Y\EEE}-\widetilde{\nabla}_{\extprod{X}{Y}}\EEE\,.
    \end{align}
\end{lem}
\begin{proof}
    We just need to check that the curvature endomorphism $\widehat{\widetilde{R}}(X,\partial_\zeta)$ vanishes identically for any $X$ projectable with respect to $\widetilde{\pr}\,$, as any other case follows either from skew-symmetry of the Riemann curvature or from compatibility of $\widetilde{\nabla}$ and $\widetilde{\bullet}\,$. An easy calculation gives that, for any section $Y$ of $\pr^*T\UC_{g\,;\,\bm{n}}\,$:
    \[
    \begin{aligned}
        \widehat{\widetilde{R}}(X,\partial_\zeta)Y&=\comm{\widehat{\widetilde{\nabla}}_X}{\widehat{\widetilde{\nabla}}_{\partial_\zeta}}Y\\
        &=\widetilde{R}(X,\partial_\zeta)Y+\widetilde{\nabla}_X(\extprod{\EEE}{Y})-\extprod{\widetilde{\nabla}_XY}{\EEE}-\extprod{X}{Y}+\widetilde{\mu}(\extprod{X}{Y})-\extprod{X}{\widetilde{\mu}(Y)}+\\
        &\quad+\tfrac{1}{\zeta}\bigl(\widetilde{\mu}(\widetilde{\nabla}_XY)-\widetilde{\nabla}_X\widetilde{\mu}(Y)\bigr)\,.
    \end{aligned}
    \]
    Since $\widetilde{\nabla}$ is flat, the first term is zero. The vanishing of the coefficient of $\zeta^{-1}$ gives $\comm{\widetilde{\mu}}{\widetilde{\nabla}_X}=0\,$. As for the remaning terms in the coefficient of $\zeta^0$, we notice that, using the compatibility of $\widetilde{\nabla}$ and $\widetilde{\bullet}$, along with \cref{prop:homogeneityE}:
    \[
    \begin{aligned}
\widetilde{\nabla}_X(\extprod{\EEE}{Y})-\extprod{\widetilde{\nabla}_XY}{\EEE}-\extprod{X}{Y}&=(\widetilde{\nabla}_X\widetilde{\bullet})(\EEE,Y)+\extprod{\widetilde{\nabla}_X\EEE}{Y}-\extprod{X}{Y}\\
&=(\widetilde{\nabla}_\EEE\widetilde{\bullet})(X,Y)+\extprod{\widetilde{\nabla}_X\EEE}{Y}-\extprod{X}{Y}\\
&=\widetilde{\nabla}_\EEE(\extprod{X}{Y})-\extprod{\comm{\EEE}{X}}{Y}-\extprod{X}{Y}-\extprod{X}{\widetilde{\nabla}_\EEE Y}\\
&=\widetilde{\nabla}_\EEE(\extprod{X}{Y})+\extprod{X}{\comm{\EEE}{Y}}-\comm{\EEE}{\extprod{X}{Y}}-\extprod{X}{\widetilde{\nabla}_\EEE Y}\\
&=\widetilde{\nabla}_{\extprod{X}{Y}}\EEE-\extprod{X}{\widetilde{\nabla}_Y\EEE}\,.
\end{aligned}
    \]
    This concludes the proof.
\end{proof}
\begin{lem}\label{lem:twistedDubconnectionlemma3}
    A connection $\widehat{\widetilde{\nabla}}$ on $\widetilde{\pr}^*T\UC_{g;\bm{n}}$ defined as in \cref{eq:liftedextendedDubconnection} is $\pi$-compatible and satisfies \cref{eq:twistedDUbflat2} if and only if there exists some $\pi$-vertical section $\phi$ of $\widetilde{\pr}^*T\UC_{g;\bm{n}}$ such that:
    \[
    \widetilde{\mu}(X)=\bigl(1-\tfrac{d}{2}\bigr)X-\widetilde{\nabla}_X\EEE+\extprod{X}{\phi}\,.
    \]
\end{lem}
\begin{proof}
    The connection is $\pi$-compatible if and only if, according to \cref{lem:picompatibilitymu} and \cref{lem:projectioneventualidentity}, on any section $X$ of $(\pi\circ\widetilde{\pr})^*T\Hw_{g;\bm{n}}\,$, $\widetilde{\mu}$ is:
\[
\widetilde{\mu}(X)=\bigl(1-\tfrac{d}{2}\bigr)X-\nabla_XE+\ker\pi_*=\bigl(1-\tfrac{d}{2}\bigr)X-\widetilde{\nabla}_X\EEE+\ker\pi_*\,.
\]
  All we are left to fix are, then, the $\pi$-vertical component of $\widetilde{\mu}(X)$ whenever $X$ is $\pi$-horizontal and $\widetilde{\mu}(V)$ for $V$ $\pi$-vertical. For the latter, we have the following necessary condition for \cref{eq:twistedDUbflat2} to hold:
    \[
    \begin{aligned}
        \widetilde{\mu}(\extprod{\ee}{V})-\extprod{V}{\widetilde{\mu}(\ee)}&=\extprod{\widetilde{\nabla}_\ee \EEE}{V}-\widetilde{\nabla}_{\extprod{\ee}{V}}\EEE\\
        \widetilde{\mu}(V)-\extprod{V}{\bigl[\bigl(1-\tfrac{d}{2}\bigr)\ee-\widetilde{\nabla}_\ee \EEE+\ker\pi_*\bigr]}&=\extprod{\widetilde{\nabla}_\ee \EEE}{V}-\widetilde{\nabla}_{V}\EEE\\
        \widetilde{\mu}(V)-\bigl(1-\tfrac{d}{2}\bigr)V+\ker\pi_*&=-\widetilde{\nabla}_{V}\EEE\,,
    \end{aligned}
    \]
where we used the fact that $\ker(\pi_*)_p$ is an ideal in the tangent space at any $p$ in the universal curve. Hence, it is necessarily $\widetilde{\mu}(X)=\bigl(1-\tfrac{d}{2}\bigr)X-\widetilde{\nabla}_X\EEE+\ker\pi_*$ for any section $X$ of $\widetilde{\pr}^* T\UC_{g\,;\,\bm{n}}\,$. Let us denote by $\phi(X)$ the remaining $\pi$-vertical component of $\widetilde{\mu}(X)$ in the previous expression.
A similar computation gives that it is necessary that the value of $\phi$ at any vector field $X$ is determined by the one at the identity as follows: $\phi(X)=\extprod{X}{\phi(\ee)}\,$. Since $\pi$-vertical vector fields form an ideal, this still ensures that $\phi(X)$ is $\pi$-vertical if $\phi(\ee)$ is. As easy check then shows that this condition is also sufficient.
\end{proof}

\begin{thm}\label{thm:flatnesstwistedDubconnection}
There exists a unique flat and $\pi$-compatible connection $\widehat{\widetilde{\nabla}}$ on the pull-back bundle $\widetilde{\pr}^*T\UC_{g\,;\,\bm{n}}$ of the form in \cref{eq:liftedextendedDubconnection} with:
    \[
\widetilde{\mu}(X):=\bigl(1-\tfrac{d}{2}\bigr)X-\widetilde{\nabla}_X\EEE\,.
\]
\end{thm}

\begin{proof}
According to \cref{lem:twistedDubconnectionlemma3}, $\pi$-compatibility and \cref{eq:twistedDUbflat2} are equivalent to the given expression for $\widetilde{\mu}\,$. We are only, therefore, left to discuss \cref{eq:twistedDUbflat1}.

To begin with, we notice that it is equivalent to requiring that $\widetilde{\mu}$ preserves flat sections of $\widetilde{\nabla}\,$. The direct implication is obvious. On the other hand, if $\widetilde{\mu}$ preserves flat sections of $\widetilde{\nabla}\,$, then let $V_1,\dots, V_{n+1}$ be a local system of flat sections of $\widetilde{\nabla}\,$. Then, any section of $\widetilde{\pr}^*T\UC_{g\,;\,\bm{n}}$ can be uniquely decomposed onto such a system, with coefficients generally being functions on $\RS\times\UC_{g\,;\,\bm{n}}\,$. An easy computation gives, then, for any $X$:
    \[
    \begin{aligned}
        \widetilde{\mu}(\widetilde{\nabla}_XY)&=\sum_\alpha \widetilde{\mu}\bigl(\widetilde{\nabla}_X(f_\alpha V_\alpha)\bigr)=\sum_\alpha X(f_\alpha)\,\widetilde{\mu}(V_\alpha)\,,\\
        \widetilde{\nabla}_X\widetilde{\mu}(Y)&=\sum_\alpha\widetilde{\nabla}_X\bigl(f_\alpha\,\widetilde{\mu}(V_\alpha)\bigr)=\sum_\alpha X(f_\alpha)\,\widetilde{\mu}(V_\alpha)\,.
    \end{aligned}
    \]
    This proves the reverse implication.
 Now, with a similar argument to the one in the proof of \cref{cor:EvIdAffine}, it is easy to see that flatness of $e$ -- which is equivalent to $\Lie_\ee\omega =0$ \cite{RomanoDoubleHurwitz} -- implies that the vertical component $\ee(x)$ of $\ee$ must be fibre-wise constant, but not necessarily zero. As a consequence, $\ee$ is $\widetilde{\nabla}$-flat.
It follows that it must be $\widetilde{\nabla}_X\phi(\ee )=\widetilde{\nabla}^2_{X,\ee }\EEE=0$ for any $X\,$. In fact, $\widetilde{\nabla}_\ee \EEE$ is a flat section of $\widetilde{\nabla}$:
\[
\widetilde{\nabla}_\ee \EEE=\widetilde{\nabla}_\EEE \ee +\comm{\ee }{\EEE}=\ee \,.
\]
As a consequence, $\widetilde{\nabla}^2_{X,\ee }\EEE=\widetilde{\nabla}_X\widetilde{\nabla}_\ee \EEE-\widetilde{\nabla}_{\widetilde{\nabla}_X \ee }\EEE=0\,$.

If, now, $Y$ is a flat section of $\widetilde{\nabla}\,$, then we require that, for any $X$:
\[
\begin{aligned}
    0=\widetilde{\nabla}_X\widetilde{\mu}(Y)&=-\widetilde{\nabla}^2_{X,Y}\EEE+\widetilde{\nabla}_X(\extprod{Y}{\phi})\\
    &=-\widetilde{\nabla}^2_{X,Y}\EEE+(\widetilde{\nabla}_X\widetilde{\bullet})(Y,\phi)\\
    &=-\widetilde{\nabla}^2_{X,Y}\EEE+(\widetilde{\nabla}_{\phi}\widetilde{\bullet})(X,Y)\\
    &=-\widetilde{\nabla}^2_{X,Y}\EEE+\comm{\phi}{\extprod{X}{Y}}-\extprod{X}{\comm{\phi}{Y}}-\extprod{Y}{\comm{\phi}{X}}\\
    &=-\widetilde{\nabla}^2_{X,Y}\EEE+(\Lie_{\phi}\widetilde{\bullet})(X,Y)\,.
\end{aligned}
\]
Since $\EEE$ is affine, then $\phi=0\,$. In fact, if it were not the case, then $\Lie_\phi(\phi,\phi)\neq0\,$.
\end{proof}
\begin{oss}
    In the proof of the previous Theorem, we have left $\widetilde{\nabla}^2\EEE$ explicit even though, in this case, it is zero. This is to highlight that, if one were to study Dubrovin duality of an arbitrary rank-one extension, then the twisted extended Dubrovin connection could still be defined provided that $\phi$ is a solution to $\Lie_\phi\widetilde{\bullet}=\widetilde{\nabla}^2\EEE\,$. In fact, all previous lemmas generalize to any arbitrary rank-one extension, provided that $\comm{\ee}{\EEE}=\ee\,$.  
    
    Notice that, in the standard open Gromov-Witten theory set up, the further following normalisation conditions are assumed for the extended prepotential \cite{PSTInttheory,BuryakArsingularity,BuryakOpenAD,dealmeida2025openhurwitzflatf}:  
    \begin{equation}
         \begin{aligned}
        \Lie_e\Omega'&=1\,,&&& \Lie_e\Omega_\alpha &=0\,,&&& \alpha&=1,\dots,n\,.
    \end{aligned}
    \end{equation}
    These hold if and only if the unity $e$ of the Frobenius product lifts to the unity of the extended product, i.e. the vertical component of $\ee$ vanishes. Equivalently, they require that $\Lie_e\lambda=1\,$. In other words, the vertical component of $\ee$ is required to vanish.
\end{oss}

\begin{defn}\label{defn:liftextDubconnection}
Let $\pi,\pr,\widetilde{\pi},\widetilde{\pr}$ like in the diagram in \cref{cd:extendedDub}.The \emph{twisted extended Dubrovin connection} is the connection $\widehat{\widetilde{\nabla}}$ on the pull-back bundle $\widetilde{\pr}^*T\UC_{g\,;\,\bm{n}}$ defined by, for any sections $X,Y$ of $ \widetilde{\pr}^*T\UC_{g;\bm{n}}$ and $\zeta\in\RS\,$:
  \[
    \begin{aligned}
       \widehat{\widetilde{\nabla}}_XY&:=\widetilde{\nabla}_XY+\zeta\,\extprod{X}{Y}\,,\\
      \widehat{\widetilde{\nabla}}_{\partial_\zeta}Y&:=\widetilde{\nabla}_{\partial_\zeta }Y+\extprod{\EEE}{Y}-\tfrac{1}{\zeta}\bigl[\bigl(1-\tfrac{d}{2}\bigr)Y-\widetilde{\nabla}_Y\EEE\bigr]\,.
    \end{aligned}
    \]
\end{defn}

\subsection{Extension of the dual connection}\label{ss:extdualconnection}
Finally, we are here considering rank-one extensions of a Frobenius manifold, and of its Dubrovin dual, to almost-flat F-manifolds. As a consequence, we are not only looking for a relation between the extended multiplications $\widetilde{\bullet}$ and $\widetilde{\ast}$ that generalises the one between $\bullet$ and $\ast$, but also between the extended connections $\widetilde{\nabla}$ and $\widetilde{\nabla^\star}$ that encodes the one between ${}^\eta\nabla$ and ${}^g\nabla$.
\begin{prop}[\cite{hertling_2002}, Theorem 9.4]\label{prop:firstsecondstructureconnection}
    Let $M$ be a Frobenius manifold with charge $d$, ${}^\eta\nabla$ and ${}^g\nabla$ denote the Levi-Civita connections of the Frobenius metric and of the intersection form respectively, $E$ be the Euler vector field. Then, for any two holomorphic vector fields $X,Y$ on $M$, we have:
    \[
    \begin{aligned}
            {}^g\nabla_XY&={}^\eta\nabla_XY-\nabla_{X\ast Y}E+\tfrac12(1-d)\,X\ast Y\\
           & =E\bullet{}^\eta\nabla_X(E^{-1}\bullet Y)-({}^\eta\nabla_{E^{-1}\bullet Y}E)\bullet X+\tfrac12(3-d)\,X\ast Y\,.
    \end{aligned}
    \]
\end{prop}
Regarding the connections, we know that the relation between $\widetilde{\nabla^\star}$ and $\widetilde{\nabla}$ should project down to the one in \cref{prop:firstsecondstructureconnection}. As a consequence, it is natural to ask whether the two extended connections satisfy the very same equation, with $E$ replaced by $\EEE$, and $\bullet$ by $\widetilde{\bullet}$. In other words, we consider, on the universal curve, the connection\footnote{Equivalence between the two expressions is the content of \cite[Theorem (9.4) a)]{hertling_2002}. In particular, it is due to connection-multiplication compatibility and homogeneity of the product with respect to the eventual identity. These properties hold true even in this setting.}:
\begin{equation}\label{eq:Enabla}
\begin{split}
    {}^\EEE\nabla_XY&:=\widetilde{\nabla}_XY-\widetilde{\nabla}_{\extdprod{X}{Y}}\EEE+\tfrac12(1-d)\,\extdprod{X}{Y}\\
    &\equiv\EEE\,\widetilde{\bullet}\,\widetilde{\nabla}_X(\EEE^{-1}\,\widetilde{\bullet}\, Y)-\bigl(\widetilde{\nabla}_{\EEE^{-1}\,\widetilde{\bullet}\,Y}\EEE\bigr)\,\widetilde{\bullet}\, X+\tfrac12 (3-d)\,\extdprod{X}{Y}\,.
\end{split}    
\end{equation}

\begin{lem}
    Consider the flat F-manifold $(\UC_{g\,;\,\bm{n}},\widetilde{\bullet},\widetilde{\nabla},\ee)$ with eventual identity $\EEE$. ${}^\EEE\nabla$ is torsion-free, compatible with the extended dual product $\widetilde{\ast}$ and flat.
\end{lem}
\begin{proof}
    As noted in \cref{rem:dualconnectionbiflatFmfld}, the connection
    \[
    \nabla^\ast_XY:=\widetilde{\nabla}_XY-\widetilde{\nabla}_{\extdprod{X}{Y}}\EEE
    \]
    determines a second flat F-manifold structure on $\UC_{g\,;\,\bm n}\,$ (and, in fact, makes it into a bi-flat F-manifold). In particular, such connection spans a pencil $\bigl\{\nabla^\ast+z\,\widetilde{\ast}\bigr\}$ of torsion-free, flat connections that are compatible with $\widetilde{\ast}\,$, for $z\in\CC\,$. Clearly, ${}^\EEE\nabla$ lies in such pencil for $z=\tfrac12(1-d)\,$.
\end{proof}

We now want to discuss when the flat coordinates of the connections ${}^\EEE\nabla$ and $\widetilde{\nabla^\ast}$ coincide. Firstly, we point out the following trivial result:
\begin{lem}
    Let $\nabla$ be a flat connection on $M$. A coordinate system $v_1,\dots, v_n$ on $M$ is flat for $\nabla$ if and only if $\nabla\partial_{v_a}=0$ for any $a=1,\dots, n$.
\end{lem}
\begin{proof}
The system is flat if and only if $\nabla\dd{v}_a=0$ for any $a=1,\dots, n$. Now, the connections on the tangent and cotangent bundles are related as follows:
    \[
    X(Y\intprod \omega)=(\nabla_XY)\intprod\omega+Y\intprod(\nabla_X\omega)\,,
    \]
    for any pair of vector fields $X,Y$ and one-form $\omega$.
    If $X=\partial_{v_a}$, $Y=\partial_{v_b}$ and $\omega=\dd{v}_c$, this boils down to:
    \[
    0=(\nabla_a\partial_{v_b})\intprod\dd{v}_c+\partial_{v_b}\intprod \,(\nabla_{a}\dd{v}_c)\,,
    \]
which implies the statement.
\end{proof}

\begin{lem}\label{lem:Enabla}
   Let $\widetilde{\nabla}$ denote the connection on $\UC_{g\,;\,\bm{n}}$ coming from extending the Levi-Civita connection ${}^\eta\nabla$ on $\Hw_{g\,;\,\bm{n}}$ as described in \cite{dealmeida2025openhurwitzflatf}, let $\EEE$ be the eventual identity in \cref{prop:eventualidentityhere}. Then, ${}^\EEE\nabla\partial_x=0$.
\end{lem}
\begin{proof}
Let $\lambda' f:=\lambda-\Lie_E\lambda$, as in \cref{prop:eventualidentityhere}. Since $x$ is a flat coordinate for $\widetilde{\nabla}$, an easy computation keeping in mind the multiplication table in \cref{eq:multtableextension} gives:
    \[
    \begin{aligned}
            {}^\EEE\nabla_{\partial_x}\partial_x&=-\tfrac{1}{\lambda}\lambda'\bigl[f'-\tfrac12(1-d)\bigr]\,\partial_x\,,\\
            {}^\EEE\nabla_X\partial_x&=-\tfrac{1}{\lambda}\Lie_X\lambda\,\bigl[f'-\tfrac12(1-d)\bigr]\,\partial_x\,,
    \end{aligned}
    \]
    for any projectable vector field $X$. Therefore, $\partial_x$ is a flat section of ${}^\EEE\nabla$ if and only if $f'=\tfrac12(1-d)\,$. This is true due to \cref{prop:LieElambda}.
\end{proof}

Finally, we need to check that the flat coordinates for the Levi-Civita connection ${}^g\nabla$ lift to flat coordinates of ${}^\EEE\nabla$. In particular, by construction the horizontal component of ${}^\EEE\nabla_XY$ will coincide with ${}^g\nabla_XY$, whenever $X$ and $Y$ are projectable. However, one also needs to take care of the vertical component:
\begin{thm}\label{thm:extendeddualconnection}
    If $\EEE$ is the eventual identity on the universal curve $\UC_{g\,;\,\bm{n}}$ given in \cref{prop:eventualidentityhere}, then the connection ${}^\EEE\nabla$ as defined in \cref{eq:Enabla} coincides with the extended dual connection $\widetilde{\nabla^\star}$.
    \end{thm}
\begin{proof}
It is enough to check that, for $\EEE$ as given in \cref{lem:Enabla}, ${}^\EEE\nabla_XY=0$ whenever $X$ and $Y$ are projectable and (the projection of) $Y$ is a flat section of ${}^g\nabla\,$. In fact, $\widetilde{\nabla}_{\partial_x}Y=0$ follows from torsion-freeness, provided that $\partial_x$ is flat.

This is an easy calculation using the first expression for ${}^\EEE\nabla$ in \cref{eq:Enabla}. In fact, if we denote by $\Lambda^\ast(X,Y)$ the vertical component function in $\extdprod{X}{Y}$ as in \cref{lem:multiplicationhorizontal}. Then:
\[
\begin{aligned}
    {}^\EEE\nabla_{X}Y&={}^\eta\nabla_XY-\widetilde{\nabla}_{\extdprod{X}{Y}}\EEE+\tfrac12(1-d)\,\extdprod{X}{Y}\\
    &={}^\eta\nabla_XY-{}^\eta\nabla_{X\ast Y}E+\tfrac12(1-d)\,X\ast Y+\bigl(\tfrac12(1-d)-f'\bigr)\Lambda^*(X,Y)\,\partial_x\\
    &={}^g\nabla_XY+\bigl(\tfrac12(1-d)-f'\bigr)\Lambda^*(X,Y)\,\partial_x\,.
\end{aligned}
\]
According to \cref{prop:LieElambda}, the vertical component vanishes and the statement follows.
\end{proof}

\section{Open WDVV solutions for ADE singularities}\label{ss:typeDE}

Let $W$ be a simply-laced Weyl group of rank $\ell$, and $f:\CC^2\times \CC^\ell\to \CC$ be the miniversal deformation of the corresponding simple singularity. We denote coordinates on $\CC^2$ by $z$ and $w$, whereas coordinates on the parameter space will be denoted by $\bm{a}\in\CC^\ell\,$. Using the normal forms of the singularities in \cite{arnold1985singularities}, it is explicitly:
\begin{equation}\label{eq:miniversaldeformationsADE}
\begin{aligned}
A_\ell&:&&& f(z,w,\bm{a})&=z^{\ell+1}+w^2+a_1\,z^{\ell-1}+\dots+a_{\ell-1}\,z+a_\ell\,,\qquad\qquad\qquad \,\,\ell\in\ZZ_{\geq 1}\,;\\
    D_\ell &: &&& f(z,w,\bm{a})&=z^{\ell-1}+zw^2+a_1\,z^{\ell-2}+\dots+a_{\ell-2}\,z+a_{\ell-1}+a_\ell\,w\,,\qquad \ell\in\ZZ_{\geq 4}\,;\\
    E_6&:&&& f(z,w,\bm{a})&=z^4+w^3+a_1\,z^2w+a_2\,zw+a_3\,z^2+a_4\,w+a_5\,z+a_6\,,\\
    E_7&:&&&f(z,w,\bm{a})&=z^3w+w^3+a_1 \,zw^2+a_2 \,w^2+a_3\,zw+a_4\,z^2+a_5\,w+a_6\,z+a_7\,,\\
    E_8&:&&&f(z,w,\bm{a})&=z^5+w^3+a_1\,z^3w+a_2\,z^2w+a_3\,z^3+a_4\,zw+a_5\,z^2+a_6\,w+a_7\,z+a_8\,.
\end{aligned}    
\end{equation}

Following \cites{Saito81,Sai83,topstringsd<1,DubLiuZhangDEspotentials}, we know that we can endow the parameter space with a Frobenius manifold structure. Multiplication is given by promoting the vector space isomorphism $\pdv{a_\mu}\mapsto\pdv{f}{a_\mu}$ of the tangent spaces with the Jacobian ring of the miniversal deformation -- called \emph{Kodaira-Spencer map} -- to a ring isomorphism. On the other hand, the metric is defined by Grothendieck's residues as follows:
\begin{equation}\label{eq:metricsingularities}
   \eta_{\bm{a}}\bigl(\tfrac{\partial}{\partial a_\mu}\,,\,\tfrac{\partial}{\partial a_\nu}\bigr):=\sum_{\bm{\zeta}\in\mathrm{Cr}(f_{\bm{a}})}\Res_{\bm{\zeta}}\biggl\{\pdv{f_{\bm{a}}}{a_\mu}\pdv{f_{\bm{a}}}{a_\nu}\frac{\dd{z}\wedge\dd{w}}{(\partial_zf_{\bm{a}})\,(\partial_wf_{\bm{a}})}\biggr\}\equiv\sum_{\bm{\zeta}\in\mathrm{Cr}(f_{\bm{a}})}\eval{\frac{(\partial_\mu f_{\bm{a}})(\partial_\nu f_{\bm{a}})}{\det \text{H}(f_{\bm{a}})}}_{\bm{\zeta}}\,,
\end{equation}
where $\mathrm{Cr}(f_{\bm{a}})\subseteq\CC^2$ is the set of critical points of the miniversal unfolding $f_{\bm{a}}$ with parameters $\bm{a}\,$ and $\text{H}(f_{\bm{a}})$ denotes the Hessian of $f_{\bm{a}}$.

Finally, we notice that there exist unique positive rational numbers $r_1$ and $r_2$ such that, with respect to the grading $\deg z:=r_1\,$, $\deg w:=r_2$ in $\CC[z,w]$, the surface singularity of type $W$ is a homogeneous polynomial of degree given by the Coxeter number $h$ of $W\,$. Equivalently, we say that the singularity is an invertible weighted-homogeneous polynomial in $\CC[z,w]$ with weights $(r_1,r_2;h)\,$. In particular \cite{Dubrovin1999}:
\begin{equation}\label{eq:gradingADE}
    \begin{aligned}
    A_\ell &:&&& r_1&=1\,,&&&r_2&=\tfrac12 (\ell+1)\,, &&& h&=\ell+1\,;\\
        D_\ell&:&&& r_1&=2\,,&&& r_2&=\ell-2\,,&&& h&=2(\ell-1)\,;\\
        E_6&:&&& r_1&=3\,,&&& r_2&=4\,, &&& h&=12\,,\\
        E_7&:&&& r_1&=4\,,&&& r_2&=6\,, &&& h&=18\,,\\
        E_8&:&&& r_1&=6\,,&&& r_2&=10\,, &&& h&=30\,.
    \end{aligned}
\end{equation}
\begin{oss}
The only possibly non-integral number in this list is clearly $r_2$ for $A_\ell\,$. It could still, however, be made into an integer by doubling all the degrees. 
\end{oss}
The Euler vector field is then given by extending this grading to $\CC[z,w,\bm{a}]$ in such a way that the miniversal deformation is still homogeneous of degree $h\,$, and then normalising to $h\,$. In other words, there are unique positive integers $d_1,\dots, d_\ell\in\ZZ_{\geq 0}$ such that the miniversal deformation is still invertible weighted-homogeneous with weights $(r_1,r_2,d_1,\dots, d_\ell;h)\,$:
\begin{equation}\label{eq:homogeneitydeformation}
    r_1\,zf_z+r_2\,wf_w+d_1\,a_1 f_1+\dots+ d_{\ell}\,a_\ell f_\ell=h\,f\,,
\end{equation}
where $f_\mu:=\partial_{a_\mu}f\,$. In particular, $d_1,\dots, d_\ell$ are the degrees of the basic invariants in the ring $\CC[\bm{x}]^W$ \cites{humphreys_1990,Dubrovin1996}. The Euler vector field is, then, the derivation on the coordinate ring $\CC[\bm{a}]$ of the parameter space whose eigenspaces are the spaces of homogeneous polynomials with respect to the induced grading, normalised so that the space of homogeneous polynomials of degree $h$ corresponds to the unital eigenvalue. In other words, $E:=\tfrac{1}{h}\sum_{\mu=1}^\ell d_\mu a_\mu\pdv{a_\mu}$ is the preimage of $f$ with respect to the Kodaira-Spencer map.  \\

In \cites{topstringsd<1,DubLiuZhangDEspotentials}, the authors also provide a Landau-Ginzburg description of such Frobenius manifolds. The idea is to replace the two-dimensional residues with ordinary residues of a suitable meromorphic function $\lambda$ -- or, rather, of a family of meromorphic functions corresponding to different points in the parameter space -- on a closed Riemann surface (possibly of positive genus). The top form $\dd{z}\wedge\dd{w}$ on $\CC^2$ is, at the same time, replaced with a meromorphic differential $\omega$ on the surface -- or, again, rather with a family of meromorphic differentials depending on the point in the parameter space. 

In particular, we denote by $\UC_f\hookrightarrow\CC^2\times\CC^\ell$ the holomorphic fibre bundle over the parameter space $\CC^\ell$ whose fibre $\UC_f(\bm{a})$ over $\bm{a}\in\CC^\ell$ is the algebraic curve:
\[
\UC_f(\bm{a}):=\bigl\{(z,w)\in\CC^2:\quad f_w(z,w,\bm{a})=0\bigr\}\hookrightarrow\CC^2\,.
\]
Then, the superpotential is given by $\lambda=f\lvert_{\UC_f}:\UC_f\to\RS\,$. In particular, it is customary to use $z$ as a local coordinate on the curve, so that $\UC_f(\bm{a})$ is (locally) parametrised by $z\mapsto(z,w(z,\bm{a}))\,$, where the function $w(z,\bm{a})$ is determined implicitly by solving for $w$ the algebraic equation $f_w=0\,$. This ensures that both the critical points and values of $\lambda$ coincide with the ones of the corresponding miniversal deformation, as $\lambda_z=f_z$ on the curve $f_w=0\,$. Similarly, $\partial_{a_k}\lambda=\partial_{a_k}f\,$ for any $k=1,\dots, \ell\,$. In type $A$, this clearly just amounts to setting $w$ to zero. For $D_\ell$, on the other hand, the curve is $2zw+a_\ell=0\,$, hence the corresponding superpotential is still a family of meromorphic functions on $\RS$ with poles at $\infty$ and $0\,$. In type-$E$, finally, the curve $f_w=0$ famously has positive genus \cites{TY97,LW96,Dubrovin1999}. 

As for the primary differential, comparing \cref{eq:metricsingularities} with the analogous expression coming from evaluating the residues as given by the Landau-Ginzburg description, and using the implicit function Theorem to compute $\lambda_{zz}$, gives that:
\begin{equation}
    \omega=\frac{\dd{z}}{\sqrt{f_{ww}\lvert_{\UC_f}}}\,.
\end{equation}
As a consequence, to construct the associated solution to the open WDVV equation, as described in \cite{almediaJacobi}, one should firstly change coordinate on the curve so that $\omega=\dd{x}\,$. This is clearly achieved by integrating the ODE:
\begin{equation}\label{eq:flatcoordinateextensionsDE}
\dv{x}{z}=\frac{1}{\sqrt{f_{ww}(z,w(z,\bm{a}),\bm{a})}}\,.
\end{equation}
The associated solution to the open WDVV equations that extends the Frobenius manifold structure on the parameter space of the miniversal unfolding is, therefore, of the form:
\begin{equation}\label{eq:openWDVVsolutionADE}
    \Omega_W(x,\bm{v})=F_W(x,\bm{a}(\bm{v}))+\varpi(\bm{v})\,,
\end{equation}
with $F_W'(x,\bm{a}(\bm{v}))=f(z(x),w(z(x),\bm{a}(\bm{v})),\bm{a}(\bm{v}))\,$. \\

In particular, $\Omega_W'$ depends on the flat coordinates $v_1,\dots, v_\ell$ of the Saito metric \cref{eq:metricsingularities} through the parameters of the deformation $a_1,\dots, a_\ell\,$. This property was already identified for the types $A$ and $D$ in \cites{BuryakArsingularity, BuryakOpenAD}, where the solutions were constructed with an independent method. We have here proven that it is also the case in type $E$, and therefore for all Saito Frobenius manifolds. In fact, the solutions associated to the remaining irreducible Coxeter groups can be determined from the $ADE$ ones as follows:
\begin{equation}\label{eq:foldingsADE}
\begin{aligned}
        \Omega_{B_\ell}(x,v_1,\dots, v_\ell)&=\Omega_{A_{2\ell-1}}(x,v_1,0,v_2,0,\dots, v_{\ell-1},0,v_\ell)\,, &&&& \ell\in\ZZ_{\geq 2}\,,\\
        \Omega_{I_2(\ell)}(x,v_1,v_2)&=\Omega_{A_{\ell-1}}(x,v_1,0,\dots, 0,v_2)\,,&&&&\ell\in\ZZ_{\geq 3}\,,\\
        \Omega_{F_4}(x,v_1,v_2,v_3,v_4)&=\Omega_{E_6}(x,v_1,0,v_2,v_3,0,v_4)\,,\\
                \Omega_{G_2}(x,v_1,v_2)&=\Omega_{D_4}(x,v_1,0,0,v_2)\,,\\
                \Omega_{H_3}(x,v_1,v_2,v_3)&=\Omega_{D_6}(x,v_1,0,v_2,0,v_3,iv_2)\,,\\
        \Omega_{H_4}(x,v_1,v_2,v_3,v_4)&=\Omega_{E_8}(x,v_1,0,v_2,0,0,v_3,0,v_4)\,.
\end{aligned}
\end{equation}
This is an obvious consequence of the corresponding rules for the Frobenius manifold prepotentials given in \cite{ZuberotherCoxeter}, as already noted in \cite{BuryakOpenAD}. In particular, the solutions for the families $B_\ell$ and $C_\ell$ coincide, therefore we have omitted the latter.

In types $A$ and $D\,$, the solution $\Omega_W$ can be computed explicitly, since the curves $\UC_f(\bm{a})$ all have genus zero. They correspond to known solutions in the literature, which have been computed with a variety of different, independent methods. The power of this approach lies both in its simplicity to get to the solutions and into the fact that it is systematic, i.e. easily generalisable to the exceptional Weyl group, even though the solution can only be written implicitly.  
\begin{ex}[Type-$A$] \label{ex:Alrankone}For $W=A_\ell\,$, then one can simply take:
\begin{equation}
    \begin{aligned}
        \lambda_{A_\ell}(x)&=x^{\ell+1}+a_1\,x^{\ell-1}+\dots+a_\ell\,,&&& \omega_{A_\ell}&=\dd{x}\,.
    \end{aligned}
\end{equation}
As a consequence, the corresponding solution to the open WDVV equations is:
\begin{equation}
    \Omega_{A_\ell}(x,\bm{v})=\tfrac{1}{\ell+2}x^{\ell+2}+\tfrac{1}{\ell}a_1(\bm{v})\,x^\ell+\dots+a_\ell(\bm{v})\,x+\varpi_\ell(\bm{v})\,.
\end{equation}
These are known polynomial solutions to the open WDVV equations.
    Some of these, for given values of $\ell\,$, appear as examples in \cites{alcoladophd,dealmeida2025openhurwitzflatf}. Furthermore, in the works \cites{BuryakArsingularity,IntSystemsAD} the general solution has been computed explicitly using mirror symmetry, which relates $\Omega_{A_\ell}$ to the generating series for the intersection numbers of Witten's $(\ell+1)$-spin classes on the moduli spaces of genus-zero stable curves with marked points \cite{mirrorsymmetryAspin}. In particular, $\varpi_\ell$ is generally not zero, and the coefficient of each monomial in the flat coordinates $\bm{v}$ of the constant of integration $\varpi_\ell $ is:
    \[
    [v_1^{k_1}\dots v_\ell^{k_\ell}]\, \varpi_\ell(\bm{v})=\tfrac{1}{\ell+1}\begin{cases}
        \tfrac{1}{k_1!\dots k_\ell !}(k_1+k_2+\dots+ k_\ell-2)!\,, &\sum_\alpha(\alpha+1)k_\alpha=\ell+2\,,\\
        0\,, & \text{ otherwise}\,.
    \end{cases}\
    \]
\end{ex}

\begin{ex}[Type-$D$]
    For $W=D_\ell\,$, as already mentioned, the Landau-Ginzburg superpotential and primary form are still defined on $\RS$. In particular,  $f_{ww}=2z$ so that:
\[    \begin{aligned}
        \lambda_{D_\ell}(z)&=z^{\ell-1}+a_1\,z^{\ell-2}+\dots+a_{\ell-2}\,z+a_{\ell-1}-\tfrac{1}{4z}\,a_\ell^2\,,&&& \omega_{D_\ell}&=\frac{\dd{z}}{\sqrt{2z}}\,.
    \end{aligned}
\]
The ODE in \cref{eq:flatcoordinateextensionsDE} giving the flat coordinate on the fibre easily integrates to $x^2:=2z\,$. Therefore, the superpotential in the coordinate $x$ is: 
\begin{equation}
     \lambda_{D_\ell}(x)=\tfrac{1}{2^{\ell-1}}\,x^{2(\ell-1)}+\tfrac{1}{2^{\ell-2}}\,a_1\,x^{2(\ell-2)}+\dots+a_{\ell-1}-\tfrac{1}{2x^2}\,a_\ell^2\,.
\end{equation}
Hence, there will be a unique rank-one extension whose extended prepotential looks like:
\begin{equation}\label{eq:openWDVVDl}
    \Omega_{D_\ell}(x,\bm{v}):=\tfrac{1}{2^{\ell-1}\,(2\ell-1)}\,x^{2\ell-1}+\tfrac{1}{2^{\ell-2}\,(2\ell-3)}\,a_1(\bm{v})\,x^{2\ell-3}+\dots+a_{\ell-1}(\bm{v})\,x+\tfrac{1}{2x}\,a_\ell(\bm{v})^2+\varpi(\bm{v})\,,
\end{equation}
These are also known solutions, as they correspond to the ones in \cite[Theorem 5.2]{BuryakOpenAD}, up to a rescaling of the coordinates. In particular, here it is $\varpi=0\,$.
\end{ex}
The open WDVV solutions coming from open Saito theory in types $A$ and $D$ were also constructed as reductions to the two open WDVV solutions constructed on some classes of infinite-dimensional Frobenius manifolds in \cite{Ma25}. 
\begin{ex}[$E_6$]
    For the exceptional group $E_6\,$, we have $f_w=3w^2+a_1z^2+a_2z+a_4\,$. The flat coordinate $x$ is, then, found by solving \cref{eq:flatcoordinateextensionsDE}, where $f_{ww}=6w\,$. This amounts to solving the following linear, non-homogeneous PDE in two variables for any fixed point in the parameter space:
    \[
    6w\pdv{x}{z}-(2a_1z+a_2)\pdv{x}{w}=\sqrt{6w}\,.
    \]
    In the variable $x\,$, the superpotential will, therefore, be:
    \begin{equation}\label{eq:superpotentialE6}
    \begin{aligned}
                \lambda(x)&=z(x)^4+w(x)^3+a_1\,z(x)^2w(x)+a_2\,z(x)w(x)+a_3\,z(x)^2+a_4\,w(x)+a_5\,z(x)+a_6\\
        &=z(x)^4-2w(x)^3+a_3\,z(x)^2+a_5\,z(x)+a_6\,.
    \end{aligned}
    \end{equation}
    The solution $\Omega_{E_6}$ is an antiderivative of $\lambda$ with respect to $x\,$. In particular, it is manifest that, up to the constant of integration $\varpi(\bm{v})\,$, it will depend on the flat coordinates $\bm{v}$ on the Saito Frobenius manifold through the functions $a_1(\bm{v}),\dots, a_6(\bm{v})\,$. This solution, and the corresponding ones for the other two Weyl groups $E_7$ and $E_8\,$, have, to our knowledge, not appeared in the literature before. In particular, existence itself of these solutions was not known. These are, of course, implicit expressions for the solutions. An explicit algebraic expression might be achieved using the expression for the superpotential in \cite{TY97}, although perhaps of difficult application to any actual example.
\end{ex}

    The eventual identity $\EEE_W$ in all these cases can also be easily computed using the grading in $\CC[z,w,\bm{a}]$ introduced by \cref{eq:homogeneitydeformation}. In fact, according to \cref{eq:flatcoordinateextensionsDE}, the degree of the flat variable $x$ is:
\[
\deg x=\deg z-\tfrac12\deg f_{ww}=r_1+r_2-\tfrac12 h=1\,,
\]
where we used the fact that, in all cases, $r_1+r_2=\tfrac12 h+1\,$, \cite{Dub04}. As a consequence, the superpotential always satisfies:
\begin{equation}
    \Lie_{E_W}\lambda=\lambda-\tfrac1h \,x\,\lambda_x=\lambda-\tfrac12(1-d)\,x\,\lambda_x\,,
\end{equation}
 where $d=1-\tfrac{2}{h}$ is the conformal dimension of the Saito Frobenius manifold. This proves the following:
 \begin{prop}\label{prop:EveIdentityADE}
       The eventual identity $\EEE_W$ that realises Dubrovin duality on $\UC_f$ is given by:
     \begin{equation}
    \EEE_W=E_W+\tfrac{1}{h}\,x\,\partial_x\,.
\end{equation}
 \end{prop}

As a consequence, it follows from \cref{prop:homogeneityE} that the solution $\Omega_W$ is weighted-homogeneous of overall degree $\tfrac12(3-d)=1+\tfrac{1}{h}$ with respect to $\EEE_W\,$. This proves that for all irreducible Coxeter group there exists a weighted-homogeneous solution to the open WDVV equations with degrees given as above. As we have already mentioned, this is the weighted-homogeneity condition considered in \cite{BuryakOpenAD}, where it was proven that there is no weighted-homogeneous polynomial solution for Coxeter groups other than the $A_\ell$, $B_\ell$ and $I_2(\ell)$ families. The reason for this is now clear: these solutions all come from anti-differentiation of the miniversal unfolding of a simple singularity. However, the operation of restricting such polynomials in two complex variable to a smooth curve in $\CC^2$ breaks the polynomiality, unless the singularity is of type $A$. We can, therefore, sum up the results we have discussed so far in this section in the following:
\begin{prop}\label{prop:ADErankone}
    Let $W$ be a finite, irreducible Coxeter group, and consider the associated polynomial solution $\Fr_W$ to the closed WDVV equations as described in \cite{Duborbitspaces}. There exists a unique solution $\Omega_W$ to the associated open WDVV equations that is weighted-homogeneous of overall degree $1+\tfrac{1}{h}$ with respect to $\EEE_W$ as defined in \cref{prop:EveIdentityADE}. In particular, $\Omega_W$ is given by \cref{eq:openWDVVsolutionADE} in the simply-laced case, and by \cref{eq:foldingsADE} when $W$ is not simply-laced.
\end{prop}
\begin{oss}
    Notice that, in type $A$, the seemingly-bizarre condition on the monomials $v_1^{k_1}\dots v_\ell^{k_1}$ that contribute to $\Omega_{A_\ell}$ is actually making sure that their degree of homogeneity is the correct one.
\end{oss}

\subsection{Rank-two extensions}
Now, based on the previous discussion, it appears that the non-polynomiality of the solution $\Omega$ to the rank-one open WDVV equations comes from restriction to the curve. It is, then, natural to ask whether the Frobenius manifold structures on the parameter spaces of all the ADE singularities admit a \emph{polynomial} rank-two extension to $\CC^2\times\CC^\ell$ such that, when restricted to the sub-bundle $\UC_f$, gives the known rank-one extension.

To this end, we firstly notice that \cref{lem:auxiliaryextension} generalises to arbitrary rank:
\begin{lem}\label{FM:lem:higherrankauxiliary}
Let $A$ be an $n$-dimensional associative, commutative complex algebra, and let $\widetilde{A}$ be an $(n+r)$-dimensional commutative complex algebra fitting in the following short exact sequence of the underlying vector spaces:
\[ \begin{tikzcd} 
0\arrow[r]&I\arrow[r]&\widetilde{A}\arrow[r,"\pi"]&A\arrow[r]&0\,. 
\end{tikzcd} \] 
Assume that the surjection $\pi$ is a homomorphism of algebras. Suppose further that the self-annihilator of $I$ is trivial, i.e.: \[ \operatorname{Ann}_I(I):=\{b\in I:\quad bx=0 \text{ for all } x\in I\}=0\,. \] 
In other words, the regular representation\footnote{Notice that, in general, $I$ does not contain the identity, as otherwise it would be equal to $\widetilde{A}\,$. That is why the associativity equations for $\dim I=1,2$ are already non-trivial. Hence, this is the regular representation of a non-unital ring.} of $I$ is faithful.

If the multiplication on $\widetilde{A}$ is associative on all triples containing at least one element of $I\,$, then it is associative. In such a case, we shall say that $\widetilde{A}$ is a \textbf{rank-$r$ auxiliary extension} of $A$.
\end{lem}
\begin{oss} 
In \cref{lem:auxiliaryextension}, the condition $\Omega''\neq 0$ is equivalent to requiring that the one-dimensional ideal have trivial annihilator in itself. Indeed, if $I=\CC a$ and $a^2=\Omega''a$, then $\operatorname{Ann}_I(I)=0$ if and only if $\Omega''\neq0$. The condition \( \operatorname{Ann}_I(I)=0 \) therefore provides the appropriate higher-rank generalisation of the rank-one auxiliary condition. 
\end{oss}
\begin{proof}
    Let $a_1,\dots, a_r$ be generators of the ideal $I$ and, for a fixed section $A\to\widetilde{A}\,$, let $e_1,\dots, e_n$ be the image of a system of generators of $A\,$. By definition, the product in $\widetilde{A}$ will be given, in such a basis, by:
    \[
    \begin{aligned}
        e_\alpha e_\beta&=\sum_{\mu=1}^nc\indices{^\mu_{\alpha\beta}}e_\mu+\sum_{p=1}^rF\indices{^p_{\alpha\beta}}\,a_p\,,\\
        e_\alpha a_p&=\sum_{q=1}^rF\indices{^q_{\alpha p}}\,a_q\,,\\
        a_pa_q&=\sum_{s=1}^rF\indices{^s_{pq}}\,a_s\,.
    \end{aligned}
    \]
    Here, $c\indices{^\mu_{\alpha\beta}}$ are the structure constants for the multiplication on $A$ in the given basis.
    
    By assumption, the product is associative on $\widetilde{A}\otimes\widetilde{A}\otimes I\,$, i.e. we have associativity in $I$ and on each triple $e_\alpha,e_\beta, a_p$ and $e_\alpha, a_p,a_q\,$. In terms of the structure constants, this means that they satisfy:
    \[
    \begin{aligned}
\sum_{\mu=1}^nc\indices{^\mu_{\alpha\beta}}F\indices{^p_{\mu q}}+\sum_{s=1}^rF\indices{^s_{\alpha\beta}}F\indices{^p_{qs}}&=\sum_{s=1}^rF\indices{^s_{\beta q}}F\indices{^p_{\alpha s}}\,,\\
\sum_{k=1}^rF\indices{^k_{pq}}F\indices{^s_{\alpha k}}&=\sum_{k=1}^rF\indices{^k_{\alpha p}}F\indices{^s_{kq}}\,,\\
\sum_{k=1}^rF\indices{^k_{qs}}F\indices{^t_{kp}}&=\sum_{k=1}^rF\indices{^k_{pq}}F\indices{^t_{ks}}\,.
    \end{aligned}
    \]
    What remains to be checked is, therefore, that $(e_\alpha e_\beta)e_\gamma=e_\alpha(e_\beta e_\gamma)\,$. Since we are assuming associativity on $A$, the difference $D_{\alpha\beta\gamma}$ between the two sides of the previous equation lies in the kernel of $\pi\,$, i.e. in the ideal $I\,$. Because the self-annihilator of $I$ is trivial, it is enough to prove that \[ D_{\alpha\beta\gamma}a_t=0 \] for every $t=1,\dots,r$.

    In coordinates, $D_{\alpha\beta\gamma}=0$ is the following condition on the structure constants:
    \[
    \begin{aligned}
\sum_{\mu=1}^nc\indices{^\mu_{\alpha\beta}}F\indices{^p_{\mu\gamma}}+\sum_{q=1}^rF\indices{^q_{\alpha\beta}}F\indices{^p_{\gamma q}}&=\sum_{\mu=1}^nc\indices{^\mu_{\beta\gamma}}F\indices{^p_{\mu\alpha}}+\sum_{q=1}^rF\indices{^q_{\beta\gamma}}F\indices{^p_{\alpha q}}\,.
    \end{aligned}
    \]
Then, $D_{\alpha\beta\gamma}a_t=0$ for any $t$ amounts to multiplying these conditions by $F\indices{^s_{pt}}$ and summing over $p\,$. Using associativity involving at least one element of $I\,$, we get, for instance on the terms on on the left-hand side:
    \[
    \begin{aligned}
\sum_{\mu=1}^n\sum_{p=1}^rc\indices{^\mu_{\alpha\beta}}F\indices{^p_{\mu\gamma}}F\indices{^s_{pt}}&=\sum_{\mu=1}^nc\indices{^\mu_{\alpha\beta}}\biggl(\sum_{p=1}^rF\indices{^p_{\gamma t}}F\indices{^s_{p\mu}}-\sum_{\nu=1}^nc\indices{^\nu_{\mu\gamma}}F\indices{^s_{\nu t}}\biggr)\\
&=\sum_{p=1}^r\sum_{\mu=1}^nc\indices{^\mu_{\alpha\beta}}F\indices{^p_{\gamma t}}F\indices{^s_{\mu p}}-\sum_{\mu,\nu=1}^nc\indices{^\mu_{\alpha\beta}}c\indices{^\nu_{\mu\gamma}}\,F\indices{^s_{\nu t}}\,,\\
\sum_{p,q=1}^rF\indices{^q_{\alpha\beta}}F\indices{^p_{\gamma q}}F\indices{^s_{pt}}&=\sum_{p,q=1}^rF\indices{^q_{\alpha\beta}}F\indices{^p_{qt}}F\indices{^s_{\gamma p}}\\
&=\sum_{p,q=1}^rF\indices{^s_{\gamma p}}F\indices{^q_{\beta t}}F\indices{^p_{\alpha q}}-\sum_{p=1}^r\sum_{\mu=1}^nc\indices{^\mu_{\alpha\beta}}F\indices{^p_{\mu t}}F\indices{^s_{\gamma p}}\,.
    \end{aligned}
    \]
Since the product is also commutative, we have $e_\alpha(e_\beta a_p)=(e_\alpha e_\beta)a_p=e_\beta(e_\alpha a_p)\,$. In terms of the structure constants, this implies that $\sum_{s=1}^rF\indices{^s_{\beta q}}F\indices{^p_{\alpha s}}$ is equal to the same expression with $\alpha$ and $\beta$ swapped. Hence, the left-hand side of the remaining associativity condition boils down to:
\[
\begin{aligned}
\sum_{\mu,p}c\indices{^\mu_{\alpha\beta}}F\indices{^p_{\mu\gamma}}F\indices{^s_{pt}}+\sum_{p,q}F\indices{^q_{\alpha\beta}}F\indices{^p_{\gamma q}}F\indices{^s_{pt}}=\sum_{p,q}F\indices{^s_{\gamma p}}F\indices{^q_{\beta t}}F\indices{^p_{\alpha q}}-\sum_{\mu,\nu}c\indices{^\mu_{\alpha\beta}}c\indices{^\nu_{\mu\gamma}}\,F\indices{^s_{\nu t}}\,.
\end{aligned}
\]
An analogous computation on the right-hand side gives:
\[
\begin{aligned}
\sum_{\mu,p}c\indices{^\mu_{\beta\gamma}}F\indices{^p_{\mu\alpha}}F\indices{^s_{pt}}+\sum_{p,q}F\indices{^q_{\beta\gamma}}F\indices{^p_{\alpha q}}F\indices{^s_{pt}}&=\sum_{p,q}F\indices{^s_{\alpha p}}F\indices{^q_{\beta t}}F\indices{^p_{\gamma q}}-\sum_{\mu,\nu}c\indices{^\mu_{\beta\gamma}}c\indices{^\nu_{\mu\alpha}}F\indices{^s_{\nu t}}\,.
\end{aligned}
\]
Using again commutativity of the product in $\widetilde{A}$ and associativity in $A$ gives the statement.
\end{proof}
\begin{oss} In other words, when the ideal has trivial self-annihilator, \( \operatorname{Ann}_I(I)=0\,, \) the first set of open WDVV equations in \cref{eq:openWDVVrankr} is redundant. This is the higher-rank analogue of the rank-one auxiliary condition $\Omega''\neq0$. \end{oss}
\begin{oss}
On the other hand, the higher-rank analogue of the condition $\Omega''=0$ is the square-zero condition \[ I^2=0\,. \] In this case, the ideal $I$ has trivial internal multiplication and is instead acted upon by the quotient algebra $A\,$. This corresponds to what is called an \emph{extension by a module} in the rank-one case studied in \cite{alcoladophd}. In adapted flat coordinates, the condition $I^2=0$ is equivalent to \[ \pdv[2]{\Omega_c}{x_a}{x_b}=0\,, \qquad a,b,c=1,\dots,r\,. \] Hence, the components of the extended prepotential are affine in the fibre variables. 
\end{oss}
\begin{oss}
In general, triviality of the self-annihilator implies that $I^2\neq 0\,$, provided that $I\neq 0\,$. The converse, however, fails in higher rank: $I^2\neq0$ only says that the multiplication on $I$ is not identically zero, whereas $\operatorname{Ann}_I(I)=0$ requires this multiplication to be faithful. These two conditions are equivalent when $\dim I=1$, which explains why the distinction is invisible in the rank-one case.

In higher rank there is also an intermediate possibility, namely \[ I^2\neq0\,, \qquad \operatorname{Ann}_I(I)\neq0\,. \] In this case, the fibre ideal has non-trivial multiplication, but this multiplication is not faithful. The associativity equations involving at least one element of $I$ then imply only that the remaining associativity defect lies in $\operatorname{Ann}_I(I)$, and the first set of open WDVV equations may still contain non-trivial components in the annihilator.

Consider, for instance, the following three examples of rank-two extensions of a one-dimensional unital algebra, with $I=\CC a\oplus\CC b$ as a vector space, and $A=\CC e\,$. We assume for simplicity that the coordinates $(t,x,y)$ corresponding to the basis $\{e,a,b\}$ of $A\oplus I$ are flat. Also, we identify the tangent spaces to $A\oplus I$ with the vector space itself, i.e. $\partial_t\equiv e\,$, $\partial_x\equiv a$ and $\partial_y\equiv b\,$.
\begin{enumerate}
    \item The extension-by-a-module case corresponds to $a^2=b^2=ab=0\,$. Here, the self-annihilator corresponds to the full ideal $\operatorname{Ann}_I(I)=I\,$. The extended prepotential is clearly given by:
    \[
    \Omega_{\mathrm{module}}=t\bigl(x\, a+y\, b\bigr)\,.
    \]
    In particular, it is linear in the fibre coordinates.
    \item Consider instead the multiplication
    \[
    \begin{aligned}
        a^2&=a\,,&&& b^2&=b\,,&&& ab&=0\,,
    \end{aligned}
    \]
    Then, clearly the self-annihilator is trivial, and this produces a rank-two auxiliary extension. The multiplication is associative since $I\cong \CC\oplus\CC$ as an algebra. The extended prepotential is given by:
    \[
    \Omega_{\mathrm{aux.}}=\bigl(tx+\tfrac12 x^2\bigr)a+\bigl(ty+\tfrac12 y^2\bigr)b\,.
    \]
    \item Finally, if:
    \[
    \begin{aligned}
        a^2&=a\,,&&& b^2&=0\,,&&& ab&=0\,,
    \end{aligned}
    \]
    then:
    \[
    \begin{aligned}
        I^2&=\ideal{a}\,,&&& \operatorname{Ann}_I(I)&=\ideal{b}\,.
    \end{aligned}
    \]
    In particular, neither is trivial, so this is genuinely an intermediate case. It is easy to check that such multiplication is associative, so the extended prepotential is:
    \[
    \Omega_{\mathrm{inter.}}=\bigl(tx+\tfrac12 x^2\bigr)a+ty\, b\,,
    \]
    which is affine only in the fibre coordinate $y$ corresponding to the self-annihilator direction $\ideal{b}\,$.
\end{enumerate}
\end{oss}


A partial answer to the previous question is given by the following:
\begin{prop}
    Let $\ell\in\ZZ_{\geq 1} $ and consider the Frobenius manifold on the parameter space of the miniversal deformation of the $A_\ell$-singularity. Let $\bm{v}$ denote its Saito flat coordinates. Define the functions:
    \[
    \begin{aligned}
        \Phi(z,w,\bm{v})&:=\tfrac{1}{\ell+2}\,z^{\ell+2}+\tfrac{1}{\ell}a_1(\bm{v})\,z^\ell+\dots+a_{\ell}(\bm{v})\,z+\varpi_\ell(\bm{v})\,,\\
        \Psi(z,w,\bm{a})&:=w\bigl[z^{\ell+1}+a_1(\bm{v})\,z^{\ell-1}+\dots+a_\ell(\bm{v})\bigr]\,,
    \end{aligned}
    \]
    with $a_1,\dots, a_\ell,\varpi_\ell$ as in \cref{ex:Alrankone}. $\Phi\,\partial_z+\Psi\,\partial_w$ is the extended prepotential of a rank-two extension to $\CC^2\times\CC^\ell$ of the corresponding Saito Frobenius manifold of type-$A\,$, in the system of flat coordinates $(\bm{v},z,w)\,$. Furthermore:
    \begin{itemize}
            \item $(\Phi\,\partial_z+\Psi\,\partial_w)\lvert_{\UC_f}=\Omega\,\partial_x\,$, $\Omega$ being the prepotential of the known rank-one extension.
    \item It is a polynomial, auxiliary rank-two extension that is weighted-homogeneous with respect to:
    \[
    \widetilde{\EEE}_W:=E_W+\tfrac{1}{h}\bigl(r_1 \,z\,\partial_z+r_2\,w\,\partial_w\bigr)\,,
    \]
    with $r_1,r_2,h\in\mathbb{Q}_{\geq 0}$ as in \cref{eq:gradingADE} and overall weights:
    \[
    \begin{aligned}
        \deg\Phi&:=1+\tfrac{r_1}{h}\,,&&& \deg\Psi&:=1+\tfrac{r_2}{h}\,.
    \end{aligned}
    \]
    \end{itemize}
\end{prop}
\begin{proof}
    For the sake of simplicity, we denote $\pdv{z}=:Z$ and $\pdv{w}=:W\,$. The multiplication defined by $\Phi$ and $\Psi$ in the ideal $I$ generated by $Z$ and $W$ is, according to \cref{eq:multtableextension}:
    \[
    \begin{aligned}
        Z^2&:=f_z\,Z+wf_{zz}\,W\,,&&& ZW&:=f_z\,W\,,&&& W^2&=0\,,
    \end{aligned}
    \]
    where $f$ is the miniversal deformation as in \cref{eq:miniversaldeformationsADE}.    
    It is an easy computation to check associativity in $I\,$. 
    
    Similarly, we denote $e_\mu:=\pdv{v_\mu}\,$, so that:
    \[
    \begin{aligned}
        e_\mu Z&:=f_\mu\,Z+wf_{z\mu},W\,,&&& e_\mu W&:=f_\mu\,W\,,&&&\mu&=1,\dots, \ell\,;
    \end{aligned}
    \]
 Again, it is easy to check associativity on triples involving one of the generators $e_1,\dots, e_\ell\,$.

 Finally, the associativity equations coming from checking it on triples containing two generators of the Jacobian ring are:
 \[
 \begin{aligned}
     f_\alpha f_\beta &=c\indices{^\mu_{\alpha\beta}}f_\mu+\Phi_{\alpha\beta}f_z\,,\\
     wf_\alpha \,f_{z\beta}+wf_\beta f_{z\alpha}&=w\,c\indices{^\mu_{\alpha\beta}}f_{\mu z}+w\Phi_{\alpha\beta}f_{zz}+\Psi_{\alpha\beta}f_z\,.
 \end{aligned}
 \]
 Since $\Psi_{\alpha\beta}=wf_{\alpha\beta}\,$, the second equation is automatically satisfied when $w=0\,$, or divisible by $w$ otherwise. In the latter case, since $\partial_z\Phi_{\alpha\beta}=f_{\alpha\beta}\,$, it is the $z$-derivative of the first one. The first equation is, then, \cref{eq:openwdvvf} with $K\indices{^z_{\alpha\beta }}:=\Phi_{\alpha\beta}$ and $K\indices{^w_{\alpha\beta}}:=0\,$. Using the previous Lemma, then, this proves that $\Phi$ and $\Psi$ give a rank-two extension. The fact that, when restricted to $w=0\,$, the extended vector potential reduces to the one of the rank-one extension is obvious, as $\Psi$ vanishes on such a curve and $\Phi_z$ is $f$ when on $w=0\,$. It is important to notice that the fact that the extended flat coordinate $x$ of the rank-one extension coincides with $z$ is important here. Weighted-homogeneity is also easily checked. In particular, $\Phi$ and $\Psi$ are, respectively, weighted-homogeneous polynomials of degree $r_1+h$ and $r_2+h$ respectively when one grades $\CC[z,w,\bm{a}]$ with the degrees defined by \cref{eq:homogeneitydeformation}. 
\end{proof}
\begin{oss}
    This is not the unique such extension. For instance, the same $\Phi$ works with $\Psi\in\CC[w]$ being any weighted-homogeneous polynomial of the correct degree -- which, in this case, it means that it is a constant multiple of $w^3\,$. Another criterion to pick out a unique rank-two extension might therefore be necessary. 
\end{oss}
We, therefore, formulate the following:
\begin{conj}\label{conj:ADE}
    Let $W$ be a simply-laced, irreducible Coxeter group of rank $\ell\in\ZZ_{\geq 1}$. There exists a rank-two extension of the associated Saito Frobenius manifold to $\CC^2\times \CC^\ell$ such that:
    \begin{enumerate}
    \item Adapted flat coordinates on the total space are given by $(\bm{v},z,w)\,$, where $\bm{v}$ are Saito flat coordinates on the base space and $(z,w)$ are coordinates on $\CC^2$ such that the miniversal deformation of the corresponding simple singularity of type $W$ is put in the normal form as in \cref{eq:miniversaldeformationsADE}.
    \item It is polynomial and weighted-homogeneous with respect to the Euler vector field:
    \[
    \widetilde{\EEE}_W:=E_W+\tfrac{1}{h}\bigl(r_1 \,z\,\partial_z+r_2\,w\,\partial_w\bigr)\,.
    \]
    Equivalently, the component functions $\Phi$ and $\Psi$ of the associated extended prepotential $\Phi\,\partial_z+\Psi\,\partial_w$ in an adapted coordinate frame are homogeneous polynomials in $\CC[z,w,\bm{a}]$ of degree $r_1+h$ and $r_2+h$ respectively, with respect to the grading induced by \cref{eq:homogeneitydeformation} and \cref{eq:gradingADE}. 
    \item When restriced to the sub-bundle $\UC_f\hookrightarrow \CC^2\times\CC^\ell\,$, it coincides with the rank-one extension described in \cref{prop:ADErankone}.
    \end{enumerate}
\end{conj}
\begin{oss}
    If the conjecture is true, then the same result holds for any irreducible Coxeter group $W$ using the rules in \cref{eq:foldingsADE}.
\end{oss}
\begin{oss}
    If it exists, the rank two solutions for the exceptional groups $E_6\,,E_7\,,E_8$ would be easily computable from the corresponding versal deformations. In particular, no implicit change of variables like the one that leads to \cref{eq:superpotentialE6} would be needed. This would give explicit expressions for the rank-two open WDVV solutions associated to the exceptional Weyl groups. 
\end{oss}
\begin{oss}
    The rank-one extension is basically constructed by identifying the additional generator with $\lambda_x\,$. The na\"{\i}ve generalisation of identifying the two additional generators of the rank-two extension with $f_z$ and $f_w$ will clearly not work. For example, it is not the case even in type-$A\,$.
\end{oss}
\begin{oss}
    As an argument against Conjecture \ref{conj:ADE}, one could say that the case of a singularity of type $A$ might, in some sense, be too simple to grasp any general behaviour. In particular, only the $z$-derivative of the singularity has a non-simple zero at the origin, and therefore gives a non-trivial relation in the Jacobian ring. In Arnold's classification, one can always increase the dimension of the hypersurface by adding the square of the corresponding variable -- du Val singularities correspond to the three-variable case, for instance -- but this will not change the singular behaviour at zero, as it will only produce a simple critical point in the additional directions. Equivalently, this does not change the Jacobian ring of the singularity. However, vice versa, one cannot go any lower than dimension two for type $D$ and $E$ by removing quadratic terms (i.e. restricting to the zero set of the corresponding derivative), while this is possible in type $A\,$.
\end{oss}

\section{Examples of dual solutions}\label{sec:dualsolutions}

In this section, we provide a few explicit solutions to the open WDVV equations in the dual case for some Hurwitz Frobenius manifolds, as described in \cref{sec:dualityextension}. In particular, the examples in \cref{subs:SaitotypeA} fit into a wider class of solutions, which is described in the subsequent paper \cite{opencech}.
\subsection{Open Hurwitz Saito flat F-manifolds of type-$A$}\label{subs:SaitotypeA}
Consider the Hurwitz space $\Hw_{0\,;\,\ell}$, for $\ell\in\ZZ_{\geq 1}$, equipped with the primary form $\omega$ being the second-kind meromorphic differential on $\RS$ with one double pole at $\infty_0$. This is a primary form of the first type, as classified in \cite[Lecture 5]{Dubrovin1996}. The Hurwitz Frobenius manifold is known to be isomorphic to the Saito Frobenius manifold $\CC^\ell\,/\,W(A_\ell)$ \cite{Duborbitspaces}.

We are, here, going to give its almost-dual open WDVV solution in the sense discussed in \cref{sec:dualityextension}. The superpotential and almost dual prepotential are given, as a function of the flat coordinates $\bm{w}$ of the intersection form, in \cite{Duborbitspaces,Dub04}:
\begin{align}
\lambda (x,\bm{w})&=(x+\overline{w})(x-w_1)\dots(x-w_\ell)\,,\\         \Frst (\bm{w})&=\tfrac12\sum_{1\leq a<b\leq \ell}(w_a-w_b)^2\,\log(w_a-w_b)+\tfrac12\sum_{a=1}^\ell (w_a+\overline{w})^2\,\log(w_a+\overline{w})\,, 
\end{align}
where $\overline{w}:=w_1+\dots+w_\ell$. The coordinate $x$ on $\RS$ is here chosen so that $\infty_0=\infty$ and $\omega=\dd{x}$. According to \cref{thm::main}, there exists a unique extension of the almost-flat Riemannian F-manifold almost-dual to $\Hw_{0\,;\,\ell}$, whose extended prepotential is of the form:
\begin{equation}\label{eq:omegaAl}
    \Omega (x,\bm{w})=(x+\overline{w})\log(x+\overline{w})+\sum_{a=1}^\ell(x- w_a)\log(x-w_a)+\varpi (\bm{w})\,,
\end{equation}
    where the constant of integration $\varpi $ is fixed -- up to quadratic polynomials -- by requiring that $\Omega$ satisfies the last set of open WDVV equations \ref{eq:openWDVV}:
    \begin{equation}\label{eq:conditionomega}
           \Omega'' \Omega_{ab}=\Omega'_a\Omega'_b-g^{cd}\,\Frst_{abc}\Omega'_d\,.
    \end{equation}
Now, both sides of the previous expressions are rational functions in $x\in\RS$ with poles at the zeros of $\lambda$. In particular, the left-hand side reads:  \[
  \begin{aligned}
      \Omega ''\Omega _{ab}&=\biggl[\frac{1}{x+\overline{w}}+\sum_c\frac{1}{x-w_c}\biggr]\biggl[\frac{1}{x+\overline{w}}+\frac{\delta_{ab}}{x-w_a}+\pdv[2]{\varpi }{w_a}{w_b}\biggr]\\
      &=\frac{1}{(x+\overline{w})^2}+\frac{1}{x+\overline{w}}\biggl[\frac{\delta_{ab}}{x-w_a}+\sum_c\frac{1}{x-w_c}\biggr]+\frac{\delta_{ab}}{x-w_a}\sum_c\frac{1}{x-w_c}+\pdv[2]{\varpi }{w_a}{w_b}\Omega'' \,.\end{aligned}\]
Hence, $\Omega''\Omega_{ab}$ has a double pole at $-\overline{w}$ and at $w_a$ when $a=b$, and simple poles at the other zeros of $\lambda$. As for the right-hand side, since only $x$-derivatives of $\Omega$ appear, it will not depend on $\varpi$. Furthermore, given that the structure constants are functions of the flat coordinates only, the double poles on the right-hand side can only come from the term $\Omega_a'\Omega_b'$:
\[
\begin{aligned}
    \Omega'_a\Omega'_b&=\frac{1}{(x+\overline{w})^2}-\frac{1}{x+\overline{w}}\biggl[\frac{1}{x-w_a}+\frac{1}{x-w_b}\biggr]+\frac{1}{(x-w_a)(x-w_b)}\,.
\end{aligned}
\]
Hence, as expected, the double poles cancel out and we are left with an expression having only simple poles at the zeros of $\lambda$ (and possibly at $\infty$):
\[
\begin{aligned}
    g^{cd}\Frst_{abc}\Omega'_d&=\Omega'_a\Omega'_b-\Omega''\Omega_{ab}\\
  \frac{1}{x+\overline{w}}\sum_dg^{cd}\Frst_{abc}-\sum_d\frac{g^{cd}\Frst_{abc}}{x-w_d}  &=\frac{1-\delta_{ab}}{(x-w_a)(x-w_b)}-\frac{1}{x+\overline{w}}\sum_{c\neq a,b}\frac{1}{x-w_a}+\\
  &\quad-\frac{\delta_{ab}}{x-w_a}\sum_{c\neq a}\frac{1}{x-w_c}-\pdv[2]{\varpi}{w_a}{w_b}\Omega''\,.
\end{aligned}
\]
Now, the rational function on the left-hand side of the previous equation and all the four summands in the right-hand side are holomorphic at $\infty$ except for the last one -- e.g. one easily checks that the residues at the zeros of $\lambda$ for each of them sum up to zero separately. It follows that we need to fix $\varpi$ so that $\varpi_{ab}\Omega''$ is also holomorphic at $\infty$. This clearly happens if and only if $\varpi_{ab}=0$ for any $a\leq b$, since $\Omega''$ has a simple pole at $\infty$ with residue $\Res_\infty\{\Omega''\dd{x}\}=\ell+1$. As a consequence, we can take $\varpi =0$ without loss of generality.

\subsection{Dubrovin-Zhang Frobenius manifolds of type-$A$}\label{ss:DZtypeA}
We consider, for fixed $\ell,r\in\ZZ_{\geq 1}$, the Frobenius manifold structure on $\Hw_{0\,;\,(\ell-1,r-1)}$ coming from the choice of the primary form $\omega$ being the third-kind Abelian differential on $\RS$ with poles and $\infty_0$ and $\infty_1$ and residues $\pm 1$ respectively. This is a primary differential of the third type. It is a well-known result that this Hurwitz Frobenius manifold is isomorphic to the one on the orbit space of the extended affine Weyl group $\widetilde{W}^{(r)}(A_{\ell+r-1})$, as defined in \cite{dubrovin_zhang_1998}.

The superpotential and almost-dual prepotential, as functions of flat coordinates $w_1,\dots, w_{\ell+r}$ of the intersection form, are given in \cites{dubrovin_zhang_1998,rileyeaw} respectively:
\begin{align}
        \lambda(x)&=e^{-rx}(e^x-e^{w_1})\dots(e^x-e^{w_{\ell+r}})\,,\\
    \begin{split}
        \Frst(\bm{w})&=-\tfrac{1}{2}\sum_{1\leq a< b\leq \ell+r}\bigl[\Li_3(e^{w_a-w_b})+\Li_3(e^{w_b-w_a})\bigr]+\tfrac{1}{12}\bigl(3+\ell-r-\tfrac{2}{r}\bigr)\bigl(w_1^3+\dots+w_{\ell+r}^3\bigr)+\\
    &\quad+\tfrac{1}{4}\bigl(1-\tfrac{2}{r}\bigr)\sum_{1\leq a< b\leq\ell+r}w_aw_b(w_a+w_b)-\tfrac{1}{r}\sum_{1\leq a<b<c\leq \ell+r}w_aw_bw_c\,.
    \end{split}
\end{align}
In this coordinate system, the primary form is $\omega=-\dd{x}$ and the poles are at $\infty_0=\infty$ and $\infty_1=0$.

According to \cref{thm::main}, there exists a unique rank-one extension of such an almost-flat F-manifold with extended prepotential:
\begin{equation}
    \Omega(x,\bm{w})=\tfrac{r}{2}
    x^2-x(w_1+\dots+w_{\ell+r})+\sum_{a=1}^{\ell+r}\Li_2(e^{x-w_a})+\varpi(\bm{w})\,,
\end{equation}
for some function $\varpi$ to be fixed in such a way that $\Omega$ satisfies \cref{eq:conditionomega}. To this end, we again consider the singularities of the derivatives of $\Omega$:
\[
\begin{aligned}
    \Omega''&=r+\sum_{a=1}^{\ell+r}\frac{1}{e^{w_a-x}-1}\,,\qquad&\Omega'_a&=-1-\frac{1}{e^{w_a-x}-1}\,,\qquad& \Omega_{ab}&=\varpi_{ab}+\frac{\delta_{ab}}{e^{w_a-x}-1}\,.
\end{aligned}
\]
Clearly, $\Omega''$ has a simple pole at $w_1,\dots, w_{\ell+r}$, with residue $-1$, and an essential singularity at $\infty$, with residue $\ell+r$. On the other hand, $\Omega'_a$ has a simple pole at $w_a$, with residue $+1$, and an essential singularity at $\infty$, with residue $-1$. Keeping this in mind, we compute the residue at $\infty$ of both sides of \cref{eq:conditionomega}:
\[
\begin{aligned}
    -\Res_\infty\{g^{cd}\Frst_{abc}\Omega'_d\dd{x}\}&=\Res_\infty\bigl\{\bigl(\Omega''\Omega_{ab}-\Omega'_a\Omega'_b\bigr)\dd{x}\bigr\}\\
    \sum_d g^{cd}\Frst_{abc}&=(\ell+r)\,\varpi_{ab}+\Res_\infty\biggl\{\biggl(\frac{r\,\delta_{ab}}{e^{w_a-x}-1}-\frac{1}{e^{w_a-x}-1}-\frac{1}{e^{w_b-x}-1}\biggr)\dd{x}\biggr\}+\\
    &\quad+\Res_\infty\biggl\{\biggl(\frac{\delta_{ab}}{e^{w_a-x}-1}\sum_{c\neq a}\frac{1}{e^{w_c-x}-1}-\frac{1-\delta_{ab}}{(e^{w_a-x}-1)(e^{w_b-x}-1)}\biggr)\dd{x}\biggr\}\\
    &=(\ell+r)\,\varpi_{ab}+r\,\delta_{ab}-2-\delta_{ab}(\ell+r-1)+(1-\delta_{ab})\\
    &=(\ell+r)\,\varpi_{ab}-\ell\,\delta_{ab}-1\,.
\end{aligned}
\]

The intersection form in the $w$-coordinates is easily computed using the residue formula. In particular, $g^{ab}=\tfrac{1}{\ell}-\delta_{ab}$. Therefore, independently of $c=1,\dots, \ell+r$, $\sum_dg^{cd}=\tfrac{\ell+r}{\ell}-1=\tfrac{r}{\ell}$. As a consequence, the previous equation becomes:
\[
(\ell+r)\,\varpi_{ab}-\ell\,\delta_{ab}-1=\tfrac{r}{\ell}\sum_c\Frst_{abc}\,.
\]

Finally, the third derivatives of $\Frst$ are as follows \cite{RileyPhD}:
\[
\begin{aligned}
    \Frst_{abc}&=-\tfrac{1}{r}\,,&&&a&\neq b \neq c\,;\\
    \Frst_{aab}&=-\tfrac{1}{r}-\frac{1}{e^{w_a-w_b}-1}\,,&&& a&\neq b\,;\\
    \Frst_{aaa}&=1+\ell-\tfrac{1}{r}+\sum_{b\neq a}\frac{1}{e^{w_a-w_b}-1}\,.
\end{aligned}
\]
Therefore, an easy computation gives:
\[
\sum_c\Frst_{abc}=-\tfrac{\ell}{r}+\ell\,\delta_{ab}\,.
\]
Plugging this back into the relations coming from comparing the residues at $\infty$ gives the simple system of PDEs $\varpi_{ab}=\delta_{ab}$ for $\varpi$. Hence, we can, without loss of generality, take:
\begin{equation}
    \varpi(\bm{w})=\tfrac12\bigl(w_1^2\dots+w_{\ell+r}^2\bigr)\,.
\end{equation}
The extended prepotential is, therefore:
\begin{equation}
    \Omega(x,\bm{w})=\tfrac12\bigl(r\,x^2+w_1^2+\dots+w_{\ell+r}^2\bigr)-x\,(w_1+\dots+w_{\ell+r})+\sum_{a=1}^{\ell+r}\Li_2(e^{x-w_a})\,.
\end{equation}

As for the eventual identity:
\begin{prop}
    The eventual identity $\EEE_{\mathrm{DZ}}$ on $\UC_{0\,;\,(\ell-1,r-1)}$ that gives the generalised Dubrovin duality at the universal-curve level is:
    \[
    \EEE_{\mathrm{DZ}}=E_{\mathrm{DZ}}+\tfrac{1}{\ell}\,\partial_x\,.
    \]

\end{prop}
\begin{proof}
It is easy to check that the Euler vector field $E$ in \cite[Eq. 3.3]{dubrovin_zhang_1998} satisfies $\Lie_E\lambda=\lambda-\tfrac{1}{\ell}\lambda'$ when applied to the corresponding superpotential. The statement then follows from \cref{prop:eventualidentityhere}. 
This is in accordance with the fact that here $d=1$ \cite[Theorem 2.1]{dubrovin_zhang_1998}.
\end{proof}


\subsection{Orbit spaces of type-$A$ Ma-Zuo extended affine Weyl groups}
We consider, for $\ell, r, k\in\ZZ_{\geq 1}$, the Frobenius manifold structure on $\Hw_{0\,;\,(\ell-1,r-1,k-1)}$ coming from the choice of the same primary differential $\omega$ as in the previous example. This is isomorphic to the structure defined on the space of orbits of the extended affine-Weyl group $\widetilde{W}^{(r,r+k)}(A_{\ell+r+k-1})$ in the work \cite{ma2023frobenius}. Flat coordinates for the intersection form are found in the same paper, whereas the corresponding almost-dual prepotential can be worked out by slightly adjusting the results in \cite[Chapter 4]{RileyPhD}:
\begin{align}
    \lambda(x)&=\frac{e^{-rx}}{(e^x-e^{u})^k}\,\bigl(e^x-e^{w_1}\bigr)\dots\bigl(e^x-e^{w_{\ell+r+k}}\bigr)\,,\\
\begin{split}
    \Frst(\bm{w},u)&=-\tfrac12\sum_{a<b}\bigl[\Li_3\bigl(e^{w_a-w_b}\bigr)+\Li_3\bigl(e^{w_b-w_a}\bigr)\bigr]+\tfrac{k}{2}\sum_{a=1}^{\ell+r+k}\bigl[\Li_3\bigl(e^{u-w_a}\bigr)+\Li_3\bigl(e^{w_a-u}\bigr)\bigr]+\\
    &\quad +\tfrac{1}{12}k\bigl(3k+r-\ell+\tfrac{2}{r}k^2\bigr)\,u^3+\tfrac{1}{12}\bigl(3-r+\ell-\tfrac{2}{r}\bigr)(w_1^3+\dots+w_{\ell+r+k}^3\bigr)+\\
    &\quad+\tfrac14\bigl(1-\tfrac{2}{r}\bigr)\sum_{ a<b}w_aw_b(w_a+w_b)-\tfrac{1}{4}k\bigl(1+\tfrac{2}{r}k\bigr)\,u^2\,(w_1+\dots+w_{\ell+r+k})+\\
    &\quad -\tfrac14 k\bigl(1-\tfrac{2}{r}\bigr)\,u\,\bigl(w_1^2+\dots+w_{\ell+r+k}^2\bigr)-\tfrac{1}{r}\sum_{a<b<c}w_aw_bw_c+\tfrac{k}{r}\,u\sum_{a<b}w_aw_b\,.
\end{split}    
\end{align}
The primary form in these coordinates is $\omega=-\dd{x}$, and the poles are at $\infty,0$ and $u\in\CC^*$.

The extended prepotential will, therefore, be of the following form:
\[
\Omega(x,\bm{w},u)=\tfrac12 rx^2-x\,\bigl(w_1+\dots+w_{\ell+r+k}-k\,u\bigr)-k\,\Li_2\bigl(e^{x-u}\bigr)+\sum_{a=1}^{\ell+r+k}\Li_2\bigl(e^{x-w_a}\bigr)+\varpi(\bm{w},u)\,.
\]
We fix $\varpi$ with computations similar to before. In particular, we use Latin indices for the $\bm{w}$-coordinates and Greek indices to also include the $u$-coordinate. For the structure constants, one has, again from \cite{RileyPhD}:
\[
\begin{aligned}
    \Frst_{abc}&=-\tfrac1r\,,&&& \Frst_{abu}&=\tfrac{k}{r}\,, &&& a\neq b\neq c\,;\\
    \Frst_{aab}&=-\tfrac1r-\frac{1}{e^{w_a-w_b}-1}\,,&&&&&&& a\neq b\,;\\ \Frst_{aau}&=\tfrac{k}{r}+\frac{k}{e^{w_a-u}-1}\,, &&&\Frst_{uua}&=-\tfrac{k^2}{r}+\frac{k}{e^{u-w_a}-1}\,,\\
    \Frst_{aaa}&=\ell+1-\tfrac1r +\sum_{b\neq a}\frac{1}{e^{w_a-w_b}-1}-\frac{k}{e^{w_a-u}-1}\,,\\
    \Frst_{uuu}&=\tfrac{k^3}{r}-k(\ell-k)-k\sum_a\frac{1}{e^{u-w_a}-1}\,.
\end{aligned}
\]
    On the other hand, the intersection form is given by $g^{a\mu }=\tfrac{1}{\ell}-\delta_{a\mu }$, $g^{u\alpha }=\tfrac{1}{\ell}+\tfrac{1}{k}\delta_{\alpha u}\,$.
Finally, derivatives of $\Omega$ are:
\[
\begin{aligned}
    \Omega''&=r-\frac{k}{e^{u-x}-1}+\sum_a\frac{1}{e^{w_a-x}-1}\,,\\
    \Omega'_a&=-1-\frac{1}{e^{w_a-x}-1}\,,&&& \Omega'_u&=k\,\biggl(1+\frac{1}{e^{u-x}-1}\biggr)\,,\\
    \Omega_{ab}&=\varpi_{ab}+\frac{\delta_{ab}}{e^{w_a-x}-1}\,,&&& \Omega_{au}&=\varpi_{au}\,,&&&\Omega_{uu}&=\varpi_{uu}-\frac{k}{e^{u-x}-1}\,.
\end{aligned}
\]
We, then, compute the residue at the essential singularity of both sides of \cref{eq:conditionomega}. We start from the $ab$-components:
\[
\begin{aligned}
  -  \Res_\infty\bigl\{g^{\mu\nu}\Frst_{ab\mu}\Omega'_\nu\,\dd{x}\bigr\}&=\Res_\infty\bigl\{(\Omega''\Omega_{ab}-\Omega'_a\Omega'_b)\dd{x}\bigr\}\\
  \bigl(\sum_c g^{c\mu }-k\, g^{u\mu}\bigr)\,\Frst_{ab \mu}&=(\ell+r)\,\varpi_{ab}-\ell\, \delta_{ab}-1\\
  \tfrac{r}{\ell}\sum_\mu \Frst_{ab\mu}&=(\ell+r)\,\varpi_{ab}-\ell\, \delta_{ab}-1\\
  -1+r\,\delta_{ab}&=(\ell+r)\varpi_{ab}-\ell\,\delta_{ab}-1\,.
\end{aligned}
\]
Similarly:
\[
\begin{aligned}
    - \Res_\infty\bigl\{g^{\mu\nu}\Frst_{au\mu}\Omega'_\nu\,\dd{x}\bigr\}&=\Res_\infty\bigl\{(\Omega''\Omega_{au}-\Omega'_a\Omega'_u)\dd{x}\bigr\}\\
    \tfrac{r}{\ell}\sum_\mu \Frst_{au\mu}&=(\ell+r)\varpi_{au}+k\\
    k&=(\ell+r)\varpi_{au}+k\,.\\
    - \Res_\infty\bigl\{g^{\mu\nu}\Frst_{uu\mu}\Omega'_\nu\,\dd{x}\bigr\}&=\Res_\infty\bigl\{\bigl(\Omega''\Omega_{uu}-(\Omega'_u)^2\bigr)\dd{x}\bigr\}\\
    \tfrac{r}{\ell}\sum_\mu \Frst_{uu\mu}&=(\ell+r)\varpi_{uu}+k(\ell-k)\\
    -k(r+k)&=(\ell+r)\varpi_{uu}+k(\ell-k)\,.
\end{aligned}
\]
Hence, $\varpi=\tfrac12 (w_1^2+\dots+w_{\ell+r+k}^2-k\,u^2)$, and the extended prepotential is:
\begin{equation}
\begin{aligned}
        \Omega(x,\bm{w},u)&=\tfrac12\bigl(r\,x^2+w_1^2+\dots+w_{\ell+r+k}^2-k\,u^2\bigr)-x\,\bigl(w_1+\dots+w_{\ell+r+k}-k\,u\bigr)+\\&\quad -k\,\Li_2\bigl(e^{x-u}\bigr)+\sum_{a=1}^{\ell+r+k}\Li_2\bigl(e^{x-w_a}\bigr)\,.
\end{aligned}
\end{equation}
\begin{prop}\label{prop:eventualidentityMZ}
    The eventual identity $\EEE_{\mathrm{MZ}}$ that gives Dubrovin almost-duality on the universal curve $\UC_{0\,;\,(\ell-1,r-1,k-1)}$ is given by:
    \[
    \EEE_{\mathrm{MZ}}=E_{\mathrm{MZ}}+\tfrac{1}{\ell}\,\partial_x\,.
    \]
\end{prop}
\begin{proof}
    According to \cite{ma2023frobenius}, $\Lie_{E_{\mathrm{MZ}}}\lambda=\lambda-\tfrac{1}{\ell}\lambda'$. The statement then follows from \cref{prop:eventualidentityhere}. 
    Again, here we have $d=1$ from \cite[Theorem 3.2]{ma2023frobenius}.
\end{proof}

More generally, one can consider, for $n,m,k_1,\dots, k_m\in \ZZ_{\geq 1}$ with $n>r+\sum_{a=1}^m k_a$, the Dubrovin almost-dual structure to the genus-zero Hurwitz Frobenius manifold with superpotential and primary form respectively given by:
\begin{equation}
    \begin{aligned}
        \lambda(x)&:=e^{-rx}\,\frac{(e^x-e^{w_1})\dots(e^x-e^{w_n})}{(e^x-e^{u_1})^{k_1}\dots(e^x-e^{u_m})^{k_m}}\,,&&&
        \omega&:=-\dd{x}\,.
    \end{aligned}
\end{equation}
Then, $(\bm{w},\bm{u})$ are flat coordinates of its intersection form -- this is an obvious generalisation e.g. of \cite[Lemma 4.1]{me} -- and the corresponding almost-dual solution to the open WDVV equations can be found again by adjusting the result of \cite[Theorem 4.12]{RileyPhD}. This will admit a unique rank-one extension to the universal curve with extended prepotential:
\begin{equation}
\begin{aligned}
        \Omega(x,\bm{w},\bm{u})&=\tfrac12\bigl(r\,x^2+w_1^2+\dots+w_n^2-k_1\,u_1^2-\dots- k_m\,u_m^2\bigr)+\sum_{a=1}^n\Li_2\bigl(e^{x-w_a}\bigr)+\\
        &\quad-x\bigl(w_1+\dots+w_n-k_1\,u_1-\dots-k_m\,u_m\bigr)-\sum_{\mu=1}^mk_\mu\,\Li_2\bigl(e^{x-u_\mu}\bigr)\,.
\end{aligned}
\end{equation}
The eventual identity will here be the same as in \cref{prop:eventualidentityMZ}, with $\ell=n-r-\sum_{a=1}^mk_a$.
However, the corresponding orbit space description is missing at present whenever $m>1$. Partial answers were provided in \cite{me}. Notice, finally, that these are the dual open WDVV solutions to the ones in \cite{Ma23} with trigonometric superpotential. In particular, their principal hierarchy is known. A similar expression can presumably be computed for the rational superpotentials.


\subsection{Orbit spaces of Jacobi groups of type-$A$}\label{ssec:Jacobigroups}
    We consider, for fixed $\ell\in\ZZ_{\geq 1}$, the genus-one Hurwitz space $\Hw_{1\,;\,\ell}$, and we endow it with the normalised holomorphic differential $\omega$ on the torus. This is a primary differential of the fifth type. The corresponding Hurwitz Frobenius manifold is known to be isomorphic to the structure on the orbit space of the Jacobi group $\mathbf{J}(A_{\ell})\,$, as described in \cites{bertolaJacobiI,bertolaJacobiII}.

    Flat coordinates $w_1,\dots, w_{\ell},u,\tau\in\CC^{\ell+1}\times\mathbb{H}$ for its intersection form, as well as the corresponding almost-dual prepotential, have been computed in the work \cite{rileygenusone}:
\begin{align}
    \lambda(x)&=\frac{e^{i2\pi u}}{\vartheta_1(x\,;\,\tau)^{\ell+1}}\vartheta_1(x+\overline{w}\,;\,\tau)\,\vartheta_1(x-w_1\,;\,\tau)\dots\,\vartheta_1(x-w_\ell\,;\,\tau)\,,\label{eq:superpotentialgenus1}\\
    \begin{split}
        \Frst(\bm{w},u,\tau)&=i2\pi\, u\bigl[\tfrac12 \pi^2\,\tau u-\overline{w}^2-w_1^2-\dots-w_\ell^2\bigr]+\\
        &\quad -\tfrac18\sum_{\bm{\alpha}\in\Delta(A_\ell)}\Phi_3\bigl(\inp{\bm{\alpha}}{-\overline{w}\,\bm{e}_0+\bm{w}}\,;\,\tau\bigr)\,+\\
        &\quad+\tfrac14 (\ell+1)\bigl[\Phi_3(-\overline{w}\,;\,\tau)+\Phi_3(w_1\,;\,\tau)+\dots+\Phi_3(w_\ell\,;\,\tau)\bigr]\,.
    \end{split}
\end{align}
The coordinate on the complex torus $\CC\,/\,(\ZZ+\tau\,\ZZ)$ is chosen so that $\infty_0=0$ and $\omega=\dd{x}$. Furthermore, in the previous equations, $\overline{w}:=w_1+\dots+w_\ell$, $\vartheta_1$ is Jacobi's first theta function:
\begin{equation}\label{eq:theta1}
\begin{aligned}
        \vartheta_1(x\,;\,\tau)&:=-i\sum_{n\in\ZZ}(-1)^n\,e^{i\pi(n+1/2)^2\,\tau}\,e^{i\pi (2n+1)\,x}\\
        &\equiv i(e^{-i\pi x}-e^{i\pi x})e^{i\pi\tau\,/\,4}\prod_{n\geq 1}(1-e^{i2 n\pi\tau})\bigl(1-e^{i2  \pi\,(n\tau+x) }\bigr)\bigl(1-e^{i2\pi (n\tau-x)}\bigr)\,,
\end{aligned}
\end{equation}
$\Delta(A_\ell)$ denotes the root system of $A_\ell$ and $\Phi_k(w\,;\,\tau)$ is the following combination of elliptic polylogarithms $\ELi_k$:
\begin{equation}
    \Phi_k(w\,;\,\tau):=\ELi_k(w\,;\, \tau)-\ELi_k(0\,;\,\tau)\,,\qquad k\in \ZZ\,.
\end{equation}
See \cites{ellipticpoly1,ellipticpoly2,ianellipticidentities} for further details. We are going to discuss a few identities involving elliptic polylogarithms in \cref{app:ellipticpolylog}. In particular, using \cref{eq:logtheta1}, it is easy to see that:
\[
(\log\lambda)(x)=i2\pi u+\tfrac{i}{2\pi}\biggl[\ELi_2'(x+\overline{w}\,;\,\tau)+\sum_{a=1}^\ell\ELi_2'(x-w_a\,;\,\tau)-(\ell+1)\ELi_2'(x\,;\,\tau)\biggr]\,.
\]
Therefore, there exists a unique rank-one extension of the almost-flat Riemannian F-manifold almost-dual to the Jacobi Frobenius manifold $\Hw_{1\,;\,\ell}$, and its extended prepotential is:
\begin{equation}
\begin{split}
        \Omega(x,\bm{w},u,\tau)&=i2\pi\,ux-\tfrac{i}{2\pi}(\ell+1)\,\ELi_2(x\,;\,\tau)+\tfrac{i}{2\pi}\ELi_2(x+\overline{w}\,;\,\tau)+\\&\quad+\tfrac{i}{2\pi}\sum_{a=1}^\ell\ELi_2(x-w_a\,;\,\tau)+\varpi(\bm{w},u,\tau)\,.
\end{split}
\end{equation}
In order to fix $\varpi$, we want $\Omega$ to satisfy the open WDVV equations \ref{eq:conditionomega}. In order to do that, we consider its second derivatives:
\[
\begin{aligned}
    \Omega''&=\frac{\vartheta_1'(x+\overline{w}\,;\,\tau)}{\vartheta_1(x+\overline{w}\,;\,\tau)}+\sum_{a=1}^\ell\frac{\vartheta_1'(x-w_a\,;\,\tau)}{\vartheta_1(x-w_a\,;\,\tau)}-(\ell+1)\frac{\vartheta_1'(x\,;\,\tau)}{\vartheta_1(x\,;\,\tau)}\,,\\
    \Omega'_a&=\frac{\vartheta_1'(x+\overline{w}\,;\,\tau)}{\vartheta_1(x+\overline{w}\,;\,\tau)}-\frac{\vartheta_1'(x-w_a\,;\,\tau)}{\vartheta_1(x-w_a\,;\,\tau)}\,,\\
    \Omega'_\tau&=\frac{1}{i4\pi}\biggl[\frac{\vartheta_1''(x+\overline{w}\,;\,\tau)}{\vartheta_1(x+\overline{w}\,;\,\tau)}+\sum_{a=1}^\ell\frac{\vartheta_1''(x-w_a\,;\,\tau)}{\vartheta_1(x-w_a\,;\,\tau)}-(\ell+1)\frac{\vartheta_1''(x\,;\,\tau)}{\vartheta_1(x\,;\,\tau)}\biggr]\,,\\
    \Omega'_u&=i2\pi\,,\\
    \Omega_{ab}&=\frac{\vartheta_1'(x+\overline{w}\,;\,\tau)}{\vartheta_1(x+\overline{w}\,;\,\tau)}+\delta_{ab}\frac{\vartheta_1'(x-w_a\,;\,\tau)}{\vartheta_1(x-w_a\,;\,\tau)}+\pdv[2]{\varpi}{w_a}{w_b}\,,& a,b&=1,\dots, \ell\,,\\
    \Omega_{a\tau}&=\frac{1}{i4\pi}\frac{\vartheta_1''(x+\overline{w}\,;\,\tau)}{\vartheta_1(x+\overline{w}\,;\,\tau)}-\frac{1}{i4\pi}\frac{\vartheta_1''(x-w_a\,;\,\tau)}{\vartheta_1(x-w_a\,;\,\tau)}+\pdv[2]{\varpi}{w_a}{\tau}\,, & a&=1,\dots, \ell\,,\\
        \Omega_{au}&=\varpi_{au}\,,& a&=1,\dots, \ell\,,\\
\Omega_{\tau u}&=\varpi_{\tau u}\,,\\
    \Omega_{uu}&=\varpi_{uu}\,,
\end{aligned}
\]
where $'$ denotes the derivative of $\vartheta_1$ with respect to its first argument and we used the heat equation $i4\pi \,\partial_\tau\vartheta_1=\vartheta''_1$ to replace derivatives with respect to $\tau$ with second derivatives with respect to $x$. We are going to deal with the second derivative $\Omega_{\tau\tau}$ separately, since it is more difficult to work out explicitly, but we shall see that there is actually no need to do so.

The intersection form in these coordinates is \cite{RileyPhD}:
\[
g^\sharp= G\,\oplus \mqty[0 & \pi^{-2}\\
\pi^{-2} & 0]\,,
\]
with $G_{ab}=\tfrac{1}{\ell+1}-\delta_{ab}$ for $a,b=1,\dots, \ell\,$.
Using the simple expressions for the third derivatives of $\Frst$ with respect to $u$ at least once \cite[Lemma 5.14]{RileyPhD}, it is easy to see that $\varpi$ does not depend on $u$, for the open WDVV equations where at least one index is equal to $u$ are automatically satisfied.

Next, we consider the $\tau\tau$ equation:
\[
\begin{aligned}
    \Omega''\Omega_{\tau\tau}&=(\Omega'_\tau)^2-\sum_{a,b}G_{ab}\Frst_{\tau\tau a}\,\Omega'_b-i\tfrac{2}{\pi}\Frst_{\tau\tau\tau}\,.
\end{aligned}
\]
Using the following Fourier-like series expansion of $\vartheta_1'/\vartheta_1$ around zero:
\[
\frac{\vartheta_1'(x\,;\,\tau)}{\vartheta_1(x\,;\,\tau)}=\cot x-4\sum_{n\geq 1}\frac{1}{1-e^{-i2\pi n \tau}}\sin (2nx)\,,
\]
one can show that $\Omega_\tau'$ is holomorphic on the torus. Therefore, the right-hand side of the previous equation only has poles at the zeros of $\lambda$, and they are all simple. On the other hand, $\Omega''$ additionally has a simple pole at zero. Therefore, in order for the identity to hold true, it needs to be $\Res_0\{\Omega''\Omega_{\tau\tau}\dd{x}\}=0$. Let us denote by $\Omega_0:=\Omega-\varpi$. Then, one has:
\[
\begin{aligned}
    \Omega_0&=i2\pi \,ux-\tfrac{i}{2\pi}(\ell+1)\,\ELi_2(x\,;\,\tau)+\tfrac{i}{2\pi}\ELi_2(x+\overline{w}\,;\,\tau)+\tfrac{i}{2\pi}\sum_{a=1}^\ell\ELi_2(x-w_a\,;\,\tau)\\
    &\equiv i2\pi \,ux+\tfrac{i}{2\pi}\sum_{n=0}^{\ell+1}k_n\ELi_2(x-z_n\,;\,\tau)\,,
\end{aligned}
\]
where $z_0=0$, $z_a=w_a$ for $a=1,\dots, \ell$ and $z_{\ell+1}=-\overline{w}$, with the obvious identifications of the coefficients $k_0,\dots, k_{\ell+1}$. In particular, these coefficients sum up to zero. As a consequence, according to \cref{eq:logtheta1} and the heat equation, we have:
\[
\begin{aligned}
    \tfrac{i}{2\pi}\sum_{n=0}^{\ell+1}k_n\,\partial_\tau^2\ELi_2'(x-z_n\,;\,\tau)&=\frac{1}{(i4\pi)^2}\sum_{n=0}^{\ell+1}k_n\biggl[\frac{\vartheta_1^{(4)}(x-z_n\,;\,\tau)}{\vartheta_1(x-z_n\,;\,\tau)}-\biggl(\frac{\vartheta_1''(x-z_n\,;\,\tau)}{\vartheta_1(x-z_n\,;\,\tau)}\biggr)^2\biggr]\,.
\end{aligned}
\]

Now, the only summand in the right-hand side that might be singular at zero is the zeroth one. However, it easy to see, using again the aforementioned Fourier-like series expansion for $\vartheta_1'/\vartheta_1$, that $\vartheta^{(4)}_1/\vartheta_1$ and $\vartheta_1''/\vartheta_1$ are both holomorphic at zero. In particular, it is, for some $\alpha,\beta\in\CC$:
\[
\begin{aligned}
    \frac{\vartheta_1^{(4)}}{\vartheta_1}-\biggl(\frac{\vartheta_1''}{\vartheta_1}\biggr)^2&=\biggl(\frac{\vartheta_1'}{\vartheta_1}\biggr)'''+4\,\frac{\vartheta'_1}{\vartheta_1}\biggl(\frac{\vartheta_1'}{\vartheta_1}\biggr)''+4\,\biggl(\frac{\vartheta_1'}{\vartheta_1}\biggr)^2\biggl(\frac{\vartheta_1'}{\vartheta_1}\biggr)'+2\,\biggl[\biggl(\frac{\vartheta_1'}{\vartheta_1}\biggr)'\biggr]^2=\alpha+\beta x^2+\order{x^4}\,.
\end{aligned}
\]
Hence, $(\Omega_0)_{\tau\tau}$ will also be holomorphic at zero. The residue at zero of $\Omega''\Omega_{\tau\tau}\omega$ is, therefore, simply given by:
\[
\begin{aligned}
0&=\Res_0\bigl\{\Omega''\Omega_{\tau\tau}\dd{x}\bigr\}=-(\ell+1)\bigl[\Omega_0(0)_{\tau\tau}+\varpi_{\tau\tau}\bigr]\,.
\end{aligned}
\]
As a consequence, it needs to be:
\[
\begin{aligned}
    \varpi(\bm{w},\tau)&=-\tfrac{i}{2\pi}\sum_{n=0}^{\ell+1}k_n\,\ELi_2(-z_n\,;\,\tau)+\tau\, f(\bm{w})+g(\bm{w})\\
    &=\tfrac{i}{2\pi}(\ell+1)\,\ELi_2(0\,;\,\tau)-\tfrac{i}{2\pi}\ELi_2(\overline{w}\,;\,\tau)-\tfrac{i}{2\pi}\sum_{a=1}^\ell\ELi_2(-w_a\,;\,\tau)+\tau\, f(\bm{w})+g(\bm{w})\,.
\end{aligned}
\]
Looking, similarly, at the residue at zero of the open WDVV equations for $a,b=1,\dots, \ell$ gives $f=g=0$. Hence, the extended prepotential is:
\begin{equation}
\begin{aligned}
        \Omega(x,\bm{w},u,\tau)&=i2\pi \,ux-\tfrac{i}{2\pi}(\ell+1)\bigl[\ELi_2(x\,;\,\tau)-\ELi_2(0\,;\,\tau)\bigr]+\tfrac{i}{2\pi}\bigl[\ELi_2(x+\overline{w}\,;\,\tau)-\ELi_2(\overline{w}\,;\,\tau)\bigr]+\\&\quad+\tfrac{i}{2\pi}\sum_{a=1}^\ell\bigl[\ELi_2(x-w_a\,;\,\tau)-\ELi_2(-w_a\,;\,\tau)\bigr]\,.
\end{aligned}
\end{equation}
\begin{oss}
    We recall that $\vartheta_1$ and its derivatives satisfy the following properties under period shifts \cite{thetavocabulary}:
\[
\begin{aligned}
    \vartheta_1(x+1\,;\,\tau)&=-\vartheta_1(x\,;\,\tau)\,,&&& \vartheta_1(x+\tau\,;\,\tau)&=-e^{-i\pi(2x+\tau)}\vartheta_1(x\,;\,\tau)\,,\\
     \vartheta_1^{(n)}(x+1\,;\,\tau)&=-\vartheta_1^{(n)}(x\,;\,\tau)\,,&&& \vartheta_1^{(n)}(x+\tau\,;\,\tau)&=-e^{-i\pi(2x+\tau)}\vartheta_1^{(n)}(x\,;\,\tau)+\,&&& n\in\ZZ_{\geq 1}\,.\\
     &&&&&\quad+(i2\pi)^n\,\vartheta_1(x+\tau\,;\,\tau)\,,\end{aligned}
\]
It follows that any linear combination of translates of $\vartheta_1^{(n)}/\vartheta_1$ for fixed $n$ with coefficients that sum up to zero is an elliptic function on the corresponding complex torus. This is the case for all the derivatives of $\Omega$ as listed above except for $\Omega_{ab}$. However, this does not affect our argument as $\Omega_{ab}$ is still well-defined on a cylinder -- and so is any other derivative -- but no computation relied on the fact that the derivative be defined on a compact surface.
\end{oss}
\begin{prop}
    The eventual identity $\EEE_{\mathbf{J}(A_\ell)}$ on the universal curve $\UC_{1\,;\,\ell}$ that realises the duality of the extensions as almost-flat F-manifolds is given by:
    \[
    \EEE_{\mathbf{J}(A_\ell)}=E_{\mathbf{J}(A_\ell)}\,.
    \]
\end{prop}
\begin{proof}
The Euler vector field is here given by $E=\tfrac{1}{i2\pi}\pdv{u}$, \cite[Eq. 3.3]{almediaJacobi}. Hence, it is straightforward to check that $\Lie_E\lambda=\lambda$. These structures also have charge one \cite{bertolaJacobiII}.
\end{proof}

\section{Conclusion and outlook}
In the present paper, we have argued that the existence of a rank-one extension of an F-manifold with compatible flat metric connection is equivalent to the existence of a Landau-Ginzburg model. This clearly applies to the case of a Hurwitz Frobenius manifold and, importantly, to its Dubrovin dual, generalising the work \cite{dealmeida2025openhurwitzflatf}. In this case, we have shown that the two extensions produce a bi-flat F-manifold on the universal curve. Firstly, this gives a natural, concrete way to construct bi-flat F-manifolds using the B-model of a Frobenius manifold. Secondly, it is a relatively simple method to construct solutions to the open associativity equations. Such constructions are rare in the literature. Some results and assumptions arising from open Gromow-Witten theory -- the setting in which open WDVV equations first appeared -- are recovered from a geometrical point of view.\\

We conclude this work with some open questions and directions for future research on this topic.
\begin{enumerate}
\item Computing the dual open WDVV solutions associated to Landau-Ginzburg models of Saito and Dubrovin-Zhang Frobenius manifolds of type $BCD$, as described in \cref{ss:typeDE} and \cite{DSZZ} respectively, would just be an easy exercise adjusting the computations in \cref{subs:SaitotypeA,ss:DZtypeA}. The only difficulty would lie in a parametrisation of the superpotentials by flat coordinates of the intersection form. This is the content of the subsequent publication \cite{opencech}, where we describe a root-theoretic procedure to construct dual-type solutions to the open WDVV equations. This method might even enable one to compute the dual solutions for the exceptional Weyl groups, where the Landau-Ginzburg superpotential is a family of meromorphic functions on curves of positive genus.

    \item Proving Conjecture \ref{conj:ADE}. In particular, if true, it would produce polynomial solutions to the open WDVV equations for exceptional Weyl groups that would be reasonably easy to construct. This would be a remarkable result, as such reflection groups are generally the hardest ones to deal with, and the associated solutions are usually very involved and any explicit expression requires considerable effort. Furthermore, such conjecture leads to natural generalisations. For instance, it is known that the Landau-Ginzburg superpotentials for the Dubrovin-Zhang Frobenius manifolds on orbit spaces of simply-laced extended affine-Weyl groups -- \cites{dubrovin_zhang_1998,DSZZ} -- can be seen as restrictions of miniversal deformations of affine cusp polynomials in $\CC^3$ to the intersection of two hypersurfaces \cites{THLaurentpoly,Rossimirror,IST19}. In particular, the tangent bundle is again isomorphic to the Jacobian ring of the miniversal deformation via the Kodaira-Spencer map, and the Euler vector field is the preimage of the unfolding. One could, therefore, ask whether there exist polynomial, weighted-homogeneous, auxiliary rank-three extensions that restrict to the known rank-one extension on the sub-bundle described by embedding the curve that gives the superpotential as a restriction of the three-variable polynomial in each fibre. In more generality, these are special cases of the construction of superpotentials in FJRW theory, therefore the same question can be asked in this setting.
    \item Bi-flat F-manifolds are known to be in one-to-one correspondence with weighted-homogeneous, non-symmetric solutions to the Darboux-Egoroff system \cite{AL13}. Given our construction of a bi-flat F-manifold on the universal curve, one could then study the corresponding solution to the Daroboux-Egoroff system, using the results of \cite{alcoladophd}. In particular, the connection potential in the fibre direction -- which determines the additional rotation coefficients -- will in this case be given by a function of the superpotential.
    \item Generalised notions of primary differentials have been considered in the literature, also in the case of moduli spaces of ramified coverings of $\RS$ with prescribed ramification over two points (the so-called \emph{double Hurwitz spaces}), \cites{Shr05,Shr05a,RomanoDoubleHurwitz}. In particular, they can produce structures with non-flat unities and/or without an Euler vector field. Nevertheless, in view of \cref{thm:extensiongeneral}, they still extend to almost-flat F-manifolds on the universal curve, and so do their Dubrovin duals, when they exist. In such cases, all these are bona fide \emph{almost}-flat F-manifolds, i.e. the identity is not flat. A more detailed description of what happens to the extension -- and therefore to the associated open WDVV solution -- in these cases, in particular in relation to the  explicit characterisation of such differentials given in \cite[Lemma 1 and Proposition 1]{RomanoDoubleHurwitz}, is an interesting direction for future research.
    \item From the A-model point of view, Frobenius manifolds correspond to the genus-zero part of Gromow-Witten theory. The higher-genus information can be recovered via an algorithmic procedure known as \emph{topological recursion} (TR) \cite{EO:2009}. In particular, to a solution to the closed WDVV equations one can associate a so-called \emph{descendent potential}, satisfying some \emph{topological recursion relations} (TRR). There is an open Gromow-Witten theory generalisation of this construction, providing \emph{open} TRR to associate an \emph{open} descendent potential to an open WDVV solution \cites{PSTInttheory,BCT1,buryakopen,BuryakArsingularity}. It is shown in \cite{BBvirasoro} that weighted-homogeneous solutions to the open WDVV equations -- in the sense of \cref{rem:quasihomogOmega} -- give rise to open descendent potentials satisfying open Virasoro constraints, which are the analogue of the ordinary Virasoro constraints \cite{DZ99} to the open WDVV equations setting. Since we now have a natural way to construct a weighted-homogeneous solutions to the open WDVV equations associated to any Hurwitz Frobenius manifold, then one could try and apply the open TRR formalism to these cases, thus providing results in open Gromow-Witten theory with a higher-genus target. Furthermore, given the geometrical understanding of the weighted-homogeneity condition for open WDVV solutions that appear to be meaningful in open Gromow-Witten theory, one could investigate the meaning of this condition in the enumerative geometry setting.
    \item In \cite{Ma25}, the open WDVV solutions for simple surface singularities of types $A$ and $D$ were constructed by embedding the superpotentials in an infinite-dimensional Frobenius manifold whose points are pairs of meromorphic functions respectively defined on the union of disjoint closed disks in $\RS$ and its complement. This approach allows for a direct study of the principal hierarchy of the corresponding flat F-manifold as a reduction of the one on the infinite dimensional counterpart, which had been studied in \cite{MWZ24}. An interesting feature of this method is that it predicts two open WDVV solutions in the infinite-dimensional setting, which happen to coincide on the finite-dimensional submanifolds considered. 
    On the other hand, such infinite-dimensional Frobenius manifolds seem to have been constructed only in the rational, genus zero case, and do not mention the dual structure. However, our method and the one in \cite{dealmeida2025openhurwitzflatf} predict the existence of the open WDVV solutions in every genus -- in fact, examples in genus one are provided. An intriguing direction for future research would, therefore, be to consider the dual structure in the infinite-dimensional setting, study the (presumably two) associated dual open WDVV solutions and relate them to the ones in \cite{Ma25}, and possibly to their finite-dimensional reductions.
\end{enumerate}
We hope to address these problems in future work.

\appendix
\section{Elliptic polylogarithm identities}\label{app:ellipticpolylog}
The aim of this appendix is only to fix the notation and discuss some identities we used in the computations leading to the open WDVV solution in \cref{ssec:Jacobigroups}. For details about the analytic theory of elliptic polylogarithms, the reader can refer to \cite{ellipticpoly1,ellipticpoly2,ianellipticidentities}. In particular, our definition differs slightly from the ones in those works mainly for computational and notational convenience. We define the $n$-th elliptic polylogarithm to be:
\begin{equation}
    \ELi_n(u\,;\,\tau):=\sum_{k\geq 0}\Li_n\bigl(e^{i2\pi( u+k\tau)}\bigr)+(-1)^{n-1}\sum_{k\geq 1}\Li_n\bigl(e^{-i2\pi (u-k\tau)}\bigr)-\chi_n(u\,;\,\tau)\,,\qquad n\in\ZZ\,,
\end{equation}
where:
\[
\chi_n(u\,;\,\tau):=(i2\pi)^n\sum_{k=0}^n\frac{B_{k+1}}{(n-k)!(k+1)!}u^{n-k}\,\tau^k+(-1)^{n-1}\frac{B_n}{2n!}\,,
\]
$B_k$ being the $k$-th Bernoulli number.

In particular, elliptic polylogarithms naturally appear when dealing with logarithms of $\vartheta_1$. For the sake of convenience, we rewrite \cref{eq:theta1} using Dedekind eta function:
\[
\vartheta_1(x\,;\,\tau)=i\,\eta(\tau)\, e^{-i\pi(x-\tau\,/\,6)}\bigl(1-e^{i2\pi x}\bigr)\prod_{n\geq 1}\bigl(1-e^{i2\pi (n\tau+x)}\bigr)\bigl(1-e^{i2\pi (n\tau-x)}\bigr)\,.
\]
As a consequence:
\[
\begin{aligned}
    \log\vartheta_1(x\,;\,\tau)&=i\tfrac{\pi}{2}-i\pi\bigl(x-\tfrac16\tau\bigr)+\log\eta(\tau)+\sum_{n\geq 0}\log(1-e^{i2\pi (x+n\tau)})+\sum_{n\geq 1}\log(1-e^{-i2\pi (x-n\tau)})\,\\
    &=\tfrac14+i\tfrac{\pi}{2}+\log\eta(\tau)-\ELi_1(x\,;\,\tau)\,.
\end{aligned}
\]
Now, clearly, the elliptic polylogarithms satisfy, under differentiation, properties similar to the ones of the standard Euler polylogarithms:
\[
\begin{aligned}
    \pdv{\ELi_n}{u}&=i2\pi \sum_{n\geq 0}\Li_{n-1}\bigl(e^{i2\pi(u+n\tau)}\bigr)-i2\pi(-1)^{n-1}\sum_{n\geq 1}\Li_{n-1}\bigl(e^{-i2\pi(u-n\tau)}\bigr)-\pdv{\chi_n}{u}\\
    &=i2\pi\, \ELi_{n-1}(u\,;\,\tau)+\tfrac{\pi}{(n-1)!}i^{2n+1} B_{n-1}\,.
\end{aligned}
\]
As a consequence:
\begin{align}\label{eq:logtheta1}
     \log\vartheta_1(x\,;\,\tau)&=i\tfrac{\pi}{2}+\log\eta(\tau)+\tfrac{i}{2\pi}\ELi'_2(x\,;\,\tau)\,.
\end{align}
In particular, since the additive constant does not depend on $x$, any linear combination of translates of $\log\vartheta_1$ with coefficients that sum up to zero contains only the polylogarithmic term:
\[
\sum_{n=1}^Nk_n\,\log\vartheta_1(x-z_n\,;\,\tau)=\tfrac{i}{2\pi}\sum_{n=1}^Nk_n\,\ELi'_2(x-z_n\,;\,\tau)\,,
\]
whenever $k_1+\dots+k_N=0\,$.

\begin{conflict}
   The authors have no conflicts of interest to declare that are relevant to the content of this
article.
\end{conflict}
\begin{datastatement}
   No data sets were used in this paper.
\end{datastatement}
\phantomsection
\addcontentsline{toc}{section}{References}
\bibliography{biblio} 
\end{document}